\documentclass[%
 reprint,
superscriptaddress,
amsmath,amssymb,
aps,
prd
]{revtex4-2}
\usepackage{subcaption}
\usepackage{ragged2e}

\usepackage{slashed}
\usepackage{xcolor}
\usepackage{amsthm}

\usepackage{diagbox}
\usepackage{slashbox}
\usepackage{array}

\usepackage{graphicx}
\usepackage{dcolumn}
\usepackage{bm}

\newtheorem*{definition}{Definition}
\newtheorem{theorem}{Theorem}
\newtheorem{lemma}[theorem]{Lemma}
\newtheorem{proposition}[theorem]{Proposition}
\newtheorem{corollary}{Corollary}[theorem]

\allowdisplaybreaks

\global\long\def\bfn{{\bf n}}
\global\long\def\bfJ{{\bf J}}

\global\long\def\mA{\mathcal{A}}%

\global\long\def\mC{\mathcal{C}}%

\global\long\def\mH{\mathcal{H}}%

\global\long\def\mL{\mathcal{L}}%
 
\global\long\def\mM{\mathcal{M}}%

\global\long\def\mX{\mathcal{X}}%

\global\long\def\choose#1#2{\begin{pmatrix}#1\\
#2
\end{pmatrix}}
 
\global\long\def\e{\epsilon}%
 
\global\long\def\ra{\rightarrow}%

\global\long\def\avg#1{\left\langle #1\right\rangle }%

\global\long\def\f#1#2{\frac{#1}{#2}}%
 
\global\long\def\del{\partial}%
 
\global\long\def\t{\theta}%
 
\global\long\def\a{\alpha}%
 
\global\long\def\b{\beta}%
 
\global\long\def\g{\gamma}%
 
\global\long\def\G{\Gamma}%
 
\global\long\def\s{\sigma}%
 
\global\long\def\r{\rho}%
 
\global\long\def\d{\delta}%

\global\long\def\tr{\text{tr}}%
 
\global\long\def\bra#1{\left\langle #1\right|}%
 
\global\long\def\ket#1{\left|#1\right\rangle }%
 
\global\long\def\N{\mathbb{N}}%
\global\long\def\I{\mathbb{I}}%
 
\global\long\def\Z{\mathbb{Z}}%
\global\long\def\E{\mathbb{E}}%
 
\global\long\def\R{\mathbb{R}}%
 
\global\long\def\C{\mathbb{C}}%

\global\long\def\w{\omega}%
 
\global\long\def\D{\Delta}%

\global\long\def\app{\approx}%
\global\long\def\vp{\varphi}%
\global\long\def\lam{\lambda}%
\global\long\def\i{\text{i}}%

 \newcommand{\be}{\begin{equation}}
\newcommand{\beq}{\begin{equation}}
 \newcommand{\ee}{\end{equation}}
 \newcommand{\bea}{\begin{align}}
 
 \newcommand{\eea}{\end{align}}

\def\nn{\nonumber}

\begin{document}

{ \flushright ~~~~~~~~~~~~~~~~~~~~~~~~~~~~~~~~~~~~~~~ USTC-ICTS/PCFT-24-43 }

\title{D-commuting SYK model: building quantum chaos from integrable blocks }

\author{Ping Gao}
\affiliation{Department of Physics and NHETC, Rutgers University, Piscataway, NJ 08854, USA}
\author{Han Lin}
\affiliation{Kavli Institute for Theoretical Sciences (KITS), University of Chinese Academy of Sciences, Beijing 100190, China}%
\author{Cheng Peng}
\affiliation{Kavli Institute for Theoretical Sciences (KITS), University of Chinese Academy of Sciences, Beijing 100190, China}
\affiliation{Peng Huanwu Center for Fundamental Theory, Hefei, Anhui 230026, China}

\date{\today}

\begin{abstract}
We construct a new family of quantum chaotic models by combining multiple copies of integrable commuting SYK models. As each copy of the commuting SYK model does not commute with others, this construction breaks the integrability of each commuting SYK and the family of models demonstrates the emergence of quantum chaos. We study the spectrum of this model analytically in the double-scaled limit. As the number of copies tends to infinity, the spectrum becomes compact and equivalent to the regular SYK model. For finite $d$ copies, the spectrum is close to the regular SYK model in UV but has an exponential tail $e^{E/T_c}$ in the IR. We identify the reciprocal of the exponent in the tail as a critical temperature $T_c$, above which the model should be quantum chaotic. $T_c$ monotonically decreases as $d$ increases, which expands the chaotic regime over the non-chaotic regime. We propose the existence of a new phase around $T_c$, and the dynamics should be very different in two phases. We further carry out numeric analysis at finite $d$, which supports our proposal. 

Given any finite dimensional local Hamiltonian, by decomposing it into $d$ groups, in which all terms in one group commute with each other but terms from different groups may not, our analysis can give an estimate of the critical temperature for quantum chaos based on the decomposition. We also comment on the implication of the critical temperature to future quantum simulations of quantum chaos and quantum gravity. 
\end{abstract}

\maketitle


\section{Introduction} \label{sec:intr}

One interesting question regarding holography is understanding how holographic duality could emerge in quantum systems. Since the birth of AdS/CFT correspondence, researchers have identified various models exhibiting holographic features. Beyond these explicit constructions, a universal microscopic mechanism underlying holography remains under exploration. In recent years, the SYK model~\cite{sachdev1993gapless,kitaev2015simple,Maldacena:2016hyu} has emerged as one of the simplest toy models possessing holographic properties. This quantum mechanical model describes $N$ Majorana fermions in (0+1) dimensions with a random $q$-local Hamiltonian. Owing to its simplified structure, it allows detailed examination of holographic emergence.

One signature of holography in the SYK model is quantum chaos~\cite{Shenker:2013pqa,Shenker:2014cwa,Roberts:2014isa,kitaev2015simple,Maldacena:2016hyu,Kobrin:2020xms}, which is dual to the boost symmetry of the AdS$_2$ black hole spacetime near the horizon. At low temperatures, the SYK model has an effective action described by a Schwarzian derivative~\cite{Maldacena:2016upp}, matching the boundary description of reparameterization modes~\cite{Maldacena:2016hyu} in Jackiw-Teitelboim gravity with negative cosmological constant~\cite{Teitelboim:1983ux,Jackiw:1984je}. The partition function derived from the Schwarzian derivative is 1-loop exact and yields the IR spectrum $\rho(E)\sim \sinh C\sqrt{E-E_0}$ for some $C>0$ \cite{Stanford:2017thb}. The $e^{C\sqrt{E-E_0}}$ scaling in the IR regime coincides with the spectrum of microscopic states of a black hole, suggesting a holographic interpretation. Moreover, a finite $E_0$ implies a lower bound for black hole energy. Since the Schwarzian derivative leads to exponential growth $e^{\lambda t}$ in out-of-time-ordered correlators (OTOCs) with $\lambda=1$~\cite{Maldacena:2016hyu}, this spectrum is strongly related to maximal quantum chaos\footnote{$\lambda$ represents the quantum Lyapunov exponent characterizing how fast the local information is scrambled over the system. A theorem \cite{Maldacena:2015waa} establishes an upper bound $\lambda\leq 1$ due to unitarity, with $\lambda=1$ corresponding to maximal chaos – a saturation condition achieved by gravitational systems.}.

 Another example of a bounded continuous spectrum in the ensemble average appears in random matrix theory (e.g., the GUE ensemble). Treating each random Hermitian matrix sample as as the Hamiltonian of some system, the system exhibits explicitly chaotic dynamics. Notably, a double-scaled matrix model has been shown dual to JT gravity~\cite{Saad:2019lba}, which is in spirit analogous to the duality between random matrices and string theories \cite{Ginsparg:1993is}.

Furthermore, SYK type models possess another appealing feature: potential experimental realizations in various platforms~\cite{Danshita:2016xbo,Garcia-Alvarez:2016wem,PhysRevLett.121.036403,Pikulin:2017mhj,Chew:2017xuo,Luo:2017bno,Brzezinska:2022mhj,Uhrich:2023ddx}. Intriguingly, a recent proposal~\cite{Jafferis:2022crx} suggests that certain SYK-like models implemented on quantum processors could exhibit signatures of traversable wormholes. However, the specific model proposed is "commuting" – all Hamiltonian terms mutually commute – resulting in numerous conserved quantities that preclude quantum chaos.

The commuting SYK model is further analyzed in~\cite{Gao:2023gta}. As a random yet integrable SYK variant, it displays an unbounded Gaussian-distributed spectrum. It has been demonstrated in~\cite{Gao:2023gta} that this integrable system does not exhibit chaotic behavior despite superficial similarities. Interestingly, this provides a new angle to explore the mechanism of quantum chaos (or perhaps holography) in the regular SYK model. Specifically, we can decompose the regular SYK Hamiltonian into multiple groups of mutually commuting terms (and is thus integrable), but different groups do not commute with each other. If we imagine the SYK Hamiltonian as the limit of adding these groups one by one, the chaotic nature emerges as we include huge numbers of these groups, and the integrability is dramatically violated.

To understand this process from integrability to chaos in more detail, we explicitly constructed $d$ copies of the commuting SYK Hamiltonians $\tilde H_a$ for $a=1,\cdots,d$. For each $\tilde H_a$, we used different commuting basis so that $\tilde H_a$ do not commute with each other. We may use a parameter $q=e^{-\lam}\in[0,1]$ to characterize strength of the non-commutativity between two $\tilde H_a$: if $q\ra 0$ the non-commutativity is extremely strong and $q\ra 1$ means $\tilde H_a$ all commute. The total Hamiltonian is given by
\be 
\tilde{H}=\f 1{\sqrt{d}}\sum_{a=1}^{d}\tilde{H}_{a}
\ee
Since each $\tilde H_a$ has an ensemble of random couplings, we are interested in the spectrum of $\tilde H$ under ensemble average
\be 
\r(E)=\E \left[\sum_{\tilde\lambda} \d(E-\tilde \lambda)\right]
\ee
where $\tilde \lambda$ is the eigenvalue of $\tilde H$ for each instantiation in the random ensemble. For convenience, let us call this model as ``d-commuting SYK model" (dcSYK). 

We show that the spectrum is non-compact unless $d\ra \infty$. In the latter case, the spectrum (in the double-scaled limit) is identical to the regular SYK model. In particular, there is an emergent finite lower bound of the spectrum $E_0$ as $d\ra\infty$. Our model gives a natural interpolation between the integrable $d=1$ commuting SYK model to an equivalently (in the spectrum sense) chaotic SYK model. 

Tentatively, we regard a compact spectrum with a square-root form of edge $\r(E)\sim \sqrt{E-E_0}$ as a characteristic of quantum chaos as it is the common feature of SYK model and random matrices. For the dcSYK model with finite $d$, the spectrum is close to the $d\ra\infty$ form above some energy scale and has an infinite exponential tail $e^{E/T_c}$ in the low energy regime. Equivalently, for the canonical ensemble of temperature $T$, there is a critical temperature $T_c$ so that we can approximately regard the spectrum as chaotic for $T>T_c$ and non-chaotic for $T<T_c$. We compute $T_c$ in two nontrivial cases and find that it is a monotonically decreasing function of $d$. This shows how the chaotic regime emerges and extends as we include more components $\tilde H_a$ in the Hamiltonian.

There is a recent work \cite{Berkooz:2024evs,Berkooz:2024ofm} constructing an interpolation between integrability and quantum chaos with a tunable linear combination of the commuting SYK model and the regular SYK model. They also find a critical temperature separating the integrable and chaotic regimes. Our model is probably a more natural and nontrivial interpolation and shows how quantum chaos emerges as $d$ increases. Moreover, our model has an operational significance in analyzing a generic random model. Following the aforementioned idea of decomposing the regular SYK model into integrable groups, this decomposition can be done for any given local random Hamiltonian. Minimizing the number of groups $d$ in the decomposition, we can compute an averaged non-commutativity parameter $\lam$, and use the dcSYK model to estimate the critical temperature $T_c$. This gives a rapid diagnosis of how chaotic a given model is. In this sense, we propose that our model is a prototype of a universality of random models, which is parameterized by $d$ and $\lam$.

The paper is organized as follows. We define the model and explain the chord diagram formalism in Sec. \ref{sec:setup}. Then we compute the exact spectrum in the $q=0$ case using free probability in Sec. \ref{sec:free}, and compute the approximate IR spectrum for $q\ra 1$ in the Schwarzian limit in Sec. \ref{sec:coarse}. In Sec. \ref{sec:5} we discuss the thermodynamical meaning of the critical temperature $T_c$. In Sec. \ref{sec:num} we use exact diagonalization to numerically compute the spectrum and spectral form factor of the dcSYK model. In Sec. \ref{sec:concl} we summarize the results and list a few future directions.

\section{Setup} \label{sec:setup}

Consider a system with $2N$ Majorana fermions $\psi_{i}$ for $i=1,\cdots,2N$
obeying $\{\psi_{i},\psi_{j}\}=2\d_{ij}$. We choose the representation
to be $2^{N}$-dimensional. The Hamiltonian is the sum over $d$ copies
of commuting SYK models \cite{Gao:2023gta}
\begin{equation}
\tilde{H}=\f 1{\sqrt{d}}\sum_{a=1}^{d}\tilde{H}_{a},~\tilde{H}_{a}=\sum_{i_{1}<\cdots<i_{p}}J_{i_{1}\cdots i_{p}}^{a}\mX_{i_{1}}^{a}\cdots\mX_{i_{p}}^{a} \label{1-dcsyk}
\end{equation}
where $\mX_{j}^{a}$ are bilinear in Majorana fermions and hermitian
\begin{equation}
\mX_{j}^{a}\equiv i \psi_{2j-1}\psi_{[(2j-4+2a)\mod (2N)]+2},~j=1,\cdots,N
\end{equation}

For each $a$, different $\mX_{i}^{a}$ commutes with each other and
thus each $\tilde{H}_{a}$ is a commuting SYK model. The commuting SYK model is indeed the $q$-local version \cite{derrida1980random,gardner1985spin} of the well-known Sherrington-Kirkpatrick (SK) model \cite{sherrington1975solvable,thouless1977solution,parisi1979infinite} for spin glass. It is an integrable model because we can simultaneously diagonalize $\mX_i^a$ for all $i$ inside $\tilde H_a$ due to their commuting nature. However, for
different $a$ they do not commute when $\mX_{i}^{a}$ and $\mX_{j}^{b}$
have one overlap of Majorana fermions. Therefore, we can regard the number of copies $d$ as a parameter deforming away from the integrability. The random coupling cofficient $J_{i_{1}\cdots i_{p}}^{a}$
has Gaussian distribution
\begin{equation}
\E[J_{i_{1}\cdots i_{p}}^{a}]=0,\quad\E[(J_{i_{1}\cdots i_{p}}^{a})^{2}]=\s^{2}\equiv J^{2}/\begin{pmatrix}N\\
p
\end{pmatrix}
\end{equation}
Note that this definition of variance is compatible with \cite{Erdos:2014zgc,Berkooz:2018jqr} though different from the usual convention in \cite{Maldacena:2016hyu} by a constant factor of $p^2/N$. For the partition function, this difference can be reconciled by a shift of inverse temperature $\b\ra\b/\sqrt{p^2/N}$. In the following we will choose $J=1$. Let us denote $I=\{i_{1}\cdots i_{p}\}$
and write $\tilde{H}_{a}=\sum_{I}J_{I}^{a}X_{I}^{a}$ with a compressed
notation. By definition, we have $(X_I^a)^2=\I_{2^N}$.

As a reference, let us write down the Hamiltonian of the SYK model \cite{sachdev1993gapless,kitaev2015simple} that is closely related to our model. To have a better comparison, we slightly changed the notation as
\be 
\tilde H_{SYK}=i^{p}\sum_{i_1<\cdots<i_{2p}}J_{i_1\cdots i_{2p}}\psi_{i_1}\cdots \psi_{i_{2p}} \label{4-syk}
\ee
and the random coupling is 
\begin{equation}
\E[J_{i_{1}\cdots i_{2p}}]=0,\quad\E[(J_{i_{1}\cdots i_{2p}})^{2}]= J^{2}/\begin{pmatrix}2N\\
2p
\end{pmatrix} \label{5-syk}
\end{equation}
with a convenient choice $J=1$.

\subsection{Multi-color chord diagrams}

As explained in Sec. \ref{sec:intr}, we are interested in the spectrum of the dcSYK model. For this purpose, we need to compute the partition function of $\tilde H$ under ensemble average of all $J^a_I$'s
\begin{align} 
Z(\b)&=\E [\tr e^{-\b \tilde H}]=\sum_{n=0}^\infty\f{(-\b)^n}{n!}M_n \label{7-Z}\\
M_n&\equiv \f{1}{d^{n/2}}\sum_{a_i}\E[\tr\tilde H_{a_1}\cdots \tilde H_{a_n}] \label{6-mn}
\end{align}
The $n$-moment $M_n$ can be expanded in terms of products of $n$ $X^a_I$
\be 
\E[\tr\tilde H_{a_1}\cdots \tilde H_{a_n}]=\sum_{I_i}\E [J_{I_1}^{a_1}\cdots J_{I_n}^{a_n}]\tr X^{a_1}_{I_1}\cdots X^{a_n}_{I_n} \label{6}
\ee
and compute each $n$-th moment of $X^a_I$. As the ensemble is Gaussian, $\E [\cdots]$ is equivalent to sum over all possible Wick contractions of $J^a_I$. It follows that $n$ must be even to have a nontrivial $M_n$, in which $\tilde H_a$ are pairwise contracted. Therefore, we can represent $M_{2n}$  as summing over all chord diagrams on a disk with $n$ chords \cite{Garcia-Garcia:2017pzl,Garcia-Garcia:2018fns,Berkooz:2018jqr,Berkooz:2018qkz} (see Fig. \ref{1a} for an illustration of a chord diagram). The simplest example is
\begin{equation}
M_2=\f{\s^2}{d}\sum_{a,I}\tr X_{I}^{a}X_{I}^{a}=2^N
\end{equation}
The second example is
\begin{align}
M_4=&\f{1}{d^2}\sum_{a,b} (\E[\tr\tilde H_a\tilde H_a\tilde H_b\tilde H_b]+\E[\tr\tilde H_a\tilde H_b\tilde H_b\tilde H_a]  \nn\\
&+\E[\tr\tilde H_a\tilde H_b\tilde H_a\tilde H_b]  ) \nn\\
=&2\cdot 2^N+\f{\s^4}{d^2}\sum_{a,b,I,I'}\tr X^a_I X^b_{I'} X^a_I X^b_{I'} \label{9} 
\end{align}
where the three terms in the first line are represented as chord diagrams in Fig. \ref{1b}.

\begin{figure}
\begin{centering}
\subfloat[\label{1a}]{\begin{centering}
\includegraphics[height=2cm]{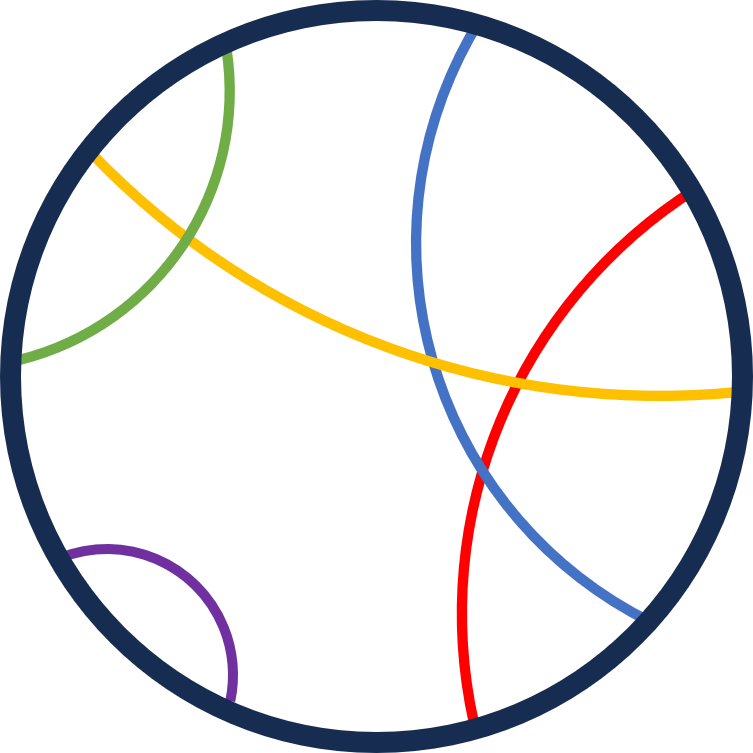}\end{centering}}
\hfill\subfloat[\label{1b}]{\begin{centering}
\vspace{0.12cm}
\includegraphics[height=1.8cm]{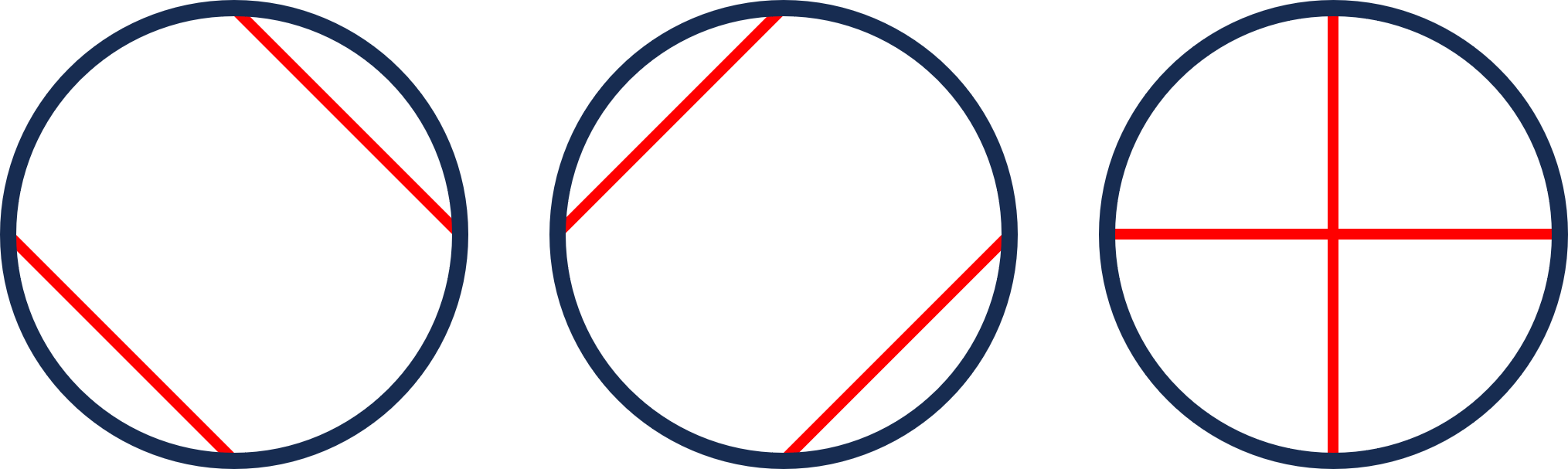}\end{centering}}
\par\end{centering}
\caption{\justifying (a) A generic chord diagram within $M_{10}$, where each end of a chord represents a $\tilde H_a$ in the trace, and each chord represents two contracted $\tilde H_a$ with identical index $a$. Different $a$ are labeled by different colors. (b) The three chord diagrams for $M_4$.}
\end{figure}

To compute the second term in \eqref{9} is generically complicated. 
Since $X_{I}^{a}$ is a string of $\psi_{i}$, they obey the following
commutation relation
\begin{equation}
X_{I}^{a}X_{I'}^{b}=(-)^{|I\cap I'|}X_{I'}^{b}X_{I}^{a},\quad(a\neq b)
\end{equation}
where $|I\cap I'|$ means the total numbers of overlapped Majorana
fermions in the two strings $X_{I}^{a}$ and $X_{I'}^{b}$. Applying this two the second term of \eqref{9} and noting $X^a_I$ and $X^b_{I'}$ being commutative if $a=b$, we can switch the order of $X^b_{I'}$ and $X^a_I$ and have
\be 
M_4=2^N(2+1/d)+\f{\s^4}{d^2}\sum_{a\neq b,I,I'}(-)^{|I\cap I'|} \label{eq:m4-11}
\ee
where the $1/d$ term comes from $a=b$.

Summing over the last term becomes an involved combinatorial problem, which becomes practically even harder for higher moments because we need to switch multiple pairs of $X^a_I$. Fortunately, the following double-scaling limit renders the computation of the $n$-th moment tractable
\begin{equation}
N\ra\infty,\quad\lambda \equiv 4p^{2}/N\quad\text{fixed}\label{eq:4}
\end{equation}
In the double scale limit (\ref{eq:4}), the combinatorial problem reduces to a probability problem \cite{Berkooz:2018jqr,Berkooz:2018qkz,Erdos:2014zgc}. In other words, we only need to compute the probability distribution of overlap $s$ fermions for
two random strings. 

Let us take $a=1$ and $b=2$. Given random choice of $I$ and $I'$, in the double scale limit,
$p\ll N$ and the typical draw of $I$ is sparse. For a fixed $I$
with $s$ Majorana fermion overlaps, there are $\begin{pmatrix}2p\\
s
\end{pmatrix}$ choices for overlapping spots, and $\begin{pmatrix}N-2p\\
p-s
\end{pmatrix}$ ways of $I'$ to overlap with them. Here $2p$ means two Majoranas
in each $\mX_{i}^{1}$ has potential to be overlapped by two different
$\mX_{i}^{2}$. In the double scale limit, the probability of $s$
Majonaran overlap is a Poisson distribution
\begin{equation}
P(s)=\begin{pmatrix}2p\\
s
\end{pmatrix}\begin{pmatrix}N-2p\\
p-s
\end{pmatrix}/\begin{pmatrix}N\\
p
\end{pmatrix}\ra\f{(\lambda/2)^{s}}{s!}e^{-\lambda/2}\label{eq:11}
\end{equation}
Summing over all $s$ leads to unity justifies our sparse distribution
assumption. 

For generic $a\neq b$, we show in Appendix \ref{app:1} that if $N$ is coprime with any number between 2 and $d-1$, we can always redefine $\psi_i$ such that the overlap problem becomes exactly the same as $a=1$ and $b=2$. If the coprime condition does not hold, this equivalence still holds in the large $N$ limit unless $d$ is fine-tuned and comparable with $N$.

With the probability distribution $P(s)$, we can easily compute the last term in \eqref{eq:m4-11} in the double scale limit
\begin{equation}
\f{\s^{4}(d^2-d)2^N}{d^2}\sum_{I_{a},I_{b}}(-)^{s}P(s)=(1-1/d)2^Ne^{-\lambda}
\end{equation}
Let us define 
\begin{equation}
q=e^{-\lambda}\in[0,1] \label{15-q}
\end{equation}
By \eqref{eq:m4-11}, $q$ is indeed the ensemble averaged penalty factor for exchanging two neighboring $\tilde H_a$ and $\tilde H_b$ with $a\neq b$. Relating this to the third chord diagram in Fig. \ref{1b}, we can assign $q$ to the crossing of the two chords therein for each pair $a\neq b$. On the other hand, if $a=b$ the penalty factor is just 1 because all $X^a_I$ are commutative to each other for fixed $a$. The latter explains the $1/d$ in the first term of \eqref{eq:m4-11}.

This pictorial algorithm generalizes to the computation of any moment $M_{2n}$ (for fixed $n$) because the overlapping of any two neighboring $\tilde H_a$ with different $a$ in the double scale limit can be regarded as indepedent Poisson distribution $P(s)$ due to the fact that the total $np$ $\mX_i^a$ in each diagram is sparsely drawn from $N$ $\mX^a_i$'s. 
Therefore, we summarize the rule as follows: 
\begin{enumerate}
    \item $M_{2n}$ is the sum over all possible Wick contraction of $\tilde H$. Each contraction pattern corresponds to a $n$-chord diagram.
    \item For each chord diagram, we can choose to color each chord from $d$ colors.
    \item For each crossing in a colored chord diagram, if the crossing is by two different colors, we give penalty factor $q$, or otherwise give penalty factor 1.
    \item The value of each colored chord diagram is the product of all penalty factors. 
    \item $M_{2n}$ is the sum over all diagrams with all possible colors times the normalization $2^N/d^{n}$.
\end{enumerate}

\subsection{Multi-color $q$-harmonic oscillators} \label{sec:2B}

The counting problem for single-color chord diagrams has been completely solved in \cite{touchard1952probleme,riordan1975distribution}. However, for the $d$-color chord diagram problem considered in this paper, there is no exact and complete solution yet (see also recent progress \cite{Lin:2023trc,Xu:2024hoc} of decomposing the representation of this algebra into irrep of 0- and 1-particle states). Nevertheless, there is an algebraic way to solve the chord diagram counting for the double-scaled SYK model \cite{Berkooz:2018qkz,Berkooz:2018jqr}, which can be easily generalized to the $d$-color case.

\begin{figure}
\begin{centering}
\includegraphics[height=2.5cm]{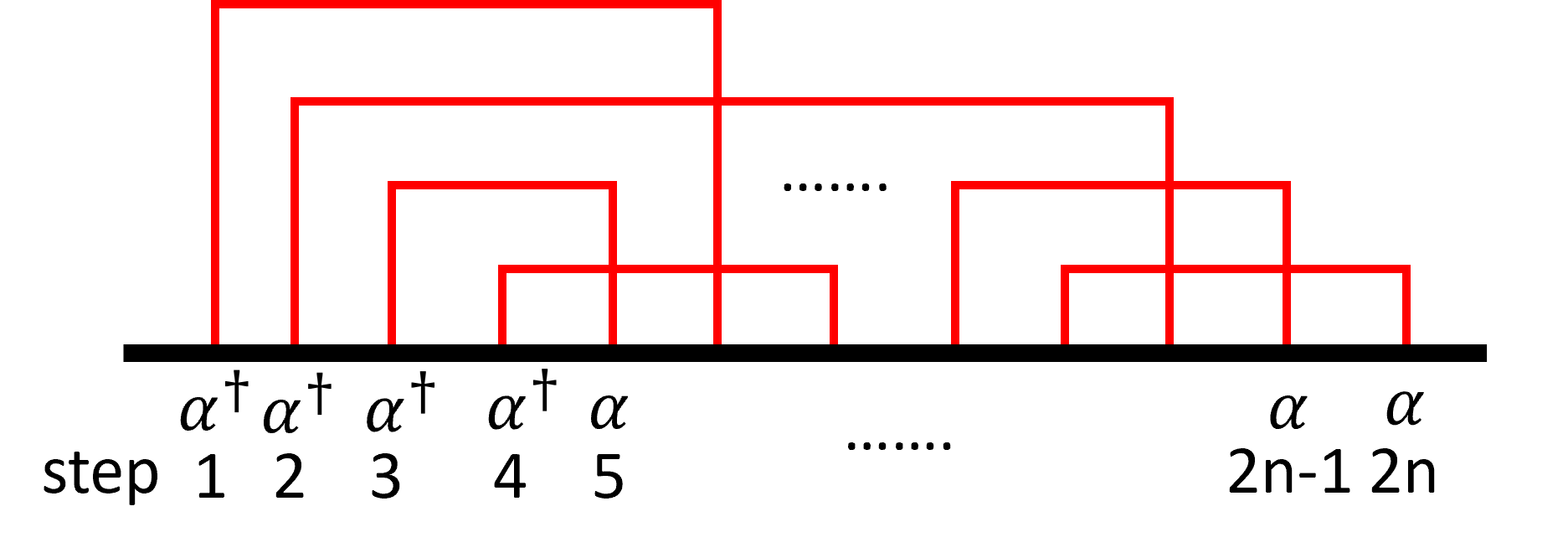}
\end{centering}
\caption{\justifying  Cut open a chord diagram. To generate a $n$-chord diagram, we need $2n$ steps. \label{fig:2}}
\end{figure}

Let us first review the case for the double-scaled SYK model, which corresponds to $d=1$ with $q=e^{-\lam}$ as the penalty factor for each crossing \footnote{The notation in \eqref{4-syk} and \eqref{5-syk} are chosen such that the penalty factor $q$ for the SYK model is the same as our multi-color commuting SYK model.}. For any chord diagram on a disk, we can cut open it and put all chords on a boundary line \cite{Berkooz:2018qkz,Berkooz:2018jqr} as illustrated in Fig. \ref{fig:2}. To draw a chord, it requires two steps: (1) creating a chord by drawing an open chord from the boundary line; (2) annihilating a chord by connecting the dangling end of the open chord to the boundary line.  An $n$-chord diagram has $2n$ ends on the boundary line, which corresponds to $2n$ steps to generate it. At each end, one can choose to create a chord, represented by $\a^\dagger$, or annhilate a line, represented by $\a$. Note that any pair of chords only cross each other at most once. To avoid overcounting the crossing number of chords, we require that the new open chord created by $\a^\dagger$ does not intersect any existing open chords. In other words, crossing only occurs at annihilating an open chord.

We can associate a Fock space representation for the these operators. Denote $\ket{0}$ as the state for no chords and $\ket{n}=\a^{\dagger n}\ket{0}$ as the state by creating $n$ open chords. Acting $\a$ on $\ket{n}$ annihilates one of the open chord from the $n$ choices. If we label the open chord from bottom to top in $\ket{n}$ as 1 to $n$, annihilating the $i$-th open chord will cross $i-1$ open chords below it. Therefore, the total penalty factor for acting $\a$ on $\ket{n}$ is
\be 
1+q+\cdots+q^{n-1}=\f{1-q^n}{1-q}=[n]_q
\ee
where $[n]_q$ is the $q$ analog of number $n$. After annihilating one open chord, state $\ket{n}$ becomes proportional to $\ket{n-1}$
\be 
\a\ket{n}=[n]_q\ket{n-1} \label{17}
\ee
For zero chord state, we require $\a\ket{0}=0$. Using \eqref{17}, we can show that $\a$ and $\a^\dag$ obey the algebra of $q$-harmonic oscillator
\be 
[\a,\a^\dag]_q\equiv \a\a^\dagger-q \a^\dag \a=1 \label{18-alg}
\ee
Using this algebra, we can define the inner product of states as
\be 
\avg{n|n'}\equiv\avg{0|\a^n\a^{\dag n'}|0}=\d_{nn'}[n]_q!
\ee
where the computation is by moving $\a$ across all $\a^\dagger$ using \eqref{18-alg} to kill $\ket{0}$, and $[n]_q!=[n]_q[n-1]_q\cdots [1]_q$ is the $q$-factorial. This definition of inner product explicitly treats $\a$ as the adjoint of $\a^\dag$. It follows that the Hilbert space spanned by $\ket{n}$ is a well-defined Fock space.

Recall that to compute the $n$-th moment of the SYK Hamiltonian $\tilde H_{SYK}$ in the double scale limit, we need to sum over all possible $n$-chord diagrams, which correspond to all Wick contraction patterns of $n$ $\tilde H_{SYK}$. Representing a chord diagram on a boundary line as Fig. \ref{fig:2}, the contraction of each $\tilde H_{SYK}$ can be through either creating an open chord or annihilating an open chord. Therefore, we define a new Hamiltonian $h=\a+\a^\dag$ \footnote{This $h$ is the same as the transfer matrix $T$ in \cite{Berkooz:2018qkz,Berkooz:2018jqr} after a similarity transformation.} as the avatar for the SYK Hamiltonian $\tilde H_{SYK}$. In other words, we have
\be 
\avg{0|h^n|0}=2^{-N}\E [\tr \tilde{H}^n_{SYK}]
\ee
Note that the Hamiltonain $h$ defined in this way is Hermitian, and related to the transfer matrix in \cite{Berkooz:2018qkz,Berkooz:2018jqr} by a similarity transformation.

The eigenvalue and eigenfunction of $h$ can be solved explicitly \cite{Berkooz:2018qkz,Berkooz:2018jqr} as
\begin{align} 
h\ket{\t}&=E(\t)\ket{\t}=\f{2\cos \t}{\sqrt{1-q}}\ket{\t}, \quad \t\in[0,\pi] \label{21-h}\\
\avg{\t|\t'}&=\f{\d(\t-\t')}{\mu(\t)},\quad \mu(\t)=\f{(e^{2i \t},e^{- 2i \t},q;q)_\infty}{2\pi}\\
\avg{\t|n}&=\psi_n(\t)=\f{H_n(\cos\t|q)}{(q;q)_n},\quad \psi_0(\t)=1
\end{align}
where $H_n(x|q)$ is the $q$-Hermite polynomial of order $n$, and the notations of $q$-Pochhammer are defined as
\begin{align} 
(a;q)_n&=\prod_{k=1}^n(1-aq^{k-1})\\(a_1,\cdots,a_n;q)_\infty&=\prod_{i=1}^k(a_i;q)_\infty
\end{align}
From \eqref{21-h} we see that the spectrum of the double-scaled SYK model is compact and bounded between $-2/\sqrt{1-q}$ and $2/\sqrt{1-q}$. The partition function is given by
\be 
Z(\b)=2^N\avg{0|e^{-\b h}|0}=2^N\int_0^\pi d\t\mu(\t)e^{-\b E(\t)}
\ee
which implies the spectrum
\be 
\r(E)=\mu(E(\t))|d\t/dE|,\quad E\in[-|E_0|,|E_0|] \label{30-spec}
\ee
where $E_0=-2/\sqrt{1-q}$ is the ground energy, and we have rescaled the normalization such that $\int dE \r(E)=1$.

The generalization to our $d$ color case is straightforward. At each step of generating a chord diagram, we now have $d$ choices of creating an open chord using $\a^\dag_a$ or annihilating an existing $a$-th color open chord using $\a_a$ for $a=1,\cdots,d$. Unlike the one-color case, the ordering of chords with different colors are not exchangable, and we need to label all states as a string
\be
\ket{a_1\cdots a_n}=\a^\dagger_{a_1}\cdots \a^\dagger_{a_n}\ket{0} \label{31-eq}
\ee
Acting $\a_a$ to this state, it can choose to kill any one of the existing $a$-th color open chord therein
\begin{align} 
&\a_a\ket{a_1,\cdots,a_n}\nn\\
=&\sum_{i=1}^n \d_{a,a_i} q^{i-1-\sum_{j<i} \d_{a,a_j}} \ket{a_1,\cdots,\slashed{a_i},\cdots,a_n}
\end{align}
and $\a_a\ket{0}=0$ for all $a$. One can use this formula to show the multi-color $q$-harmonic oscillator algebra
\begin{equation}
[\a_{a},\a_{a}^{\dagger}]=1,\quad[\a_{a},\a_{b}^{\dagger}]_q=0\quad(a\neq b) \label{31-cr}
\end{equation}
We define the bra state as $\bra{a_1,\cdots,a_n}=\bra{0}\a_{a_n}\cdots \a_{a_1}$, and the inner product is computed using the commutation relation \eqref{31-cr} to move all $\a_a$ to right to kill $\ket{0}$. With this inner product, the states defined in \eqref{31-eq} are neither unit normalized or orthogonal basis. However, the full Hilbert space is split into direct sum over subspaces labeled by $\{n^a\}_{a=1,\cdots,d}$, which counts how many $a_i=a$ in the string $a_1\cdots a_n$ in \eqref{31-eq} for each $a=1,\cdots,d$
\be 
\mH=\bigoplus_{\{n^a\}} \mH_{\{n^a\}} 
\ee
The dimension of each subspace is 
\be 
\dim \mH_{\{n^a\}}=\f{(\sum_a n^a)!}{\prod_a n^a!} \label{35-dim}
\ee
One can further use Gram–Schmidt process for \eqref{31-eq} to find the orthogonal unit basis in each $\mH_{\{n^a\}}$. We can also generalize the commutation relation to
\be 
\a_a \a_b^\dagger -q_{ab}\a_b^\dagger\a_a=\d_{ab}
\ee
for a generic symmetric matrix $q_{ab}$ and show that the Fock space is well-defined if $q_{ab}\in[-1,1]$ \cite{speicher1993generalized}.

The corresponding Hamiltonian for this system is
\begin{equation}
H=\f 1{\sqrt{d}}\sum_{a=1}^{d}H_{a}\equiv\f 1{\sqrt{d}}\sum_{a=1}^{d}(\a_{a}+\a_{a}^{\dagger})\label{eq:2}
\end{equation}
One can easily check that the $n$-th moment \eqref{6-mn} in the double scale limit can be computed as
\be 
M_n=2^N\avg{0|H^n|0}
\ee
Unfortunately, solving the eigenvalues and eigenfunctions of $H$ for generic $d$ and $q$ is very hard and we do not have a way to write down the spectrum.

\subsection{Double-scaled SYK: $d\ra \infty$ limit} \label{sec:2C}

Though we are unable to solve the problem for generic $d$, the solution should be equivalent to the double-scaled SYK model in $d\ra\infty$ limit. The idea is straightforward. For $2n$-th moment $M_{2n}$, it is the sum over all colored $n$-chord diagrams with $d$ possible colors on each chord. Suppose $d\gg n$, and for a typical colored chord diagram, all $n$ chords should have different colors. Therefore, the value of a typical colored chord diagram is simply $q$ to the power of the number of crossings, which is identical to the value of a chord diagram in the double-scaled SYK model as we reviewed in the last section. The number of such typical colored chord diagrams is $\begin{pmatrix}d\\
n
\end{pmatrix}n!$ and the total ways of coloring $n$ chords are $d^n$. In large $d$ limit, we have 
\be 
\lim_{d\ra\infty}\begin{pmatrix}d\\
n
\end{pmatrix}n!/d^n =1
\ee
which shows that 
\be 
\lim_{d\ra \infty} M_n=2^N\avg{0|h^n|0}
\ee

Indeed, the leading order large $d$ correction can be carried out by the same probability crossing argument. For any two crossing chords, they have $1/d$ probability to be the same color and $1-1/d$ probablity to be different colors. Therefore, for each crossing, the averaged penalty factor is 
\be
\bar q=1/d+(1-1/d)q=1-(1-1/d)(1-q) \label{39-eq}
\ee
which implies that the leading order correction to the spectrum of the $d$-commuting SYK model is the same as the double-scaled SYK model \eqref{30-spec} with replacement $q\ra \bar q$.

Physically, that the large $d$ limit reduces back to the double-scaled SYK model is not surprising. Given the regular SYK model Hamiltonian \eqref{4-syk}, we can decompose the total $\choose{2N}{2p}$ terms into $d$ groups such that all terms within each group are mutually commutative but terms from different groups are possibly non-commutative. For each group, the upper bound of the number of terms is $\choose{N}{p}$, namely the number of terms in each commuting SYK model $\tilde H_a$ \footnote{To show this, note that each term in the group need to overlap with another term with even numbers of Majorana fermions. Then one can always find a common 2-fermion basis like $\mX^a_j$ for the largest possible group. Then this group is equivalent to one $\tilde H_a$.}. This decomposition is far from unique, but since $\choose{N}{p}\ll\choose{2N}{2p}$ we should naturally expect existing a decomposition with most groups almost saturate this upper bound, which implies
\be 
d\app \choose{2N}{2p}/\choose{N}{p}\sim N^p
\ee
which is infinite in the double scale limit. Even though for each group, they are not organized as the $\tilde H_a$ in the dcSYK model, we may expect the overlapping counting between two groups obey the same statistics in the double scale limit. Taking this $d$ into the definition of our dcSYK model \eqref{1-dcsyk}, the variance of the random coupling is
\be 
\E[(J_{i_{1}\cdots i_{p}}^{a}/\sqrt{d})^{2}]\app 1/\choose{2N}{2p}
\ee
which coincides with the variance \eqref{5-syk} of the regular SYK model. This suggests the validity of this estimation of the decomposition in the double scale limit.

Taking this perspective of decomposing the regular SYK model into $d$ copies of commuting SYK model with $d\ra\infty$ actually means something more general. In above estimation, we take $d\sim N^p$ but in our setup we only require large $d$, which can be much smaller than $N$ and $p$. In this sense, we can view the dcSYK model as a sparse version of the regular SYK model with the same spectrum in the large $d$ and double scale limit. Since there were many discussions about sparse SYK models, we will comment on this point further in Sec. \ref{sec:concl}.

\section{Exact result for $q=0$: free probability} \label{sec:free}

Though the dcSYK model is hard to solve for a generic case, there is an exact solution for the special case $q=0$, which corresponds to $p\gg\sqrt{N}$. In the chord diagram language, the crossing between two chords with different colors is disallowed while the crossing between same-color chords has the same penalty factor 1 as before. This non-crossing feature between different color chords reflects the fact that all $H_a$ can be regarded as freely independent random variables. Given the spectrum of each $H_a$, there are well-established methods about free probability to compute the spectrum of their sum. For self-containedness, let us quickly review some relevant concepts of free probability in the following section. Curious readers may refer to \cite{nica2006lectures,mingo2017free}.

\subsection{A snapshot of free probability}

The materials in this subsection mostly follow \cite{nica2006lectures,mingo2017free}.
\begin{definition}
    For any algebra $\mA$ containing multiplicative identity $1$ (unital) and linear map $\vp:\mA\ra\C$ with $\vp(1)=1$, we call the pair $(\mA,\vp)$ as a {\it non-commutative probability space}. The elements of $\mA$ are called random variables. 
\end{definition}
As discussed in Sec. \ref{sec:2B}, the algebra $\mA$ of the $d$ color chords is generated by $\a_a$ and $\a_a^\dagger$, and $\avg{0|\cdot|0}$ is a linear map $\vp$.
\begin{definition}
    For $(\mA,\vp)$, suppose $\mA_1,\cdots,\mA_s$ are unital subalgebras. We say $\mA_1,\cdots,\mA_s$ are freely independent with respect to $\vp$ if for $r\geq 2$ elements $a_1,\cdots,a_r\in\mA$ such that for all $i=1,\cdots,r$
    \begin{itemize}
         \item $\vp(a_i)=0$ 
        \item $a_i\in \mA_{j_i}$ with $1\leq j_i\leq s$
        \item $j_1\neq j_2,j_2\neq j_3,\cdots,j_{r-1}\neq j_r$
    \end{itemize}
    we must have $\vp(a_1\cdots a_r)=0$. 
\end{definition}
By this definition, the subalgebras $\mA_a$ generated by $\a_a$ and $\a_a^\dag$ respectively are freely independent when $q=0$. Indeed, for   $\avg{0|H_{a_1}\cdots H_{a_n}|0}$ with all neighboring pairs $H_{a_i}$ and $H_{a_{i+1}}$ from different subalgebras, we will only have chord diagrams with crossings of different color chords. All of these diagrams vanish due to $q=0$. To characterize the free independence, we need to introduce free cumulant, which involves the non-crossing partition of $n$ numbers.
\begin{definition}
A partition $\pi$ is called non-crossing if there do not exist numbers $i,j,k,l$ with $1\leq i<j<k<l\leq n$ such that: $i$ and $k$ are in the same block of $\pi$, $j$ and $l$ are in the same block of $\pi$, but $i$ and $j$ are not in the same block of $\pi$. The collection of all non-crossing partitions of $\{1,\cdots,n\}$ is denoted as $NC(n)$.
\end{definition}
For example, the partition of $\{1,2,3,4\}$ into $\{(1,2),(3,4)\}$ is non-crossing while $\{(1,3),(2,4)\}$ is not. 
\begin{definition}
    Let $(\mA,\vp)$ be a non-commutative probability space. The corresponding free cumulants $\kappa_n:\mA^n\ra\C~(n\geq 1)$ are defined inductively in terms of moments by the moment-cumulant formula
    \be 
    \vp(a_1\cdots a_n)=\sum_{\pi\in NC(n)}\kappa_n(a_1,\cdots,a_n)
    \ee
    where, by definition, if $\pi=\{V_1,\cdots,V_r\}$ then
    \be  
\kappa_\pi(a_1,\cdots,a_n)=\prod_{\substack{V\in \pi \\ V=(i_1,\cdots,i_l)}} \kappa_l(a_{i_1},\cdots,a_{i_l})
    \ee
\end{definition}
This definition is inductive and starts with $n=1$: for $n=1$ we have $\vp(a)=\kappa_1(a)$; for $n=2$ we have $\vp(a_1 a_2)=\kappa_{\{(1,2)\}}(a_1,a_2)+\kappa_{\{(1),(2)\}}(a_1,a_2)=\kappa_2(a_1,a_2)+\kappa_1(a_1)\kappa_1(a_2)$, which implies $\kappa_2(a_1,a_2)=\vp(a_1a_2)-\vp(a_1)\vp(a_2)$. One can repeat this algorithm to any $n$ and express $\kappa_n$ in terms of a linear combination of products of $\vp$.

The most important property of free cumulants is it characterizes free independence. 
\begin{theorem}
The subalgebras $\mA_1,\cdots,\mA_s$ are freely independent if and only if all mixed cumulants vanishes, namely $\kappa_n(a_1,\cdots,a_n)=0$, where each $a_i$ comes from one of these subalgebras, but they do not all come from the same subalgebra. 
\end{theorem}
Using this theorem, we immediately have the following result for the sum of freely independent random variables. 
\begin{corollary}\label{45-cor}
    Let us define
\be 
\kappa_n^a\equiv \kappa_n(a,\cdots,a)
\ee
If $a$ and $b$ are freely independent random variables, then we have for all $n\geq 1$
    \be 
    \kappa^{a+b}_n=\kappa^a_n+\kappa^b_n 
    \ee
\end{corollary}
Since our Hamiltonian $H$ is the sum over freely independent $H_a$, the free cumulant of $H$ is simply the sum over the cumulant of $H_a$, which is just $d$ times the latter.

For any $a\in\mA$, let us define 
\be 
s_n^a\equiv \vp(a^n) 
\ee
and consider two formal series of $z$ defined as
\begin{align} 
M_a(z)&=1+\sum_{n=1}^\infty s_n^a z^n \\
C_a(z)&=1+\sum_{n=1}^\infty \kappa_n^a z^n \label{47-C}
\end{align}
Using the definition of free cumulant, one can show that these two series have a simple relation
\be 
M_a(z)=C_a(zM_a(z)) \label{48-mc}
\ee
Define the resolvent of a random variable $a$ as
\be 
R_a(z)=\vp(1/(z-a))=M_a(1/z)/z
\ee
Due to Corollary \ref{45-cor}, for two freely independent variables $a$ and $b$, we have $C_{a+b}(z)=C_a(z)+C_b(z)-1$. It follows from \eqref{48-mc} that
\begin{theorem} \label{thm-2}
    For $d$ freely independent variables $a_1,\cdots,a_d$, and $a\equiv a_1+\cdots+a_d$,  we have 
    \be 
    R_{a}(z)=z^{-1}\sum_{i=1}^d C_{a_i}(R_a(z))-(d-1)/z
    \ee
    where each $C_{a_i}(z)$ obeys
    \be 
    R_{a_i}(z)=C_{a_i}(R_{a_i}(z))/z \label{51-thm2}
    \ee
\end{theorem}

For any finite distribution of random variables, to determine whether the spectrum is compact or not, let us quote the following two lemmas
\begin{lemma} \label{lemma3} (Lemma 13.13 of \cite{nica2006lectures}) Let $\r$ be a finite measure on $\R$. Then the following statements are equivalent:\\
(1) $\r$ has compact support.\\
(2) The moment $s_n$ exist for all $n\in\N$ and the sequence does not grow faster than exponentially, namely existing a constant $c>0$ such that
\be |s^a_n|\leq c^n  \ee
for all $n\in N$.
\end{lemma}
\begin{lemma} \label{lemma4} (Proposition 13.15 of \cite{nica2006lectures})
    That the moment sequence $s_n^a$ does not grow faster than exponentially is equivalent to that the free cumulant $\kappa_n^a$ does not grow faster than exponentially.
\end{lemma}
Using these two lemmas, we immediately have the following result.
\begin{proposition} \label{prop}
for $d$ freely independent variables $a_i$, if each has non-compact spectrum, then their sum $a=a_1+\cdots+a_d$ must have non-compact spectrum if $d$ is finite. 
\end{proposition}
\begin{proof}The proof follows from Corollary \ref{45-cor} that the free cumulant of $a$ is simply the sum of that of each $a_i$. By above two lemmas, $\kappa_n^{a_i}$ grows faster than exponentially, and so is $\kappa_n^a$ and $s_n^a$.
\end{proof}

\subsection{Spectrum of $H$ for $q=0$}

Let us define the resolvent for $H$ as $R(z)$ and the resolvent for each commuting element $H_a/\sqrt{d}$ as $R_c(z)$. Accordingly, we define $C_c(z)$ in \eqref{47-C} for $H_a/\sqrt{d}$. The Hamiltonian of the commuting SYK model has Gaussian distribution \cite{Gao:2023gta}
\be 
\r_c(x)=\sqrt{\f{d}{2\pi}}e^{-d x^2/2} \label{55-rho}
\ee
which implies 
\be 
R_c(z)=\sqrt{\f{d}{2\pi}}\int_\R \f {dx} {z-x}e^{-d x^2/2} \label{53-rc}
\ee
It obeys a simple differential equation
\be 
d^{-1}\del_z R_c(z)+z R_c(z)-1=0 \label{54-dif}
\ee
By \eqref{51-thm2}, the inverse function of $R_c(x)$ is $R^{-1}_c(x)=C_c(z)/z$. It follows from \eqref{54-dif} that
\be 
\del_z(C_c(z)/z)=\f 1 {d(1-C_c(z))} \label{55-dC}
\ee
Using Theorem \ref{thm-2}, we have 
\be  
z/d+(d-1)/(d r)=C_c(r)/r,\quad r=R(z)
\ee
Taking $\del_r$ on both sides and applying \eqref{55-dC} leads to
\be 
\del_r\left(z+\f{d-1}{r}\right)=\f 1 {1-\left(z+\f{d-1}{r}\right)\f r d}
\ee
Define a new variable $w=z/d+(1-1/d)/r$, and above equation can be written as
\be 
d^{-1}\del_w r+w r-1=0 \label{58-dif}
\ee
which is exactly the same differential equation \eqref{54-dif}. By \eqref{53-rc}, we derive the following integral representation of $R(z)$
\be 
R(z)=\sqrt{\f{d}{2\pi}}\int_\R \f {dx} {z/d+(1-1/d)/ R(z)-x}e^{-d x^2/2} \label{59-r}
\ee
A consistent check is $R(z)\ra 1/z$ for $z\ra \pm i\infty$, which indeed fixes the integral constant of the differential equation \eqref{58-dif} and justifies \eqref{59-r}.

Let us first consider the $d\ra \infty$ limit. Redefining $x\ra x/\sqrt{d}$ in \eqref{59-r} and expanding the integrand in series of $1/d$, the RHS becomes
\be 
 \int \f {R(z)e^{-x^2/2}} {\sqrt{2\pi}} \left(1+ \f {1-zR(z)+x^2 R(z)^2}{d}+O(d^{-2})\right) \label{60}
\ee
The first term simplies gives $R(z)$ that matches with the LHS. Therefore, we should expect $1/d$ order term in \eqref{60} vanishes after integrating $x$, which leads to
\begin{align} 
&1-zR(z)+ R(z)^2=0 \\
\implies R(x)=&R_{\infty}(x)\equiv(x-\sqrt{x^{2}-4})/2
\end{align}
The spectrum is given by the discontinuity of the resolvent
\be 
\r(x) =\f{R(x+i\e)-R(x-i\e)}{-2\pi i} \label{63-rho}
\ee
which leads to the Wigner's semicircle
\be 
\r_\infty(x)=\f 1 {2\pi}\sqrt{4-x^2},\quad x\in [-2,2]
\ee
This is exactly the result from the free central limit theorem \cite{nica2006lectures,mingo2017free}, and also consistent with the $q\ra 0$ limit of the double-scale SYK model
\be 
\mu(\t)|d\t/dE|\ra \f {\sin\t}{\pi}=\r_\infty(E),\quad E=2\cos\t
\ee

For finite but large $d$, one can expand $R(z)$ as a series of $1/d$ and solve the integral equation \eqref{59-r} order by order. However, this perturbative solution of resolvent is divergent near the edge of the spectrum at $z=\pm 2$ even at $1/d$ order. This would naively give singular spectrum at $1/d$ order by \eqref{63-rho}. Indeed, for $d<\infty$ the spectrum is non-compact and has support on the whole real axis due to Proposition \ref{prop}. 

\begin{figure}
\begin{centering}
\includegraphics[width=8cm]{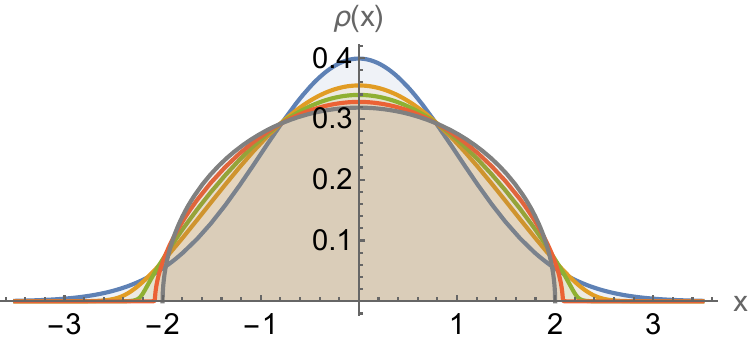}
\end{centering}
\caption{\justifying  The spectrum for $q=0$. In the plot, blue, yellow, green, red and gray are for $d=1,3,6,15,\infty$ respectively. \label{fig:3}}
\end{figure}

To numerically see how the spectrum looks like as we increase $d$, we can evaluate the integral \eqref{59-r} in terms of the Dawson function $D$
\be 
r=\sqrt{2d}D(\sqrt{d/2}w)\mp i\sqrt{\pi d/2}e^{-dw^2/2} \label{68-r-d}
\ee
where the $\mp$ sign is for $w\in\C_\pm$. As the spectrum is non-negative, by \eqref{63-rho} the map $z\mapsto r$ must flip the upper and lower half complex plane $\C_\pm \mapsto \C_\mp$. It follows that the map $z\mapsto w$ stays on the same half complex plane as $\C_\pm \mapsto \C_\pm$. Therefore, equivalently the $\mp$ sign in \eqref{68-r-d} is for $z\in\C_{\pm}$ respectively. Solve \eqref{68-r-d} numerically for real $z$ and different $d$. Using \eqref{63-rho}, we plot the spectra in Fig. \ref{fig:3}. From this figure, we see clearly that the spectrum becomes tighter and has a smaller tail as we increase $d$. 

As $d$ becomes large, the tail is almost invisible (e.g. the red curve for $d=15$ in Fig. \ref{fig:3}) and the spectrum globally looks quite close to the Wigner semicircle. However, as the spectrum must be non-compact, we will show later that the tail are exponentially suppressed and come from the non-perturbative effect of large $d$. On the other hand, the square-root law near the edge similar to Wigner semicircle is a perturbative effect from large $d$.

\subsubsection{Perturbative in $1/d$}

Let us zoom into the edge near $z=-2$ for large but finite $d$ case. It turns out to be convenient to work in $w$ variable rather than $r$. Writing the LHS of \eqref{68-r-d} as $(1-1/d)/(w-z/d)$, we can decompose \eqref{68-r-d} into two equations respectively for the real and imaginary parts
\begin{widetext}
\begin{align} 
 \frac{\left(1-\frac{1}{d}\right) w_i}{\left(w_r-\frac{z}{d}\right)^2+w_i^2} &= i\sqrt{\f{d}{2}}(D(\sqrt{d/2}w_+)-D(\sqrt{d/2}w_-))+ \sqrt{\f{\pi d}{2}}\f{e^{-dw_+^2/2}+e^{-dw_-^2/2}}{2} \label{69-eq}\\
 \frac{\left(1-\frac{1}{d}\right) \left(w_r-\frac{z}{d}\right)}{\left(w_r-\frac{z}{d}\right)^2+w_i^2}&=\sqrt{\f{d}{2}}(D(\sqrt{d/2}w_+)+D(\sqrt{d/2}w_-))-i\sqrt{\f{\pi d}{2}}\f{e^{-dw_+^2/2}-e^{-dw_-^2/2}}{2} \label{70-eq}
\end{align}
\end{widetext}
where $w_\pm=w_r \pm i w_i$ are separated as the real and imaginary parts with $w_i>0$. From \eqref{63-rho}, we know the spectrum is proportional to the imaginary part of $R(x)$, which roughly scales with $w_i$. Therefore, we need to consider small $w_i$ in the large $d$ expansion. Indeed, the expansion order depends on how small we assume $w_i$ is. Given $w_i=\tilde w_i/d^{n/2}$, we can take the following power law expansion ansatz
\be 
w_r=\sum_{j=0}^{n-1}w_{r,j}/d^j,~ z=E_0+y/d^n,~E_0=\sum_{j=0}^{n-1}z_j/d^j
\ee
where $z_0=-2$. Since we treat $1/d$ perturbatively, we can ignore all $e^{-dw_\pm^2/2}$ terms in \eqref{69-eq} and \eqref{70-eq} and solve them order by order. The results for the first a few orders are as follows
\begin{align} 
w_i&=\sqrt{z-E_0} \\
\{w_{r,j}\}_{j=0}^{n-1}&=\{-1,-5/2,-31/8,-357/16,\cdots\} \\
\{z_j\}_{j=0}^{n-1}&=\{-2,-1,-7/4,-57/8,\cdots\} 
\end{align}
With this perturbative solution, the spectral density is
\begin{align} 
\r(z)=&\frac{(1-1/d) w_i}{\pi((w_r-z/d)^2+w_i^2)} \label{75-eq}\\
=&\f{\sqrt{z-E_0}}{\pi}(1+O(1/d)) \label{76-eq}
\end{align}

This perturbative expansion at any order gives the $\sqrt{z+E_0}$ edge behavior and does not resolve the non-compactness of the spectrum. However, it gives the correction to the ground energy $E_0$ order by order. As we can see, the expansion coefficients $w_{r,j}$ and $z_j$ grow quickly as we go to higher orders, which reflects the fact that the large $d$ expansion is asymptotic and has zero radius of convergence. Usually, when we deal with the asymptotic expansion, we need to cutoff at some order $n$ depending on the choice of $d$. For practical purposes, we find for $d\sim O(10)$ that $n=4$ for \eqref{76-eq} gives a good approximation as illustrated in Fig. \ref{fig:4}.

\subsubsection{Non-perturbative in $1/d$} \label{sec: 3.2}

As we move away from the edge $z+E_0<0$, the perturbative expansion in $1/d$ does not hold anymore and $w_i$ is not polynomially but exponentially small. For this case, we should treat $w_i$ as the small parameter rather than $1/d$ to expand the two equations \eqref{69-eq} and \eqref{70-eq}. Expanding in linear order of $w_i$, these two equations become
\begin{align} 
\frac{(d-1) w_i}{(d w_r-z)^2}&=\sqrt{\f{\pi}{2d}}e^{-dw_r^2/2}+(1-\sqrt{2d}w_rD(\sqrt{d/2}w_r))w_i \label{77-eq}\\
\f{1-1/d}{dw_r-z}&=\sqrt{\f 2 d}D(\sqrt{d/2}w_r)-\sqrt{\f{\pi d}{2}}e^{-dw_r^2/2}w_rw_i \label{78-eq}
\end{align}
We can solve $w_i$ by \eqref{78-eq}
\be 
w_i=\f{\sqrt{\f 2 d}D(\sqrt{d/2}w_r)-\f{1-1/d}{dw_r-z}}{\sqrt{\f{\pi d}{2}}e^{-dw_r^2/2}w_r} \label{79-eq}
\ee
From \eqref{77-eq}, we know $w_i$ must be as small as $\sim e^{-dw_r^2/2}$, which by \eqref{79-eq} implies that we can assume
\be 
D(\sqrt{d/2}w_r)=\f{1-1/d}{\sqrt{2/d}(dw_r-z)}+\e \label{80-eq}
\ee
for small $\e\sim e^{-d w_r^2}$. Taking \eqref{79-eq} into \eqref{77-eq}, replacing the Dawson function using \eqref{80-eq}, and expanding it in leading order of $\e$, we can solve
\be 
\e=\frac{\pi  \sqrt{d} w_r e^{-d w_r^2} (z-d w_r)^2}{2 \sqrt{2} \left(d w_r^2-(d+1) w_r z+d+z^2-1\right)}
\ee
which by \eqref{79-eq} leads to
\be 
w_i=\sqrt{\frac{\pi }{2d}}\frac{ e^{-d w_r^2/2} (z-d w_r)^2}{d w_r^2-(d+1) w_r z+d+z^2-1} \label{82-eq}
\ee
This expression shows the leading order explicit non-perturbative effect $w_i\sim e^{-dw_r^2/2}$ for any $z,w_r,d \gg w_i$. Therefore, it universally holds for the asymptotic tail of any finite $d$. Apply the small $w_i$ expansion to \eqref{75-eq}, the spectrum density is
\be 
\r(z)\app\frac{(1-1/d) w_i}{\pi((w_r-z/d)^2)} \label{83-eq}
\ee

For large $d$, note that any non-perturbative correction to $w_r$ does not contribute to \eqref{82-eq}, we can simply set $\e=0$ in \eqref{80-eq} and solve it perturbatively in $1/d$. Taking ansatz
\be 
w_r(z)=\sum_{j=0}w_{r,j}(z)/d^j
\ee
we find the first two orders
\begin{align} 
w_{r,0}(z)&=(z-\sqrt{z^2-4})/2 \label{85-eq}\\
w_{r,1}(z)&=z+\f{z^2-3}{\sqrt{z^2-4}} \label{86-eq}
\end{align}
As these expressions have singularity at $z=-2$, they are valid only for $z<-2$. Taking them into \eqref{83-eq}, expanding it in large negative $y=z+2$ leads to
\be 
\r(z)\sim \f{d^{3/2}e^{-d-1}}{\sqrt{2\pi}}
e^{-2d(-y)-dy^2/2} \label{87-eq}
\ee

\begin{figure}
\begin{centering}
\includegraphics[width=2.8cm]{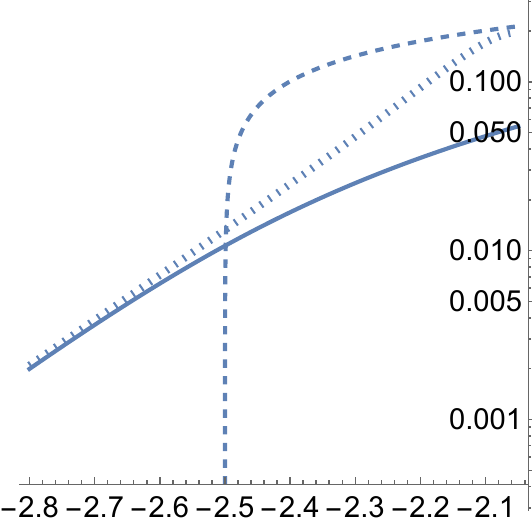}
\hfill
\includegraphics[width=2.8cm]{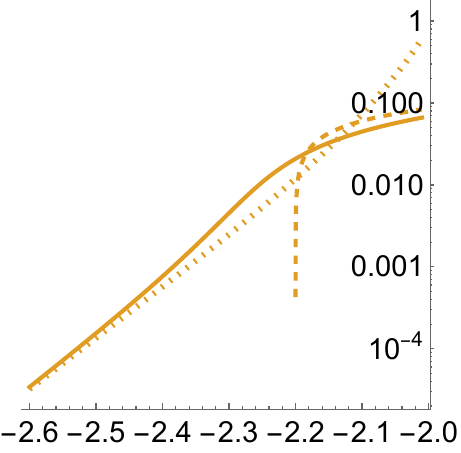}
\hfill
\includegraphics[width=2.8cm]{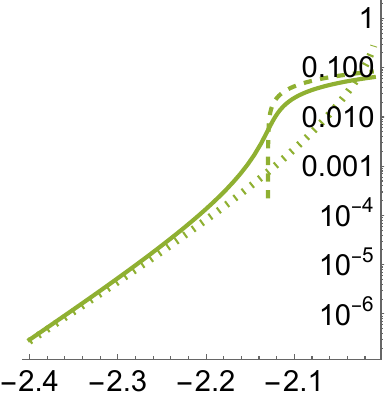}\\
\includegraphics[width=2.8cm]{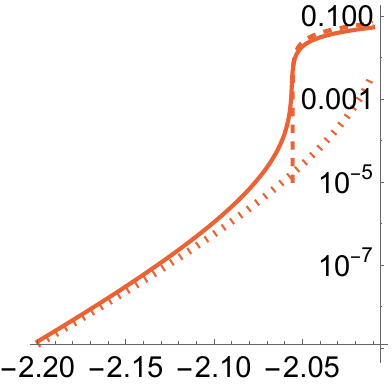}
\includegraphics[width=2.8cm]{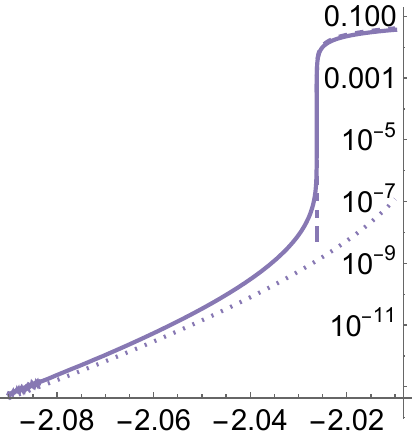}
\end{centering}
\caption{\justifying  The edge of the spectrum for $q=0$. In each figure, the solid line is the exact result, the dashed line is the large $d$ perturbative solution with square-root law, and the dotted line is the non-perturbative asymptotic solution with exponential suppresion. The figures with blue, yellow, green, red and purple curves are for $d=2,5,10,20,40$ respectively. \label{fig:4}}
\end{figure}

For generic $d$, we can numerically solve \eqref{80-eq} for $w_r$ and plug into \eqref{82-eq} and \eqref{83-eq}. Surprisingly, we numerically find that the large $d$ approximation \eqref{85-eq} and \eqref{86-eq} for \eqref{83-eq} work quite well for the exponentially small tail for both small and large $d$. In Fig. \ref{fig:4} we combine the perturbative solution and the non-perturbative approximation and compare them with the exact result. For small $d$, it is obvious that the large $d$ perturbative solution breaks down, but for $d\sim O(10)$ the square-root law becomes visible and close to \eqref{76-eq}. On the other hand, for all $d$ the asymptotic behavior is well captured by the non-perturbative effect.

From Fig. \ref{fig:4}, we see that the transition between these two behaviors is not quite sharp especially when $d$ is not large enough. To understand the interpolation, one needs to resum a special sequence including perturbative and non-perturbative effects simultaneously. This is beyond the scope of this work and we hope to understand it better in the future. Nevertheless, we can still define a transition point when these two approximations intersect each other, roughly at $y\app -1/d$. At this point, we can ignore the quadratic term in \eqref{87-eq} and the spectrum exponentially decays. We can estimate the decay rate as the inverse critical temperature $\b_c$ for the transition between the two effects, namely
\be 
T_c=1/\b_c=1/(2d)
\ee

This critical temperature is not in the thermodynamical limit sense because there is no sharp transition. We postulate the physical interpretation as follows. Prepare a thermal state with temperature $T$ of the dcSYK model with $q=0$, if $T>T_c$ we can approximately regard the system as a chaotic system because the spectrum is almost compact and obeys the square-root law; if $T<T_c$ we should not regard it as a chaotic system and its dynamics will be very different. To justify this postulation, we need to check correlation functions, which is for a future investigation. Moreover, we will have some discussions on the thermodynamical critical temperature in Sec. \ref{sec:5}.

\section{Coarse-grained path integral for $q$ close to 1} \label{sec:coarse}

In this section, we will study the spectrum, and especially the part near the edge for $q$ close to 1. This limit is specially interesting because for double-scaled SYK model, if we zoom into the edge and take $q\ra 1$, we will find the spectrum has the form of $\sinh C\sqrt{E-E_0}$, which is consistent with the Schwarzian derivative as the IR effective action. For convenience, let us call this limit as the Schwarzian limit from now. This gives a strong evidence that the low temperature double-scaled SYK model with $q\ra 1$ has a dual of JT gravity.

As explained in Section \ref{sec:2C}, our dcSYK model becomes equivalent to the double-scaled SYK model in $d\ra \infty$ limit. It is natural to ask how the finite $d$ effect changes the spectrum in the Schwarzian limit. For a finite $d$, each moment of $H$ for $q> 0$ must be greater than that for $q=0$ because each moment is the sum over many chord diagrams and $q>0$ has more diagrams than $q=0$ because the latter disallows intersection between different chords. Since diagrams are evaluated as positive numbers, the $n$-th moment of $q>0$ must grow faster than that of $q=0$, which is already faster than exponentially. By Lemma \ref{lemma3}, we immediately conclude that any finite $d$ dcSYK model must have non-compact spectrum. Therefore, our focus will be understanding how the finite $d$ effect gives the interplay between the chaotic bounded spectrum of $\sinh C \sqrt{E-E_0}$ and the unbounded tail.

\subsection{Coarse-grained path integral}

Though the exact solution for dcSYK is unavailable, for $q$ close to 1 there is a nice approximate method developed in \cite{Berkooz:2024evs,Berkooz:2024ofm} to deal with 2 types of chords in the double-scaled SYK model. The method is a coarse-grained path integral representation for chord diagrams, and we will slightly improve the form and generalize it to chord diagrams with $d$ colors. We will briefly review the construction of the coarse-grained path integral representation without giving derivation details. Interesting readers may refer to their original papers \cite{Berkooz:2024evs,Berkooz:2024ofm}.

\begin{figure}
\begin{centering}
\subfloat[\label{5a}]{\begin{centering}
\includegraphics[height=3cm]{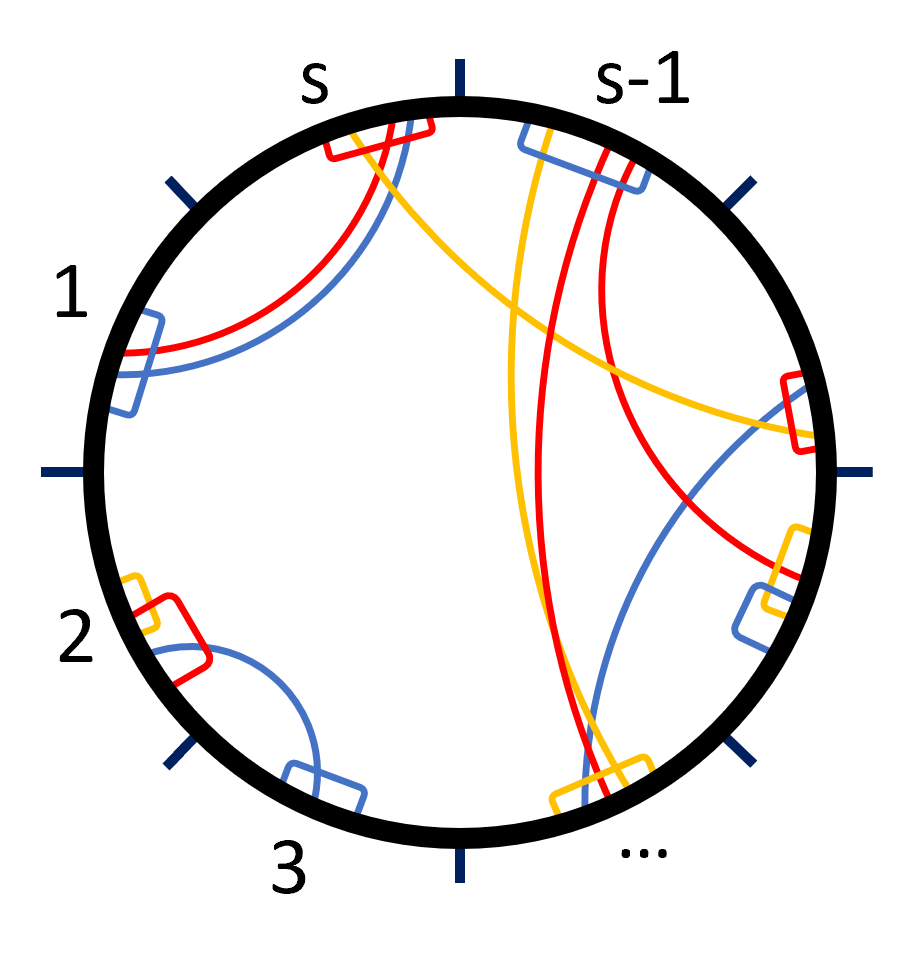}\end{centering}}
\hfill\subfloat[\label{5b}]{\begin{centering}
\includegraphics[height=3cm]{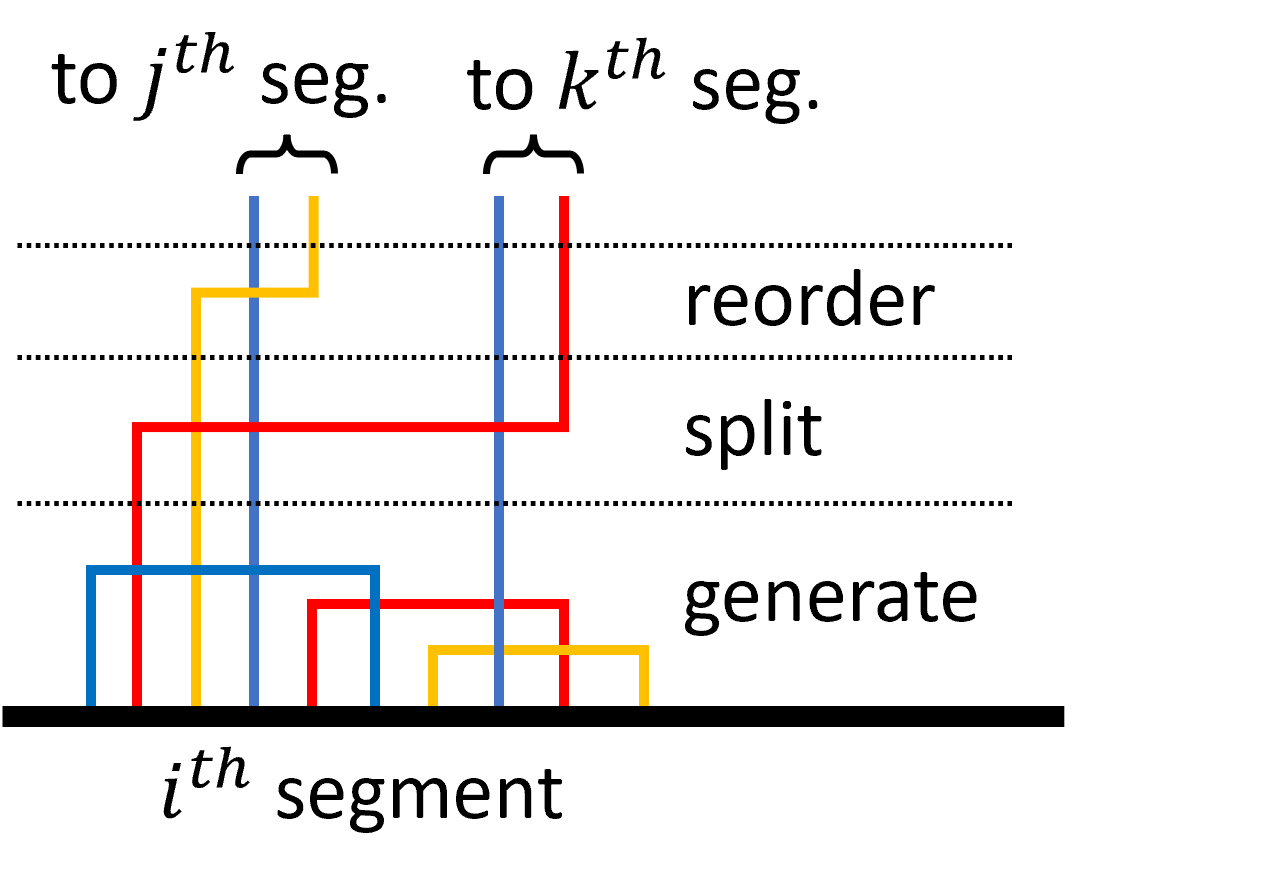}\end{centering}}
\par\end{centering}
\caption{\justifying  (a) The boundary of a $d=3$ color chord diagram is split into $s$ segments. (b) For the open chords from the $i$-th segment, there are three steps before landing on the $j$-th and $k$-th segments: generate, split and reorder.}
\end{figure}

For each chord diagram in the $2n$-moment $M_{2n}$, we can separate the boundary circle into $s$ segments and put $2n$ ends of chord into these $s$ segments (see Fig. \ref{5a}). Then the intersection counting of the chord diagrams include two types: the intersection between the chords emanating from the same segment, and the intersection between the chords emanating from different segments. For the low temperature physcis that is dominated by high moments, we will take $s$ large but not too large such that each segment also has plenty of chord terminals. Due to this scale separation, the intersection of the first type reflects the short-time physics and the intersection of the second type reflects the long-time physics that is relavant to IR.

For the $i$-th segment, assume there are $n^a_i$ color $a$ open chords emanated to intersect with chords from other segments. Among these $n^a_i$ open chords, we define $n^a_{ij}$ as the number of color $a$ open chords that will end on the $j$-th segment for $j\neq i$. By definition, we have 
\be 
\sum_{j\neq i}n^a_{ij}=n^a_i,\quad n^a_{ij}=n^a_{ji} \label{89-eq}
\ee
By the analysis of \cite{Berkooz:2024evs,Berkooz:2024ofm}, there are three steps to emanate $n^a_i$ color $a$ chords from the $i$-th segments that involve the first type of intersection. By the assumption of scale separation, there are plenty of chord terminals on the $i$-th segment. The first step is to generate $n^a_i$ color $a$ open chords emanating from these terminals, and the rest terminals connect to each other to form chords on the same segment. The second step is to split the $n^a_i$ color $a$ open chords into groups labeled by $j$ and each group has $n^a_{ij}$ color $a$ open chords. The third step is to reorder the $n^a_{ij}$ open chords for all $a$ because their self intersections need to be counted. See Fig. \ref{5b} for an illustration of these three steps.

For the second type of intersection, the counting is simple. For two pairs of segments $i,j$ and $k,l$, they will have chord intersection only when $i>k>j>l$ or $k>i>l>j$ (due to symmetry \eqref{89-eq}, we assume $i>j$ and $k>l$). These two cases need to be counted only once. The intersection weight is 1 if $a=b$ and it is $q$ if $a\neq b$. Therefore, we will have the following counting for the second type of intersection
\be  
q^{\sum_{a\neq b}\sum_{i>k>j>l}n_{ij}^{a}n_{kl}^{b}}
\ee

Summing over all moments in \eqref{7-Z}, we will have an expression of partition function in the form of summing over all configurations of $\{n^a_{ij}\}$ for a fixed number $s$ of segments. Let us denote $\ket{{\bf a}}$ to represent the state \eqref{31-eq} with $n=|{\bf a}|$ the length of the string ${\bf a}$. By above procedure, the partition function can be written as \footnote{The normalization $\avg{{\bf a}_{i}|{\bf a}_{i}}$ in the denominator is due to the Hermitian construction of $H$. If we use the transfer matrix $T$ from \cite{Berkooz:2018qkz,Berkooz:2018jqr} in stead of $H$ in \eqref{91-eq}, we do not need to divide this factor.}
\begin{align}
Z(\b)&=\sum_{\{n_{ij}^{a}\}}q^{\sum_{a\neq b}\sum_{i>k>j>l}n_{ij}^{a}n_{kl}^{b}} \nn\\
&\times\sum_{\{{\bf a}_i\}}\prod_{i=1}^{s}\left(\mC(n_{ij}^{a};{\bf a}_{i};q)\f{\avg{{\bf a}_{i}|e^{-\b H/s}|0}}{\avg{{\bf a}_{i}|{\bf a}_{i}}\dim \mH_{\{n^a_i\}}} \right)\label{91-eq}
\end{align}
where $\dim \mH_{\{n^a_i\}}$ is the dimension of the subspace of with fixed $\{n^a_i\}$ given by \eqref{35-dim}, and the sum over ${\bf a}_{i}$ obeys the constraint
\be 
n_{i}=|{\bf a}_{i}|=\sum_a n^a_i=\sum_a \sum_{j\neq i}n^a_{ij}
\ee
In \eqref{91-eq}, we split $\b$ Euclidean time into $s$ segments with length $\b/s$. The expectation value $\avg{{\bf a}_{i}|e^{-\b H/s}|0}$ represents the first step of generating $\{n^a_i\}$ open chords. The function $\mC(n_{ij}^{a}$;${\bf a}_{i};q)$ depending on $n_{ij}^{a}$ with $j\neq i$ and  includes the contribution from splitting and reordering of $n_{i}$ open chords.

To solve $\mC$ and $\avg{{\bf a}_{i}|e^{-\b H/s}|0}$ in closed form is very hard. Nevertheless, for $q$ close to 1 and Schwarzian limit, as argued in \cite{Berkooz:2024ofm}, we can replace them with the known result for $q=1$ but keep $q$ nontrivial in the first line of \eqref{91-eq}. At $q=1$ the algebra of \eqref{31-cr} becomes $d$ independent harmonic oscillators, and all states $\ket{{\bf a} _i}$ with the same $\{n^a_i\}$ is equivalent to the number state $\ket{\{n^a_i\}}=\prod_a (\a_a^\dag)^{n^a_i}\ket{0}$. 

Summing over $\{{\bf a}_i\}$ for a fixed $\{n^a_i\}$ gives the dimension of the subspace $\dim \mH_{\{n^a_i\}}$, which cancels the same term in the denominator of \eqref{91-eq}. Given $\{n_{ij}^{a}\}$ and $\{n^a_i\}$, there are $\prod_a (n^a_{i}!/\prod_{j\neq i}(n_{ij}^{a})!)$ ways to split the $n_i$ open chords into $(s-1)d$ groups. Moreover, for fixed $a,i,j$, the $n_{ij}^{a}$ color $a$ open chords can self intersect and give a contribution $\prod_{i>j}\prod_{a}(n_{ij}^{a})!$. This is the reordering process. Lastly, for $d$ harmonic oscillators, we can easily compute
\begin{align} 
\f{\avg{\{n_{i}^{a}\}|e^{-\b H/s}|0}}{\avg{\{n_{i}^{a}\}|\{n_{i}^{a}\}}}&=\prod_a \f{\avg{n^a_i|e^{-\b (\a^\dag_a+\a_a)/(\sqrt{d}s)}|0}}{\avg{n^a_i|n^a_i}}  \nn\\
&=\f {e^{\b^{2}/(2s^2)}}{\prod_{a}n^{a}_i!}\left(-\f{\b}{\sqrt{d}s}\right)^{\sum_{a}n^{a}_i}
\end{align}
Putting everything together, the partition function becomes 
\begin{equation}
Z(\b)=\sum_{\{n_{ij}^{a}\}}q^{\sum_{a\neq b}\sum_{i>k>j>l}n_{ij}^{a}n_{kl}^{b}}  \f{e^{\b^{2}/(2s)}}{\prod_{i>j}\prod_{a}(n_{ij}^{a})!}  \f{\b^{2n}}{(ds^{2})^{n}} 
\end{equation}
where we define
\begin{equation}
n=\sum_{a=1}^{d}\sum_{i>j}n_{ij}^{a} \label{96-eq}
\end{equation}

We can rewrite the quadratic form of $n_{ij}^{a}$ in the exponent
of $q$ as linear by a Gaussian integral. Let us denote 
\begin{equation}
\sum_{a\neq b}\sum_{i>k>j>l}n_{ij}^{a}n_{kl}^{b}=\f 12 \bfn\cdot M\cdot \bfn\label{eq:8}
\end{equation}
where $\bfn$ is naively a $d\times s(s-1)/2$ dimensional vector with
components labeled by $(a,i,j)$ and $M$ is a square matrix. However,
as (\ref{eq:8}) does not include any term of $n_{ij}^{a}$ with $i=j+1$
because they do not contribute to any crossing. In other words, we
should treat $M$ as a $d\times s(s-3)/2$ dimensional matrix with
all $i=j+1$ rows and columns stripped off. This reduced $M$ matrix
is invertible with real eigenvalues. Using \eqref{15-q}, we can write
\begin{align}
&q^{\sum_{a\neq b}\sum_{i>k>j>l}n_{ij}^{a}n_{kl}^{b}}=e^{-\f 12\lambda \bfn\cdot M\cdot \bfn}\nn\\
=&\f 1{\sqrt{\det(-2\pi\lambda M)}}\int d \bfJ e^{\bfn\cdot \bfJ+\f 1{2\lambda}\bfJ\cdot M^{-1}\cdot \bfJ}\label{eq:9}
\end{align}
where $\bfJ$ is a $d\times s(s-3)/2$ dimensional vector. Note that
$M$ has both positive and negative eigenvalues, to define the $\bfJ$
integral, we need to take the diagonal basis of $M$ and integrate
the corresponding eigenvariable $\bfJ$ along real or imaginary axis
for positive or negative eigenvalues. 

Using (\ref{eq:9}) and notation \eqref{96-eq}, the sum over $n_{ij}^{a}$ leads to
\begin{equation}
Z(\b)=\int d\bfJ\f{e^{\f 1{2\lambda}\bfJ\cdot M^{-1}\cdot \bfJ}}{\sqrt{\det(-2\pi\lambda M)}}e^{\left(\f 1s+\f 2{s^{2}d}\sum_{i>j;a}e^{J_{ij}^{a}}\right)\b^{2}/2} \label{99-eq}
\end{equation}
We are interested in low temperature and large $\b=\b_{r}/\sqrt{\lambda}=1/(\sqrt{\lambda}T_{r})$,
which gives the following effective action
\begin{equation}
S=\f 12 \bfJ\cdot M^{-1}\cdot \bfJ+\f{\b_{r}^{2}}{s^{2}d}\sum_{i>j;a}e^{J_{ij}^{a}}\label{eq:30}
\end{equation}
and the partition function is $Z(\b)\sim\int d\bfJ e^{S/\lambda}$. In $q\ra 1 $ limit, we have $\lambda \sim 1-q\ll 1$. Therefore, the partition function can be computed by saddle approximation. Variation
of $\bfJ$ leads to
\begin{equation}
(M^{-1}\cdot J)_{ij}^{a}=-\f{\b_{r}^{2}}{s^{2}d}e^{J_{ij}^{a}}
\end{equation}
which can be solved as
\begin{equation}
    J_{ij}^{a}=-\f{\b_{r}^{2}}{s^{2}d}\sum_{b\neq a}\left(\sum_{i>k>j>l}+\sum_{k>i>l>j}\right)e^{J_{kl}^{b}}\label{eq:8-1}
\end{equation}

To avoid solving such discrete matrix equation, we will take large $s$ limit and define Euclidean time $\tau_i=\b_r i/s$. In this way, we can promote the vector $\bfJ$ as a bilocal continuous function $J_{ij}^{a}=J^{a}(\tau_{i},\tau_{j})$. In this continuum limit, \eqref{eq:8-1} becomes
\begin{equation}
J^{a}(\tau_{1},\tau_{2})=-\f 1d\sum_{b\neq a}\iint_\Delta d\tau d\tau'e^{J^{b}(\tau,\tau')}
\end{equation}
where $\D$ means the ranges $\tau_1>\tau>\tau_2>\tau'>0$ and $\b_r>\tau>\tau_1>\tau'>\tau_2$. Taking derivatives on both sides, we will have a Liouville-like equation
\begin{equation}
\del_{1}\del_{2}J^{a}=-\f 2d\sum_{b\neq a}e^{J^{b}(\tau_{1},\tau_{2})} \label{eom}
\end{equation}
To find the boundary condition for this equation, note that in (\ref{eq:8-1})
when $i=j+1$ or $i=s,j=1$ the RHS has no term in the sum, which
in continuum limit leads to
\begin{equation}
J^{a}(\tau,\tau)=J^{a}(\b_r,0)=0
\end{equation}
Indeed, at the action level, we can take the continuum limit and write
\begin{equation}
M^{-1}=\f{\b_{r}^{2}}{2s^{2}}K^{-1}\del_{1}\del_{2}
\end{equation}
where $K$ is $d$-dimensional matrix $K_{ab}=1-\d_{ab}$. Therefore,
the action in continuum limit is
\begin{align}
S&=\int_{0}^{\b_{r}}d\tau\int_{0}^{\tau}d\tau' \mL[J^a(\tau,\tau')] \label{eq:40}\\
\mL&=\f 14J^{a}(\tau,\tau')K_{ab}^{-1}\del_{\tau}\del_{\tau'}J^{b}(\tau,\tau')+\f 1d\sum_{a}e^{J^{a}(\tau,\tau')}
\end{align}

\subsection{Homogeneous saddle}

To solve the equation of motion \eqref{eom}, we can assume homogeneity $J^{a}=J$ and translation
invariance, which leads to the SYK saddle
\begin{align}
e^{J}&=\f{\cos^{2}\w\b_{r}/2}{\cos^{2}\w(\tau_{12}-\b_{r}/2)}\\
\w&=\sqrt{1-1/d}\cos\w\b_{r}/2
\end{align}
Taking this into the action (\ref{eq:30}), we have the on-shell action
\begin{align}
S_{0} & =\int_{0}^{\b_{r}}d\tau\int_{0}^{\tau}d\tau'(1-J(\tau,\tau')/2)e^{J(\tau,\tau')}\nonumber \\
 & =\b_{r}^{2}(\f{\sin x}x-\f 12\cos^{2}\f x2)
\end{align}
where
\be 
x\equiv\w\b_{r},\quad x=\sqrt{1-1/d}\b_{r}\cos x/2
\ee
For
latter convenience, we define $y=\pi-x$ and temperature is expressed
in terms of $y$ as
\begin{equation}
T_{r}=1/\b_{r}=\f{\sqrt{1-1/d}\sin(y/2)}{\pi-y}\label{eq:t-y}
\end{equation}
For large $\b_{r}$, $x$ takes value very close to $\pi$ and $y$
can be solved in $1/\b_{r}$ expansion as 
\begin{equation}
y=\f{2\pi}{\sqrt{1-1/d}\b_{r}}-\f{4\pi^{2}}{(\sqrt{1-1/d}\b_{r})^{2}}+\cdots\label{eq:36}
\end{equation}
It follows that
\begin{equation}
S_{0}=\f {2\b_{r}}{\sqrt{1-1/d}}-\f{\pi^{2}}{2(1-1/d)}+\f{\pi^{2}}{(1-1/d)^{\f 3 2}\b_{r}}
\end{equation}
Then the on-shell partition function is
\begin{equation}
Z\sim e^{\f 1{\lambda}S_{0}}\sim\exp\left[\f{2\b}{\sqrt{\lambda(1-1/d)}}+\f{\pi^{2}}{[\lambda(1-1/d)]^{\f 3 2}\b}\right]\label{eq:38-1}
\end{equation}
where the constant term is omitted. The first term gives the ground energy $E_{0}=-2/\sqrt{\lambda(1-1/d)}$. This is consistent with \eqref{21-h} with $\t=\pi$ and $q\ra \bar q$ by \eqref{39-eq} in the $1-q\ra \lambda\ll 1$ limit. Note that \eqref{39-eq} originally requires large $d$, but its validity extends to generic $d$ in the $q\ra 1$ limit. The second term by inverse Laplace transformation leads to exponential growth $e^{C\sqrt{(E-E_0)/\lambda^{3/2}}}$ of the spectrum with $C=2\pi/(1-1/d)^{3/4}$ for $(E-E_0)/\lambda^{3/2} \gg 1$.

\subsection{The quadratic action}

Though the saddle of dcSYK model is very similar to the SYK model, the 1-loop order has essential difference. The quadratic expansion of (\ref{eq:40}) around the saddle $J^{a}=J$
is
\begin{align}
S_{2}=&\int_{0}^{\b_{r}}d\tau\int_{0}^{\tau}d\tau'\left[\f 14\d J^{a}(\tau,\tau')K_{ab}^{-1}\del_{\tau}\del_{\tau'}\d J^{b}(\tau,\tau') \right.\nn\\
&\left. +\f 1{2d}e^{J(\tau-\tau')}\sum_{a}\d J^{a}(\tau,\tau')\d J^{a}(\tau,\tau')\right]
\end{align}
Since $K$ is a symmetric real matrix, it can be diagonalized by an
orthonormal matrix. The eigenvalues of $K$ include one $d-1$ and $d-1$ numbers of $-1$ with
multiplicity $d-1$. Along the diagonal basis of $\d J^{a}$, we can
split it into $\d J$ and $\d J^{i}$ and rewrite the quadratic action
as
\begin{widetext}
    \begin{align}
S_{2}= & \int_{0}^{\b_{r}}d\tau\int_{0}^{\tau}d\tau'\left[\f 1{4(d-1)}\d J(\tau,\tau')\del_{\tau}\del_{\tau'}\d J(\tau,\tau')+\f 1{2d}e^{J(\tau-\tau')}\d J(\tau,\tau')\d J(\tau,\tau')\right.\nonumber \\
 & \left.+\sum_{i=1}^{d-1}\left(-\f 14\d J^{i}(\tau,\tau')\del_{\tau}\del_{\tau'}\d J^{i}(\tau,\tau')+\f 1{2d}e^{J(\tau-\tau')}\d J^{i}(\tau,\tau')\d J^{i}(\tau,\tau')\right)\right] \label{118-S2}
\end{align}
\end{widetext}

For $\d J$ and $\d J^{i}$, they both obey the same type of eigen
equation
\begin{equation}
\left[\del_{1}\del_{2}+\f{h(h-1)\w^{2}}{\cos^{2}\w(\tau_{12}-\b_{r}/2)}\right]f(\tau_{1},\tau_{2})=\eta f(\tau_{1},\tau_{2})\label{eq:47}
\end{equation}
Define new variables  $\t=\pi/2-\w(\tau_{12}-\b_{r}/2)$ and
$\tau=(\tau_{1}+\tau_{2})/2$. We can rewrite the above equation
as 
\begin{equation}
\left[\f 14\del_{\tau}^{2}-\w^{2}\del_{\t}^{2}+\f{h(h-1)\w^{2}}{\sin^{2}\t}\right]f(\tau,\t)=\eta f(\tau,\t)
\end{equation}
For $\d J$, the parameter $h$ is 
\[
h(h-1)\w^{2}=2(1-1/d)\cos^{2}\w\b_{r}/2\implies h=2
\]
For $\d J^{i}$, the parameter $h$ is 
\begin{align}
&h(h-1)\w^{2}=-2/d\cos^{2}\w\b_{r}/2 \nn\\
&\implies h=\f 12\left(1\pm\sqrt{\f{d-9}{d-1}}\right)\label{eq:45-1}
\end{align}

Assuming $f(\tau,\t)=e^{ik\tau}f(\t)$, we can solve the eigen equation \eqref{eq:47}
as
\begin{align}
f_{\pm}(\t)&=\sqrt{\sin\t}\left[P_{(n_{\eta}/\w-1)/2}^{1/2-h}(\cos\t)\pm P_{(n_{\eta}/\w-1)/2}^{1/2-h}(-\cos\t)\right]\nn\\
\quad n_{\eta}&=\sqrt{k^{2}+4\eta}\label{eq:46}
\end{align}
where $\pm$ means even/odd under $\t\ra\pi-\t$ (which is equivalent
to $\tau_{12}\ra\b-\tau_{12}$). For $\t$ variable, the boundary
condition is 
\be 
f_{\pm}(\tau,\t_{r})=0,\quad \t_{r}=(\pi-\w\b_{r})/2
\ee
This leads to nontrivial discrete solutions of $n_{\eta}$ obeying
\begin{equation}
\mM_{\pm}=\{n_{\eta}|P_{(n_{\eta}/\w-1)/2}^{1/2-h}(\cos\t_{r})\pm P_{(n_{\eta}/\w-1)/2}^{1/2-h}(-\cos\t_{r})=0\}\label{eq:47-1}
\end{equation}
Note that $n_{\eta}\ra-n_{\eta}$ gives the same solution because
of the symmetry $P_{\nu}^{\mu}(z)=P_{-1-\nu}^{\mu}(z)$. The independent
solutions are labeled by the eigenvalues $\mM_{\pm}/\Z_{2}$, where
$\Z_{2}$ refers to this reflection symmetry. On the other hand, since $\tau$ labels the segments on a circle, we need to impose some periodic boundary condition for $\tau$. In $(\tau_{1},\tau_{2})$
coordinate, we should require $f(\b_{r},\tau_{2})=f(\tau_{2},0)$,
which in $(\tau,\t)$ coordinate means
\begin{equation}
f(\tau,\t)=f(\tau+\b_{r}/2,\pi-\t)
\end{equation}
It follows that we need choose $k=2\pi\Z_{\pm}/\b_{r}$ for $f_{\pm}$
respectively, where $\Z_{\pm}$ means even/odd integers.

Before computing the 1-loop determinant, let us first check the determinant
of the free theory, namely $\det M^{-1}$. For this piece the continuous
limit leads to computing $\det(\f 14\del_{\tau}^{2}-\del_{z}^{2})$
where $z=\tau_{12}$. The eigenfunctions are simple
\begin{equation}
f(\tau,z)=e^{ik\tau}\sin(mz\pi/\b_{r})
\end{equation}
where the boundary condition restricts $m\in\Z_{\mp}^{>0}$ and $k\in2\pi\Z_{\pm}/\b_{r}$
for even/odd modes under $z\ra\b_{r}-z$. The eigenvalues are 
\begin{equation}
\eta_{0}=\pi^{2}((\Z_{\mp}^{>0})^{2}-(\Z_{\pm})^{2})/\b_{r}^{2}
\end{equation}
Here we see that each eigenvalue is proportional to $1/\b_{r}^{2}$, which is reasonable by dimensional analysis.
Recall the $\sqrt{\det(-2\pi \lam M)}$ term in \eqref{99-eq}, the total 1-loop contribution is given by the product of the ratio $\eta/\eta_{0}$
\begin{equation}
Z_{\text{1-loop}}=\f 1{\sqrt{\prod(\eta/\eta_{0})}}
\end{equation}
Our purpose is to find eigenvalues $\eta$ with
different scaling of $\b_{r}$. In particular, $\eta$ with $\b_{r}^{-2}$
scaling does not contribute to the $\b_{r}$ dependence at 1-loop.

\subsection{1-loop determinant by Sommerfeld-Watson resummation} \label{sec:4D}

In the following, we will define a simpler notation by using $x$ and
\begin{equation}
n\equiv n_{\eta}\b_{r},\quad\mM_{\pm}\b_{r}\ra\mM_{\pm}
\end{equation}
Let us assume the generic solution of (\ref{eq:47-1}) is 
\begin{equation}
F_{\pm}(n)=0\label{eq:68}
\end{equation}
The 1-loop determinant is related to
\begin{equation}
\prod\f{\eta}{\eta_{0}}=\prod_{k=\Z_{\pm}}\f{\prod_{n\in\mM_\pm/\Z_2}[(n/(2\pi))^{2}-k^{2}]}{\prod_{m=\Z_{\mp}^{>0}}[m^{2}-k^{2}]}
\end{equation}
The infinite product of $k$ can be computed explicitly
\begin{equation}
\prod\f{\eta}{\eta_{0}}=\begin{cases}
\f{\prod_{n\in\mM_{+}/\Z_{2}}\sin^{2}\f n4}{\prod_{m=\Z_{-}^{>0}}\sin^{2}\f{m\pi}2}=\prod_{n\in\mM_{+}/\Z_{2}}\sin^{2}\f n4, & k\in\Z_{+}\\
\f{\prod_{n\in\mM_{-}/\Z_{2}}\cos^{2}\f n4}{\prod_{m=\Z_{+}^{>0}}\cos^{2}\f{m\pi}2}=\prod_{n\in\mM_{-}/\Z_{2}}\cos^{2}\f n4, & k\in\Z_{-}
\end{cases}\label{eq:74}
\end{equation}
where in both cases the denominator is unity for any $m$.

\begin{figure}
\begin{centering}
\includegraphics[width=7cm]{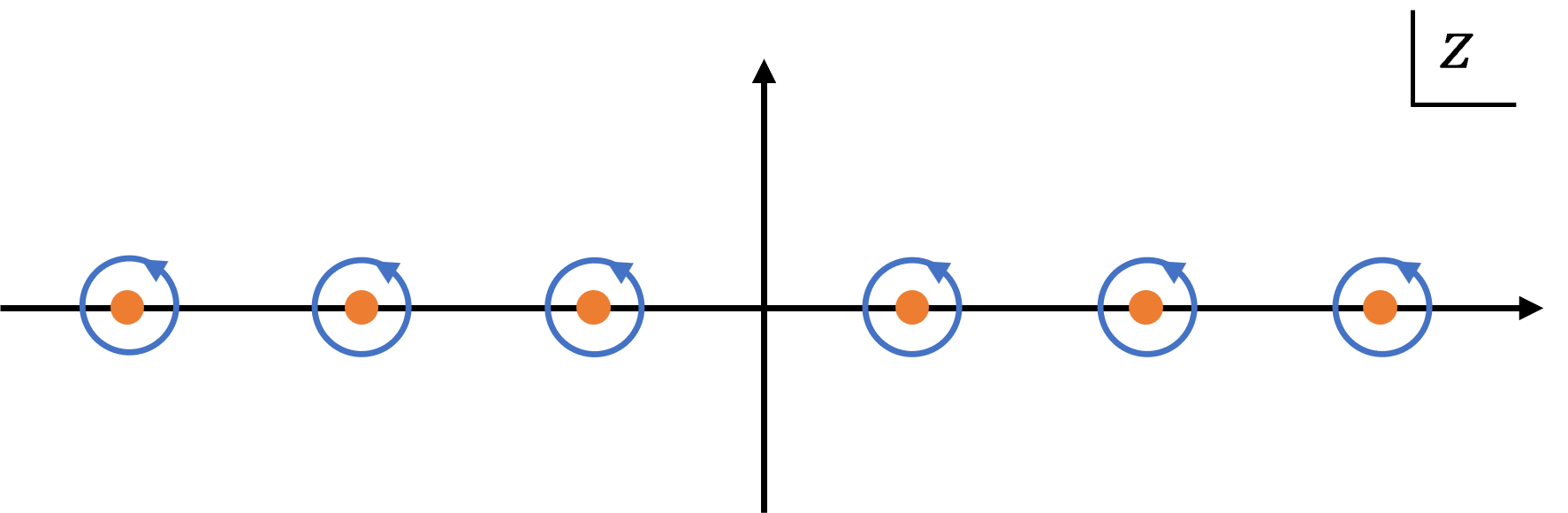}
\end{centering}
\caption{\justifying The contour $\mC_\pm$ in blue is the anticlockwise circles around the points in $\mM_\pm$ (orange dots). \label{fig:SW}}
\end{figure}

To compute
this infinite product, we can instead compute an infinite sum on the
exponent using Sommerfeld-Watson resummation. Given a function
$G_{\pm}(n)$ invariant under $n\ra-n$, we have
\begin{align}
&\prod_{n\in\mM_{\pm}/\Z_{2}}G_{\pm}(n)=\exp\left[\f 12\sum_{n\in\mM_{\pm}}\log G_{\pm}(n)\right]\nn\\
=&\exp\left[\f 1{4\pi i}\int_{\mC_{\pm}}dz\log G_{\pm}(z)\f{F_{\pm}'(z)}{F_{\pm}(z)}\right]\label{eq:75}
\end{align}
where $\mC_{\pm}$ is the anticlockwise contours circling around the
points in $\mM_{\pm}$ (see Fig. \ref{fig:SW}). For convenience, here we also include $n=-(\mM_{\pm}/\Z_{2})$
in the sum with additional factor $1/2$ because they are also solutions
to $F_{\pm}(z)=0$.

\subsubsection{$h=2$ \label{subsec:h=2}}

Let us apply the Sommerfeld-Watson resummation to compute the determinant for $h=2$. At $h=2$, the eigenfunctions $f_\pm(\t)$ are largely simplified  to  trigonometric functions \cite{Choi:2019bmd}
\begin{align}
f_{+}(\t) & =\cot\t\sin\f{n(\f \pi 2-\t)}{2x}+\f{n}{2x}\cos\f{n(\f \pi 2-\t)}{2x}\label{eq:52}\\
f_{-}(\t) & =\cot\t\cos\f{n(\f \pi 2-\t)}{2x}-\f{n}{2x}\sin\f{n(\f \pi 2-\t)}{2x}\label{eq:53}
\end{align}
For $f_{+}$ we need exclude
$n=0$ as it is a trivial solution; similarly, we exclude $n=2x$
for $f_{-}$. In the end, the eigenvalues are 
\begin{align}
\mM_{+} & =\{n\neq 0|\tan\f{x}2\tan\f{n}4+\f{n}{2x}=0\}\label{eq:54}\\
\mM_{-} & =\{|n|\neq 2x|\tan\f{x}2\cot\f{n}4-\f{n}{2x}=0\}\label{eq:55}
\end{align}

From (\ref{eq:54}) and (\ref{eq:55}) we can tentatively choose 
\begin{align}
F_{+}(z)&=\tan\f x2\tan\f z4+\f z{2x}\\
F_{-}(z)&=\tan\f x2\cot\f z4-\f z{2x}
\end{align}
First of all, using $F_{\pm}$, we can replace the sine and cosine
function in (\ref{eq:74}) both by the same rational function
\begin{equation}
\sin^{2}\f n4,\cos^{2}\f n4\ra\f{n^{2}}{n^{2}+(2x\tan\f x2)^{2}}=G_{\pm}(n)\label{eq:77}
\end{equation}
Second, the contour integral of $\log G_{\pm}(z)$ (\ref{eq:75})
may have subtle issue with branch cut. For (\ref{eq:77}), we can
instead using 
\begin{equation}
\del_{\a}\log\f{z^{2}}{\a+z^{2}}=-\f 1{\a+z^{2}}
\end{equation}
to compute the exponent in (\ref{eq:75}) as
\begin{equation}
I_{\pm}=-\f 1{4\pi i}\int_{0}^{4x^{2}\tan^{2}x/2}d\a\int_{\mC_{\pm}}dz\f 1{\a+z^{2}}\f{F_{\pm}'(z)}{F_{\pm}(z)}\label{eq:79}
\end{equation}
Moreover, $F_{\pm}'(z)$ may have additional poles on the complex
plane. For practical purpose of the resummation, we hope the additional
poles to be as few as possible. Since $\mC_{\pm}$circling around
the poles of $F_{\pm}(z)$, we are free to replace part of $F_{\pm}'(z)$
using the equation $F_{\pm}(z)=0$ because this does not change the
residue at the poles. Lastly, we should note that the choice of $F_{\pm}$
is not unique. We can consider any function $g$ that does not share
the zeros with $F_{\pm}$ and define a new $\tilde{F}_{\pm}=F_{\pm}g$,
which leads to 
\begin{equation}
\f{\tilde{F}_{\pm}'(z)}{\tilde{F}_{\pm}(z)}=\f{F_{\pm}'(z)}{F_{\pm}(z)}+\f{g'(z)}{g(z)}\label{eq:48}
\end{equation}
The purpose of $g$ is to kill the possible slowly decaying piece
in $F_{\pm}'(z)/F_{\pm}(z)$ at infinity such that we can deform the
contour $\mC_{\pm}$ to possibly finite numbers of poles on the complex
plane and evaluate the integral (\ref{eq:79}) simply by the residue theorem. 

Keeping these tricks in mind, we find 
\begin{align}
F_{+}'(z) & =-\f{2+x\csc^{2}(z/4)\tan(x/2)}{4x}\\
\implies\f{F_{+}'(z)}{F_{+}(z)}\simeq & \f{4x^{2}+8x\cot(x/2)+z^{2}\cot^{2}(x/2)}{8xz\cot(x/2)-16x^{2}\cot(z/4)}\label{eq:82}
\end{align}
where we replaced $\csc^{2}(z/4)$ with a quadratic polynomial of
$z$ using $F_{+}(z)=0$. Moreover, (\ref{eq:82}) grows linearly
for large $z$, together with $1/(\a+z^{2})$ does not guarantee trivial
integral at infinity. We can add an entire function $g(z)=e^{cz^{2}}$
in (\ref{eq:48}) for some $c$ to kill this linearly growing piece.
It follows that 
\begin{equation}
\f{\tilde{F}_{+}'(z)}{\tilde{F}_{+}(z)}\simeq\f{2x+\cot(x/2)(4-z\tan(z/4))}{4z\cot(x/2)+8x\tan(z/4)}
\end{equation}
For simplicity, from now on we will omit the tilde symbol and redefine $\tilde F_+\ra F_+$. We apply
similar tricks to $F_{-}$ and find
\begin{equation}
\f{F_{-}'(z)}{F_{-}(z)}\simeq\f{2x+\cot(x/2)(4+z\cot(z/4))}{4z\cot(x/2)-8x\cot(z/4)}
\end{equation}

Taking these into (\ref{eq:79}), we deform the contour $\mC_{\pm}$
to just a few poles on the complex plane. For $I_{+}$, the remaining
poles are $z=0,\pm i\sqrt{\a}$; for $I_{-}$, the remaining poles
are $z=\pm2x,\pm i\sqrt{\a}$. Using the residue theorem at these points
and integrate over $\a$, we find that 
\begin{align}
I_{+} & =\log(1+\f x2\tan\f x2)-\f 12x\tan\f x2\\
I_{-} & =-2\log\cos\f x2-\f 12x\tan\f x2
\end{align}
For low temperature, taking $y=\pi-x$ and expanding around $y\sim0$,
we find
\begin{equation}
I_{+}+I_{-}=-\f{2\pi}y+2+\log4\pi-3\log y+\pi y/6+O(y^{2})\label{eq:55-2}
\end{equation}
Since $y\sim1/\b_{r}$ by \eqref{eq:36}, the total 1-loop contribution at low temperature
is 
\begin{align}
Z_{\text{1-loop}}(\b)=e^{-(I_{+}+I_{-})/2}\sim\f{1}{\b^{3/2}}e^{e_1\b+c_1/\b} \\
e_1=\f1 2 \sqrt{(1-1/d)\lam},\quad c_1=-\f{12+\pi^2}{4\sqrt{(1-1/d)\lam}}
\end{align}
where the $1/\b^{3/2}$ is the expected 1-loop contribution consistent with the Schwarzian derivative \cite{Maldacena:2016hyu}, and the
exponents are $\lambda$ order corrections to the on-shell action
(\ref{eq:38-1}). Together with the tree-level result and applying inverse Laplace transformation, the spectrum with the $h=2$ mode is 
\be 
\r_{h=2}(E)\sim \sinh C \sqrt{(E-E_0)/\lambda^{3/2}}
\ee
where 
\begin{align} 
C&=\f{2\pi}{(1-1/d)^{3/4}}(1-\lambda(1-1/d)(12+\pi^2)/8) \\
E_0&=-\f{2}{\sqrt{(1-1/d)\lam}}(1+\lam(1-1/d)/4) \label{156-eq}
\end{align}

There is no surprise that we get the same spectrum from Schwarzian derivative if we only include the 1-loop contribution from the $h=2$ mode. Indeed, the saddle and the first line of the quadratic action in \eqref{118-S2} is essentially the same as the double-scaled SYK model up to $1/d$ corrections to some unimportant parameters.

\subsubsection{Approximate eigenvalues for $h\protect\neq2$}

Recall that there are $d-1$ modes in the quadratic action \eqref{118-S2} 
 with $h\neq 2$. They turn out to have a much different 1-loop contribution. However, the difficulty of applying Sommerfeld-Watson resummation is that
we do not have a simple expression for the eigen function and the
eigenvalues. Nevertheless, for large enough eigenvalues, we may be
able to find an approximate expression that captures the asymptotic
behavior in the sum over infinite eigenvalues. Let us denote $\tilde{\mM}_\pm$ as the set of approximate eigenvalues. The trick is that we can split
the sum over eigenvalues into two parts. We can set up a subset $K_{\pm}\subset\mM_{\pm}$
of eigenvalues. For eigenvalues $n\in K_{\pm}$ we use exact numeric
solution, and for $n\notin K_{\pm}$ we use the approximation of large
eigenvalues in $\tilde{\mM}_\pm$ (see Fig. \ref{fig:An-illustration-of}). For those eigenvalues in $\tilde{\mM}_\pm$ not used in this approxiamtion, we denote them as $\tilde{K}_\pm$. In this treatment,
the error is suppressed in the large size of $K_{\pm}$ limit.

\begin{figure}
\begin{centering}
\includegraphics[height=3cm]{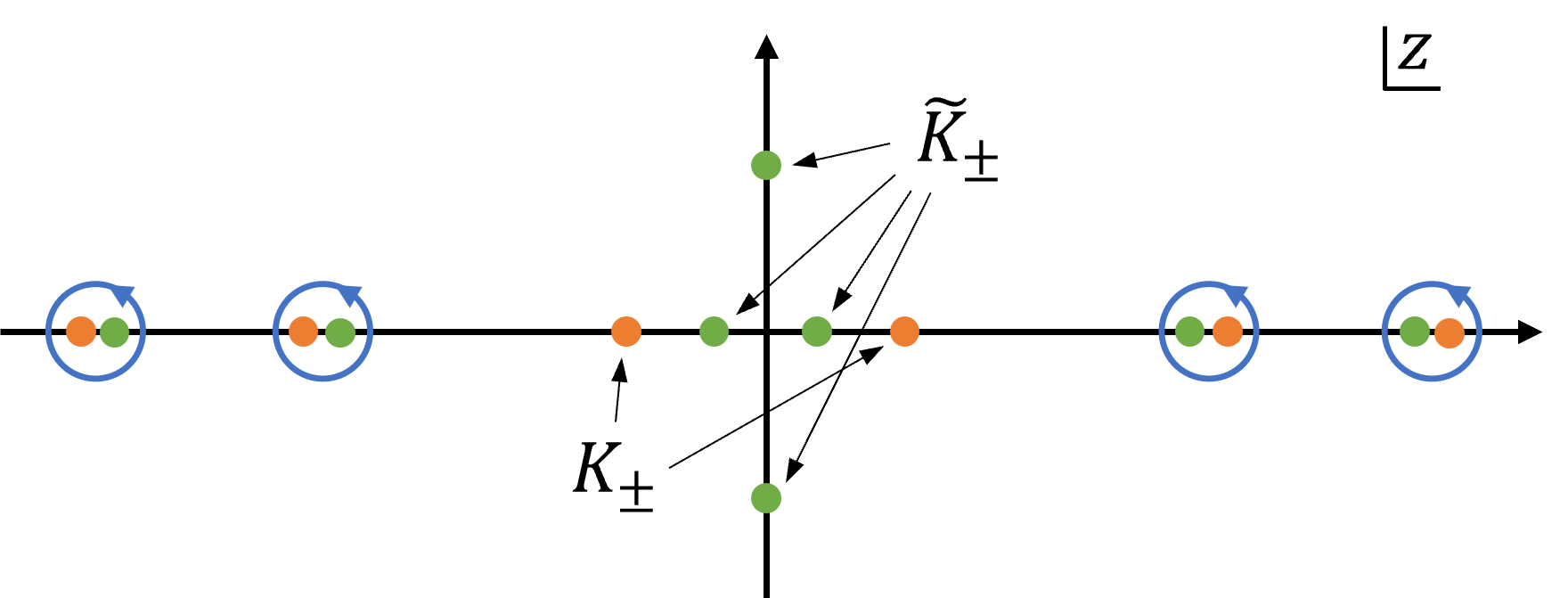}
\par\end{centering}
\caption{\justifying  An illustration of the contour choice of $\protect\mC_{\pm}$ (blue
circles) and the subset $K_{\pm}$ and $\tilde K_{\pm}$. The orange dots
are the exact eigenvalues $\protect\mM_{\pm}$, and the green dots
are the approximate eigenvalues $\tilde{\mM}_{\pm}$. When
the approximate eigenvalues are very close to the exact ones, we choose
$\protect\mC_{\pm}$ to circle around them and use $\tilde{\protect\mM}_{\pm}$
to replace $\protect\mM_{\pm}$. When the approximation is bad near
the origin, we take the exact ones in $K_{\pm}$ and ignore the approximate
ones in $\tilde K_{\pm}$. As they obey quite different equations near the
origin, the number of points in $K_{\pm}$ and $\tilde K_{\pm}$ can be
different. \protect\label{fig:An-illustration-of}}
\end{figure}

The large parameter approximation of associated Legendre function
is given in \cite{olver1997asymptotics}, from which we
find for large $n$, the eigenfunctions (\ref{eq:46}) are approximately
\begin{align}
\tilde{f}_{+}(\t) & \sim\f nx\cos\f {n(\f \pi 2-\t)}{2x}+\g\cot\t\sin\f {n(\f \pi 2-\t)}{2x}\label{eq:92}\\
\tilde{f}_{-}(\t) & \sim\f nx\sin\f {n(\f \pi 2-\t)}{2x}-\g\cot\t\cos\f {n(\f \pi 2-\t)}{2x}\label{eq:93} \\
\g&\equiv h(h-1)=-2/(d-1)<0
\end{align}
It is remarkable that if we take $h=2$, the above approximate eigenfunctions become exact and identical to (\ref{eq:52}) and (\ref{eq:53}). 

Taking $\t=\t_{r}$ in (\ref{eq:92}) and (\ref{eq:93}) and requiring
$\tilde{f}_{\pm}=0$ leads to discrete eigenvalues. Unlike the $h=2$
case, where we only have real eigenvalues, here we may have pure imaginary
eigenvalues because $\g<0$. For $\tilde{f}_{+}$,
the defining equation is 
\begin{equation}
\tan\f x2\tan\f n4+\f n{\g x}=0
\end{equation}
Taking $n\ra in$, the equation becomes 
\begin{equation}
\tanh\f n4=\f n{(-\g)x\tan\f x2}
\end{equation}
which has two nonzero solutions (with opposite sign) when 
\begin{equation}
x\tan\f x2>\f 4{-\g}=2(d-1)\iff\b_{r}>\b_{0}\label{eq:97}
\end{equation}
for a some temperature $T_{0}=1/\b_{0}$. For $\tilde{f}_{-}$,
the case is similar that with $n\ra in$ we have
\begin{equation}
\coth\f n4=\f n{(-\g)x\tan\f x2}\label{eq:98}
\end{equation}
which always have two opposite solutions as long as $\g<0$. 

Let us
define the approximate eigenvalues $\tilde{\mM}_{\pm}$ as
\begin{align}
\tilde{\mM}_{+} & =\{n\neq0|\tan\f x2\tan\f n4+\f n{\g x}=0\}\label{eq:54-1}\\
\tilde{\mM}_{-} & =\{n|\tan\f x2\cot\f n4-\f n{\g x}=0\}\label{eq:55-1}
\end{align}
where $n=0$ is excluded for $\tilde{\mM}_{+}$ because it corresponds
to a trivial $\tilde{f}_{+}$. On the other hand, for $\tilde{\mM}_{-}$
we do not need to exclude a trivial eigenvalue like $n=2x$ as in
$h=2$ case because there is no $n$ leading to trivial $\tilde{f}_{-}$.
Therefore, $n=2x$ can regarded as an emergent mode for $h\ra2$.
Even though we find two different types of $h$ for $d>9$ and $1<d<9$,
the asymptotic eigenvalues $\tilde{\mM}_{\pm}$ only depend on $\g$,
which is smooth across $d=9$. Another feature of $\tilde{\mM}_{\pm}$
is that the defining equation holds for $n\ra-n$ just like the $h=2$
case. This is also the symmetry of full eigenfunction $f_{\pm}$,
which means that we can double count the positive and negative eigenvalues
in the Sommerfeld-Watson resummation with an additional $1/2$ factor
as before.

The computation with the approximate eigenvalues is very similar to
Section \ref{subsec:h=2}. We define
\begin{align}
F_{+}(z)&=\tan\f x2\tan\f z4+\f z{\g x} \label{eq:102-1}\\
F_{-}(z)&=\tan\f x2\cot\f z4-\f z{\g x} \label{eq:102-1b}
\end{align}
Using $F_{\pm}$, we can replace the sine and cosine function in (\ref{eq:74})
both by the same rational function
\begin{equation}
\sin^{2}\f n4,\cos^{2}\f n4\ra\f{n^{2}}{n^{2}+(\g x\tan\f x2)^{2}}=G_{\pm}(n)\label{eq:77-1}
\end{equation}
It follows that the exponent in (\ref{eq:75}) is
\begin{equation}
I_{\pm}=-\f 1{4\pi i}\int_{0}^{\g^{2}x^{2}\tan^{2}x/2}d\a\int_{\mC_{\pm}}dz\f 1{\a+z^{2}}\f{F_{\pm}'(z)}{F_{\pm}(z)}\label{eq:79-1}
\end{equation}
Using the same trick as before, we find
\begin{align}
\f{F_{+}'(z)}{F_{+}(z)} & \simeq\f{\g x+\cot(x/2)(4-z\tan(z/4))}{4(z\cot(x/2)+\g x\tan(z/4))}\\
\f{F_{-}'(z)}{F_{-}(z)} & \simeq\f{\g x+\cot(x/2)(4+z\cot(z/4))}{4(z\cot(x/2)-\g x\cot(z/4))}
\end{align}

Taking these into (\ref{eq:79-1}), we deform the contour $\mC_{\pm}$
to other poles on the complex plane. For this step, it depends on
how much we trust the approximation of asymptotic eigenvalues $\tilde{\mM}_{\pm}$.
If we define $\mC_{\pm}$ excluding the subset $K_{\pm}\in\mM_{\pm}$,
and approximate the remainder eigenvalues as $\tilde{\mM}_{\pm}$,
after deformation of $\mC_{\pm}$, we should consider the residue
at the corresponding set of eigenvalues $\tilde K_{\pm}$ in $\tilde{\mM}_{\pm}$.
As a complement, we need to include the eigenvalues in $K_{\pm}$
explicitly in (\ref{eq:74}). A subtlety is that the number of eigenvalues
in $K_{\pm}$ may not equal to that in $\tilde K_{\pm}$ because these excluded
eigenvalues are near origin and $\tilde{\mM}_{\pm}$ may differ from
$\mM_{\pm}$ by a finite amount (see Fig. \ref{fig:An-illustration-of}).

Besides these nontrivial poles in \eqref{eq:79-1} depending on $\tilde K_{\pm}$, the remaining
poles  are $z=0,\pm i\sqrt{\a}$ for $I_{+}$; and the remaining poles
are $z=\pm i\sqrt{\a}$ for $I_{-}$. Using residue theorem at these
points and integrate over $\a$, we find that \footnote{There is a subtlety when we have imaginary eigenvalues. For the imaginary
eigenvalue $n_*$ in $\tilde{\mM}_{+}$, we have $\sin^{2}(n_*/4)<0$,
which implies $n_*^{2}+(\g x\tan\f{x}2)^{2}>0$. Therefore,
when we do the integral over $\a$ in (\ref{eq:79-1}), it will hit
a singularity of $1/F_{\pm}(\pm i\sqrt{\a})$ when $\pm i\sqrt{\a}=n_*$.
To avoid this issue, note that the 1-loop determinant only cares about
$|\sin^{2}(n/4)|$ and we can shift $n\ra n+\e$
in (\ref{eq:77-1}) for infinitesimal positive $\e$ such that the
average over $z=\pm n_*$ for $\log\f{(z+\e)^{2}}{(z+\e)^{2}+(\g x\tan\f{x}2)^{2}}$
gives the real value $\log|\sin^{2}(n_*/4)|$. Adding this
$\e$ leads to the absolute value in the log term of (\ref{eq:102}). On
the other hand, for the imaginary eigenvalue in $\tilde{\mM}_{-}$,
we have $\cos^{2}(n/4)>0$ and thus $n^{2}+(\g x\tan\f{x}2)^{2}<0$.
It follows that the $\a$ integral does not hit any singularity and
we do not need the $\e$-prescription.}
\begin{align}
I_{+} & \supset\log\left|1+\f 14\g x\tan\f x2\right|-\f 14\g x\tan\f x2\label{eq:102}\\
I_{-} & \supset-\f 14\g x\tan\f x2\label{eq:103}
\end{align}
Expanding in small $y=\pi-x$, we have 
\begin{align}
I_{+} & \sim-\f{\g\pi}{2y}+\f{\g}2+\log\f{|\g|\pi}{2y}+\f{48-24\g+\g^{2}\pi^{2}}{24\g\pi}y \label{eq:108}\\
I_{-} & \sim-\f{\g\pi}{2y}+\f{\g}2+\f{\g\pi}{24}y \label{eq:109-1}
\end{align}
For the nontrivial poles at $\tilde K_\pm$ (because we exclude them in $\mC^\pm$) after deforming $\mC_{\pm}$, we have
\begin{align}
I_{\pm}\supset &\f 12\sum_{p\in \tilde K_{\pm}}\log\left|1+\f{\g^{2}x^{2}\tan^{2}x/2}{p^{2}}\right|\nn\\
=&\begin{cases}
-\f 12\sum_{p}\log\left|\sin^{2}\f p4\right| & p\in \tilde K_{+}\\
-\f 12\sum_{p}\log\left|\cos^{2}\f p4\right| & p\in \tilde K_{-}
\end{cases}\label{eq:109}
\end{align}
where we used the fact that $\text{Res}_{z=p}\left(F'_{\pm}(z)/F_{\pm}(z)\right)=1$
at any simple pole. On other hand hand, we should take explicit eigenvalues
in $K_{\pm}$
\begin{align}
I_{+}\supset&\f 12\sum_{p\in K_{+}}\log\left|\sin^{2}\f p4\right|\label{eq:110}\\
I_{-}\supset&\f 12\sum_{p\in K_{-}}\log\left|\cos^{2}\f p4\right|\label{eq:110-1}
\end{align}

As we decrease the temperature, there is a singularity at $x\tan\f x2=4/(-\g)$
in (\ref{eq:102}), which corresponds to the temperature
$T_{0}$. However, it is unphysical and needs to be resolved because
we should not trust the approximate eigenvalues near the origin. This
requires non-empty choice of $\tilde K_{+}$. Indeed, as we decrease temperature
across $T_{0}$, there is a pair of zeros of $F_{+}$ in (\ref{eq:102-1})
moves from real axis to imaginary axis. In particular, at $T_{r}=T_{0}$,
$F_{+}$ has triple zero at $z=0$. If we include this pair of zeros
in $\tilde K_{+}$, (\ref{eq:109}) will give a singularity at $T_{r}=T_{0}$,
which compensates the singularity in (\ref{eq:102}). Let us check
this explicitly by expanding $\g$ in the neighborhood as
\begin{equation}
\g=-\f 4{x\tan x/2}+\d
\end{equation}
The singularity of (\ref{eq:102}) is $\log|\d|$. The eigenvalue
around the origin from $F_{+}=0$ can be expanded in terms of $\d$
\begin{equation}
F_{+}(z)=0\implies z=p=\pm2\sqrt{3\d x\tan\f x2}
\end{equation}
Taking this $p$ into (\ref{eq:109}), we find divergence $\sim\log1/|\d|$,
which cancels the singularity of (\ref{eq:102}) exactly. This is
reasonable because the sum over (\ref{eq:102}) to (\ref{eq:109})
is nothing but the Sommerfeld-Watson resummation of $\tilde{\mM}_{\pm}$
excluding $\tilde K_{\pm}$. Since we exclude the only zero that may across
the origin and cause singular determinant (\ref{eq:74}), all other
eigenvalues in $\tilde{\mM}_{+}$ are finite and should give a smooth
result.

\begin{figure*}
\begin{centering}
\subfloat[$h=0.3,\protect\t_{r}=0.2\;(T_{r}>T_{0})$\label{fig:1a}]{\begin{centering}
\includegraphics[height=3cm]{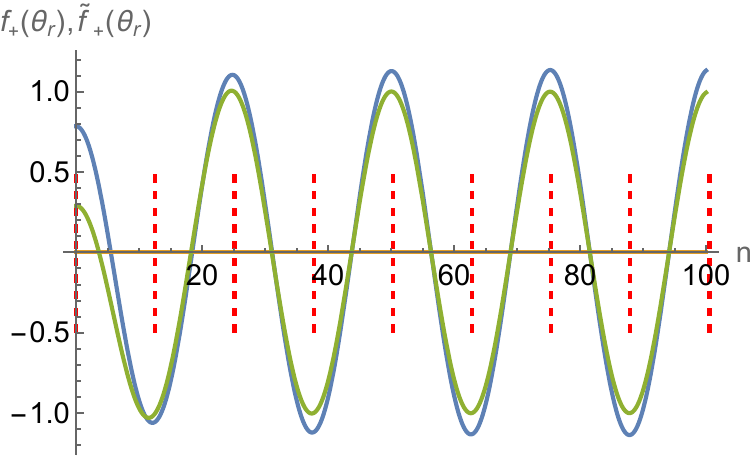}
\par\end{centering}
}\subfloat[$h=0.3,\protect\t_{r}=0.1\;(T_{r}<T_{0})$\label{fig:1b}]{\begin{centering}
\includegraphics[height=3cm]{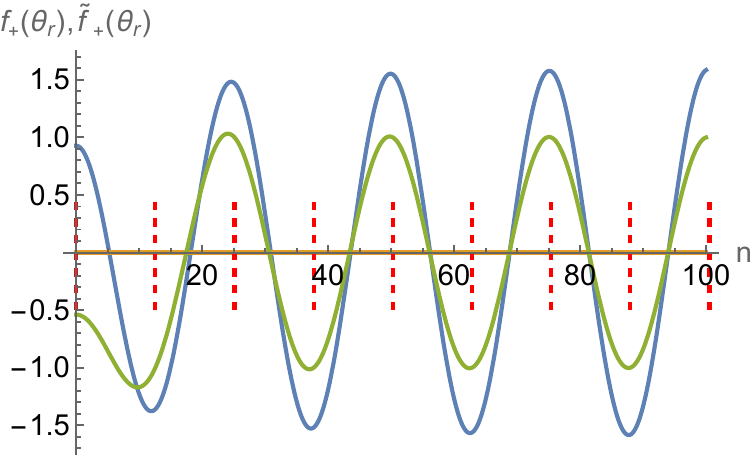}\includegraphics[height=3cm]{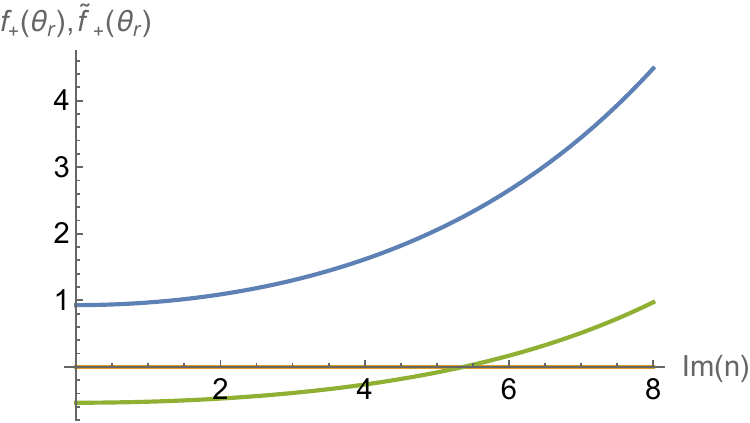}
\par\end{centering}
}\\
\subfloat[$h=1/2+i,\protect\t_{r}=0.4\;(T_{r}>T_{+})$\label{fig:1c}]{\begin{centering}
\includegraphics[height=3cm]{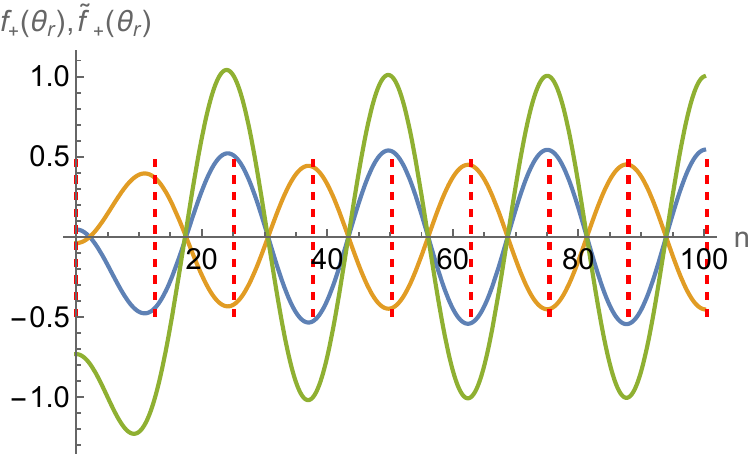}
\par\end{centering}
}\subfloat[$h=1/2+i,\protect\t_{r}=0.25\;(T_{r}<T_{+})$\label{fig:1d}]{\begin{centering}
\includegraphics[height=3cm]{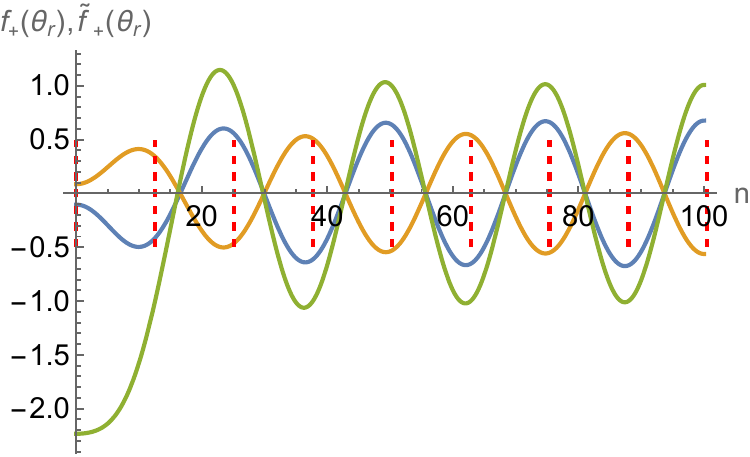}\includegraphics[height=3cm]{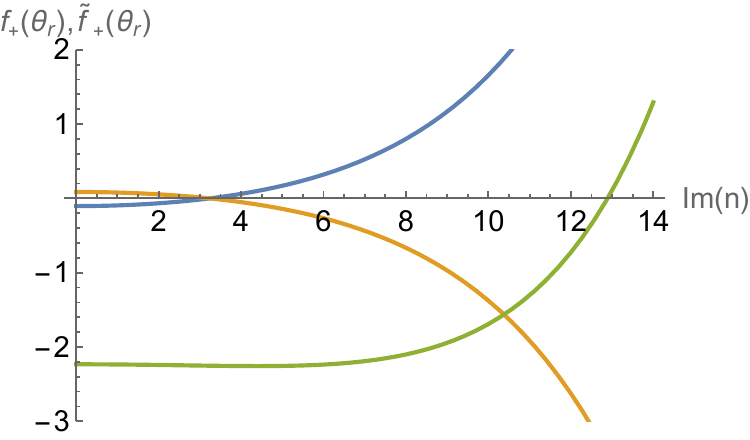}
\par\end{centering}
}
\par\end{centering}
\caption{\justifying  Comparison of $f_{+}(\protect\t_{r})$ and $\tilde{f}_{+}(\protect\t_{r})$
as function of $n$. Zeros are eigenvalues. The first row (a) and
(b) are for $h<1/2$, and the second row (c) and (d) are for $h=1/2+i\nu$.
The blue/yellow curve is the real/imaginary part of the exact eigenfunction
$f_{+}(\protect\t_{r})$ (in the first row $f_{+}(\protect\t_{r})$
is real) and the green curve is the approximate eigenfunction $\tilde{f}_{+}(\protect\t_{r})$.
The normalization has been tuned for a better visual comparison, in
which some trivial zeros in the expression \eqref{eq:46} and \eqref{eq:92} are also removed. Since
the functions are symmetric in $n$, we only plot the $n>0$ or $\Im n>0$
part. The red dashed lines are $n=4\pi\protect\Z$, on which the determinant
(\ref{eq:74}) will be singular. (a) $T_{r}>T_{0}$ and all eigenvalues
of $\tilde{\protect\mM}_{+}$ are real. (b) $T_{r}<T_{0}$ and a pair
of eigenvalues of $\tilde{\protect\mM}_{+}$ is imaginary (the right
picture of (b)) though all eigenvalues of $\protect\mM_{+}$ are still
real. (c) $T_{r}>T_{+}$ and all eigenvalues of $\protect\mM_{+}$
are real. (d) $T_{r}<T_{+}$ and a pair of eigenvalues of $\protect\mM_{+}$
is imaginary (the right picture of (d)).}
\end{figure*}

From this analysis, we learn another lesson that the choice of $\tilde K_{+}$
also depends on the temperature. For $T_{r}>T_{0}$, we can choose
$\tilde K_{+}$ to be empty set without too much error (see Fig. \ref{fig:1a}).
However, as temperature drops to $T_{r}<T_{0}$,
we must choose $\tilde K_{+}$ to include at least two opposite imaginary
eigenvalues to cure the unphysical singularity (see Fig. \ref{fig:1b}).
Let us consider the low temperature contribution from this pair of
imaginary eigenvalues in (\ref{eq:109}). Taking $y=\pi-x$ and expand
in small $y$ limit, the imaginary eigenvalues are 
\begin{equation}
p\app\pm2\pi i\g/y
\end{equation}
which in (\ref{eq:109}) leads to
\begin{equation}
I_{+}\supset-\log\sinh^{2}\f{\pi(-\g)}{2y}\app\f{\pi\g}y\label{eq:114}
\end{equation}

\begin{figure*}
\begin{centering}
\subfloat[$h=0.3,\protect\t_{r}=0.05$\label{fig:2a}]{\begin{centering}
\includegraphics[height=3cm]{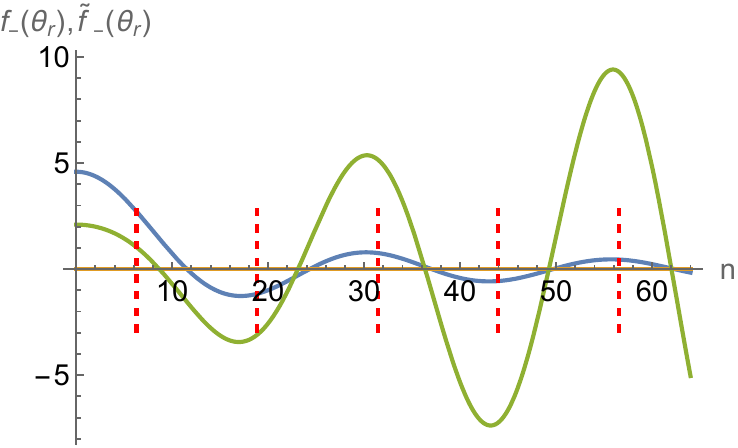}\includegraphics[height=3cm]{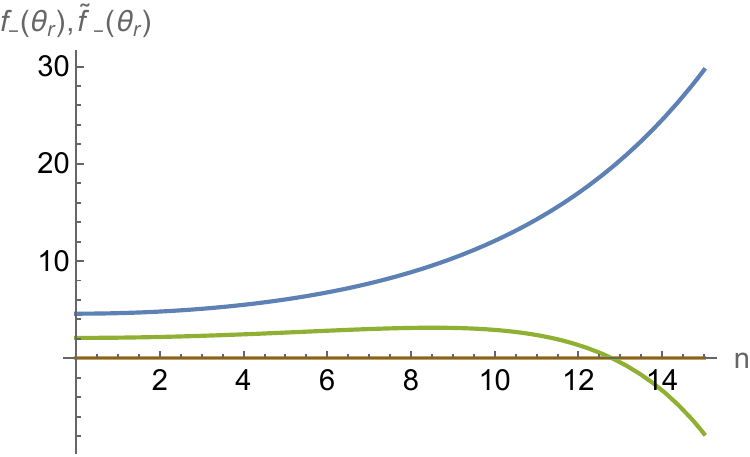}
\par\end{centering}
}\\
\subfloat[$h=1/2+i,\protect\t_{r}=0.2\;(T_{r}>T_{-})$\label{fig:2b}]{\begin{centering}
\includegraphics[height=3cm]{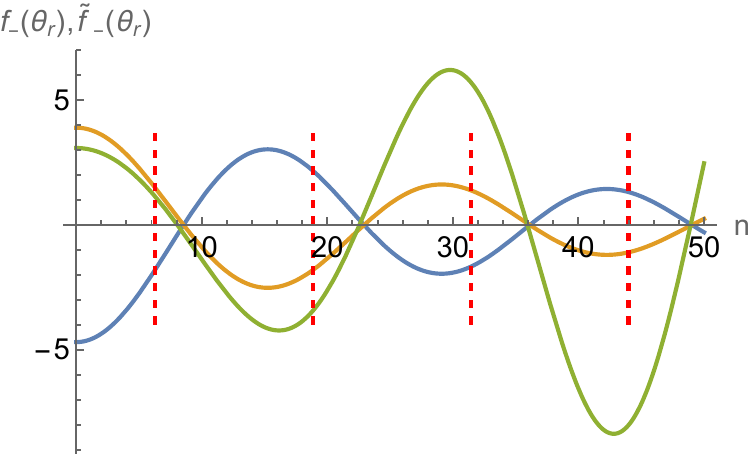}
\par\end{centering}
}\subfloat[$h=1/2+i,\protect\t_{r}=0.1\;(T_{r}<T_{-})$\label{fig:2c}]{\begin{centering}
\includegraphics[height=3cm]{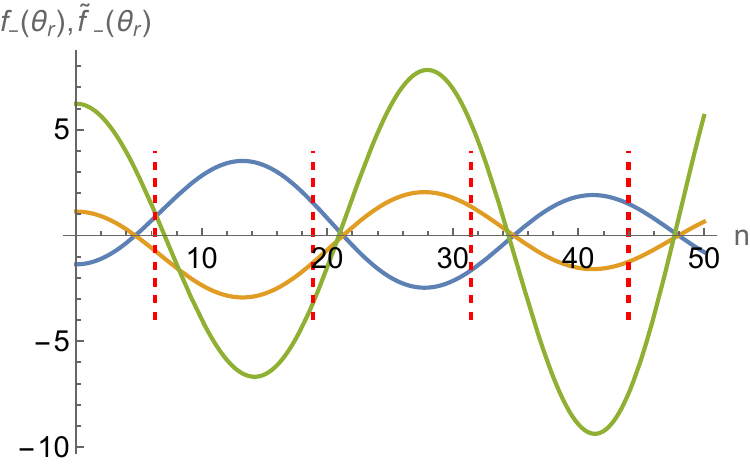}
\par\end{centering}
}
\par\end{centering}
\caption{\justifying  Comparison of $f_{-}(\protect\t_{r})$ and $\tilde{f}_{-}(\protect\t_{r})$
as function of $n$. Zeros are eigenvalues. The first row (a) is for
$h<1/2$, and the second row (b) and (c) are for $h=1/2+i\nu$. The
blue/yellow curve is the real/imaginary part of the exact eigenfunction
$f_{-}(\protect\t_{r})$ (in the first row $f_{-}(\protect\t_{r})$
is real) and the green curve is the approximate eigenfunction $\tilde{f}_{-}(\protect\t_{r})$.
The normalization has been tuned for a better visual comparison, in
which some trivial zeros in the expression \eqref{eq:46} and \eqref{eq:93} are also removed. Since the
functions are symmetric in $n$, we only plot the $n>0$ or $\Im n>0$
part. The red dashed lines are $n=2\pi+4\pi\protect\Z$, on which
the determinant (\ref{eq:74}) will be singular. (a) $\protect\mM_{-}$
only has real eigenvalues and $\tilde{\protect\mM}_{-}$ has one pair
of imaginary eigenvalues. (b) $T_{r}>T_{-}$ and the smallest positive
eigenvalue of $\protect\mM_{-}$ is larger than $2\pi$. (c) $T_{r}<T_{-}$
and the smallest positive eigenvalue of $\protect\mM_{-}$ is smaller
than $2\pi$. The smallest positive eigenvalue of $\tilde{\protect\mM}_{-}$
never crosses $2\pi$.}
\end{figure*}

We can consider a nontrivial $\tilde K_{-}$ as well. As we discussed before,
for $\g<0$ there is a pair of imaginary eigenvalues in $\tilde{\mM}_{-}$
given by (\ref{eq:98}) (see Fig. \ref{fig:2a}). For small $y$,
this eigenvalue is given by
\begin{equation}
p=\pm i2\pi(-\g)/y
\end{equation}
If we include them in $\tilde K_{-}$, the contribution to (\ref{eq:109})
is 
\begin{equation}
I_{-}\supset-\log\cosh^{2}\f{\pi(-\g)}{2y}\app\f{\pi\g}y\label{eq:85}
\end{equation}
which is the same as (\ref{eq:114}) at leading order.

\subsubsection{The first exact eigenvalue for $2\leq d<9$} \label{sec:iv-3}

Since we choose a nontrivial $\tilde K_{\pm}$, we need to include the exact
eigenvalues in $K_{\pm}$ for compensation in (\ref{eq:110}). We
will split the discussion into two parts: $h=1/2+i\nu$ ($2\leq d<9$) and $h<1/2$ ($d\geq 9$) 
because they have quite different behaviors. In this section, we will focus on  the $h=1/2+i\nu$ case.

\begin{figure}
\begin{centering}
\includegraphics[height=4cm]{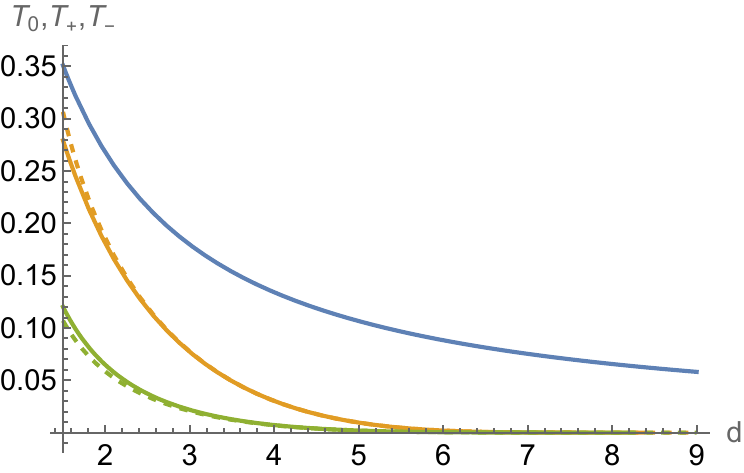}
\par\end{centering}
\caption{\justifying Three temperatures for $d\in[2,9]$ and $h=1/2+i \nu$.
The blue curve is the temperature $T_{0}$ related to approximate
eigenvalues $\tilde{\protect\mM}_{+}$, which has no physical significance. The
yellow curve is the critical temperature $T_{+}$ for exact eigenvalues
$\protect\mM_{+}$ and the dashed curve is the approximation (\ref{eq:89}).
The yellow curve is the critical temperature $T_{-}$ for exact eigenvalues
$\protect\mM_{-}$ and the dashed curve is the approximation (\ref{eq:95}).
Here we see that the approximation for $T_{\pm}$ works quite well
and $T_{0}>T_{+}>T_{-}$.\protect\label{fig:Three-critical-temperatures}}
\end{figure}

To solve the exact eigenvalues in low temperature, we can expand the exact wave function around
small $\t_{r}=y/2$ to first two orders \footnote{For $d=9$ and $h=1/2$, the small $y$ expansion is different and
the leading divergence is $\log y$. }
\begin{align}
&f_{\pm}(\t_{r})\propto (y/4)^{2h-1}(\cos h\pi\pm\cos\f{n\pi}{2x}) \nn\\
&\mp\f{\pi\G(h+1/2)}{\G(3/2-h)\G(h-n/(2x))\G(h+n/(2x))}+\cdots\label{eq:115}
\end{align}

For $h=1/2+i\nu$ with $\nu>0$, there is a critical temperature $T_{+}$. If $T_{r}<T_{+}$
there will be pairs of opposite imaginary eigenvalues as shown in
Fig. \ref{fig:1d}; if $T_{r}>T_{+}$, all eigenvalues are along
the real axis as shown in Fig. \ref{fig:1c}. The critical temperature
occurs at a value of $\t_{r}\in[0,\pi/2]$ such that $f_{+}(\t_{r})=0$
and $n=0$. For $2\leq d\leq9$, $\nu=\f 1 2 \sqrt{\f{9-d}{d-1}}\leq1.32$  by \eqref{eq:45-1} and we find that the
small $\t_{r}$ expansion (\ref{eq:115}) gives a quite good approximation.

Taking $n=0$ in (\ref{eq:115}), we find
\begin{equation}
(y/4)^{2i\nu}=\f{\pi\G(1+i\nu)}{\G(1-i\nu)\G(1/2+i\nu)^{2}(1-i\sinh\pi\nu)}
\end{equation}
Note that the RHS has unit norm, which leads to a real solution of
$y$. Moreover, there are infinite numbers of solutions of $y$ because
$y\ra ye^{k\pi/\nu}$ for any $k\in\Z$ is also a solution. This means
that as we decrease the temperature, there will be more and more pure
imaginary eigenvalues. For our purpose, we find the critical temperature
$1/\b_{+}$ for the existence of the first pure imaginary eigenvalue
corresponds to $k=-1$, namely
\begin{align}
y_{+}=&4\left[\f{\pi\G(1+i\nu)}{\G(1-i\nu)\G(1/2+i\nu)^{2}(1-i\sinh\pi\nu)}\right]^{-\f{i}{2\nu}}e^{-\pi/\nu}\nn\\
\sim &2\pi/(\sqrt{1-1/d}\b_{+})\label{eq:89}
\end{align}
where the second line is copied from \eqref{eq:36}. More accurately, the critical temperature $T_{+}$ is given by (\ref{eq:t-y}) in terms
of $y_{+}$. We show in Fig. \ref{fig:Three-critical-temperatures}
that this critical temperature is lower than $T_{0}$.

Existence of such a critical temperature has a significant impact
on the 1-loop determinant. At the critical temperature, there is
an exact null eigenvalue, which leads to vanishing determinant. This
means that the partition function will be singular at a specific temperature
and imply an infinite tail in the spectrum. To see this, let us
expand $y=y_{+}+t$ and solve $n$ for small $t$. Since there are
two symmetric eigenvalues, the leading order must be 
\begin{equation}
n/x=\pm A_{+}t^{1/2},\qquad A_{+}>0
\end{equation}
where for $t>0$ two eigenvalues are real and for $t<0$ two eigenvalues
are imaginary. Existence of imaginary eigenvalues implies instability of the homogenous saddle. To avoid finding new saddles, let us focus on $t>0$ and study the critical behavior as $t$ gets close to zero. The contribution to $I_{+}$
in (\ref{eq:110}) is 
\begin{equation}
I_{+}\sim\log\sin^{2}(A_{+}t^{1/2}x/4)\sim\log(y-y_{+})\label{eq:120}
\end{equation}
Practically, taking this exact eigenvalue into $K_{+}$ does not affect
the choice of $\tilde K_{+}$ because $\tilde{\mM}_{+}$ does not have a
corresponding eigenvalue for this one (but it has another pair of
pure imaginary eigenvalues as discussed before because $T_{+}<T_{0}$)
as shown in Fig. \ref{fig:1c}.

Together with $\log1/y$ term in (\ref{eq:108}), we find the contribution to the 1-loop
determinant is
\begin{equation}
e^{-(d-1)I_{+}/2}\sim\f 1{(\b_{+}-\b_r)^{(d-1)/2}}
\end{equation}
where the power $(d-1)$ is due to $(d-1)$ numbers of $h\neq2$ modes. As we restrict to $t>0$, here we should consider the validity of this partition function only for $\b_r<\b_+$. The inverse Laplace transformation for $E<0$ \footnote{Here we momentarily take $E_0=0$ for simplicity. Recovering $E_0$ simply replaces $E\ra E-E_0$. Note that the $\b$ contour in the inverse Laplace transformation is to the left of $\b_+\lam^{-1/2}$ and leads to the exponential decay for $E<0$.} leads to
\begin{equation}
\r(E)\sim\int_{i\R}\f{(2\pi i)^{-1}d\b e^{E\b}}{(\b_{+}\lam^{-\f 1 2}-\b)^{(d-1)/2}}\sim(-E)^{(d-3)/2}e^{\b_{+}\lam^{-\f 1 2}E}\label{eq:122}
\end{equation}
which means that the spectrum is non-compact and has a tail! As $d$
increases, though the power part grows, the more important critical inverse
temperature $\b_{+}\lam^{-1/2}$ increases as well and leads to a stronger suppression.
This is consistent with our expectation that larger $d$ has a smaller
tail.

Physically, at the critical temperature $T_{+}$, having a null eigenvalue
at quadratic order simply means we need to consider cubic or higher
order terms around the saddle to soften this singularity. This analysis
will be quite involved and beyond the scope of this work. However,
this does not affect (\ref{eq:120}) for $T_r\gtrsim T_{+}$ or the
existence of the exponential tail of the spectrum in (\ref{eq:122})
 though its shape might be modified for large enough $|E|$.

As for $K_{-}$, the situation is similar to $K_{+}$.
As we drop the temperature over some value, a pair of real eigenvalues
with smallest magnitude in $\mM_{-}$ moves to imaginary axis after
colliding at the origin. However, a null eigenvalue for $\mM_{-}$
does not cause any issue because the contribution to determinant (\ref{eq:110})
is a cosine function rather than a sine function. Indeed, the singularity
occurs at a higher temperature $T_{-}$ when an exact eigenvalue
hit $n=2\pi$ (see Fig. \ref{fig:2b} and \ref{fig:2c}). To solve
this critical temperature $T_{-}$, we set $f_{-}(\t_{r})=0$ with
$n=2\pi$ in (\ref{eq:115}), which leads to
\begin{equation}
\left(\f{\pi-x}4\right)^{2i\nu}=\f{\G(1+i\nu)\G(\f 1 2-i\nu+\f \pi x)\cos\pi(i\nu-\f \pi x)}{\G(1-i\nu)\G(\f 1 2+i\nu+\f \pi x)(\cos\f{\pi^{2}}x+i\sinh\pi\nu)}
\end{equation}
where we plugged in $\t_{r}=(\pi-x)/2$. 

We do not have a small parameter
to solve this equation perturbatively. However, the solution can be
approximated by iteration: take $x\ra\pi$ on the RHS and solve $x$
for LHS, then take this value to RHS and solve $x$ for LHS, and so
on. The leading order approximation in $y=\pi-x$ variable is 
\begin{equation}
y_{-}\app4\left(\f{\G(1+i\nu)\G(3/2-i\nu)\cosh \pi \nu}{\G(1-i\nu)\G(3/2+i\nu)(1-i\sinh\pi\nu)}\right)^{-\f{i}{2\nu}}e^{-\pi/\nu}\label{eq:95}
\end{equation}
Numerically, this approximation works well for $d\in[2,9]$ as shown
in Fig. \ref{fig:Three-critical-temperatures}. For this range of
$d$, we show that $T_{-}<T_{+}$ in Fig. \ref{fig:Three-critical-temperatures}.
Therefore, the true physical critical temperature (with $\lam^{1/2}$ restored) is 
\be 
T_{c}=T_{+} \lam^{1/2}
\ee

\begin{figure}

\begin{centering}
\subfloat[$d=2$\label{10a}]{\begin{centering}
\includegraphics[height=2.7cm]{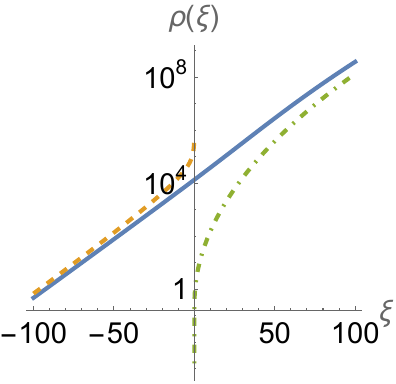}\end{centering}}
\hfill\subfloat[$d=4$\label{10b}]{\begin{centering}
\includegraphics[height=2.7cm]{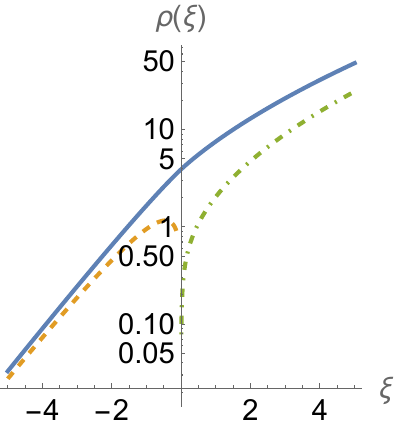}\end{centering}}
\hfill\subfloat[$d=6$\label{10c}]{\begin{centering}
\includegraphics[height=2.7cm]{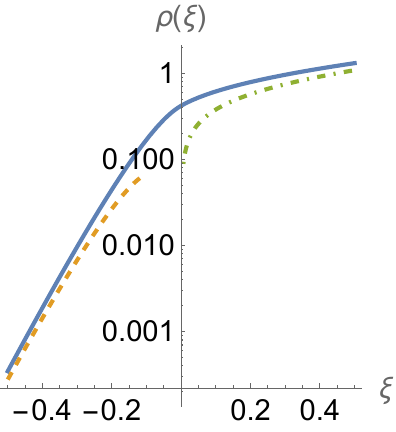}\end{centering}}
\par\end{centering}
\caption{\justifying The edge spectra for different $d<9$ (blue curves). Here we choose $\lam=1/2$ for all three figures. As $d$ increases, the transition between $\xi>0$ and $\xi<0$ becomes sharper. The orange dashed curves are $\sim(-\xi)^{(d-3)/2}e^{\b_* \xi}$ and the green dot-dashed curves are $\sim\sinh 2\sqrt{\xi}$. The normalizations of the asymptotic curves are correctly estimated to match the exact spectra. \label{fig:rho}}
\end{figure}

For temperature above but close to $T_c$, the full tree level plus 1-loop low temperature partition function is 
\begin{equation}
Z(\b)\sim\f 1{(1/T_c-\b)^{(d-1)/2}\b^{3/2}}e^{-\b E_{0}+c/\b}\label{eq:103-3}
\end{equation}
where $c=\pi^{2}/[\lambda(1-1/d)]^{3/2}+O(1/\lambda^{1/2})$ and
the ``ground" energy $E_{0}$ has leading contribution from the saddle
(\ref{eq:38-1}) and subleading contribution from (\ref{eq:55-2}),
(\ref{eq:108}), (\ref{eq:109-1}), (\ref{eq:114}) and (\ref{eq:85})
\begin{equation}
E_{0}=-\f 2{\sqrt{\lambda(1-1/d)}}\left(1+\f 12(1-1/d)\lambda+\cdots\right)\label{eq:104}
\end{equation}
The edge spectrum is given by the inverse Laplace transformation of
(\ref{eq:103-3}). It is convenient to define the triple-scaled energy near the edge
\be 
\xi=(E-E_0)c \label{197-xi}
\ee
and the spectrum can be written as
\be 
\r(\xi)\sim \f 1 {2\pi i}\int_{\e-i\infty}^{\e+i\infty}\f {d\b e^{\b \xi+1/\b}}{(1/(cT_c)-\b)^{(d-1)/2}\b^{3/2}} \label{197-eq}
\ee
where we should choose $0<\e<1/(cT_c)$. From this equation, the effective parameter that controls the sharpness of the transition between the chaotic part of the spectrum $\sim e^{2\sqrt{\xi}}$ and the exponential tail $\sim (-\xi)^{(d-3)/2}e^{\b_*\xi}$ is 
\be 
\b_*\equiv 1/(cT_c)\app \f{(1-1/d)^{3/2}}{\pi^2} \cdot \f {\lam}{T_+}
\ee
Though $\lam$ is a small parameter in our approximation, from Fig. \ref{fig:Three-critical-temperatures}  we see $T_+$ can also go to a very small value when $d$ gets close to 9. Therefore, how sharp the transition is depends on the competition between $\lam$ and $d$. Physically, this is quite reasonable. When $\lam$ is smaller, the model is closer to the commuting SYK model with non-compact spectrum and we should expect a bigger tail and a milder transition; when $d$ is larger, the model is closer to the double-scale SYK model with compact spectrum and we should expect a smaller tail and sharper transition. This feature is illustrated in Fig. \ref{fig:rho}.

In the case $T_+ \ll \lam$ that is relevant to the numerical simulation in Section \ref{sec:num}, for the energy window very close to the tail, we will mainly have the transition between $\sim \sqrt{\xi}$ and $\sim(-\xi)^{(d-3)/2}e^{\b_* \xi}$. For this regime, we can drop the $c/\b$ term in the exponent of the partition function \eqref{eq:103-3} and the energy resolution changes from $1/c$ to $T_c$. Let us define 
\be
\zeta=(E-E_0)/T_c
\ee
and the spectrum is given by
\begin{align} 
\r(\zeta)&\sim \f {e^{\zeta}} {2\pi i}\int_{-1/2-i\infty}^{-1/2+i\infty}\f {dt e^{\zeta t}}{(1+t)^{3/2}(-t)^{(d-1)/2}} \nn\\
&=\begin{cases}
    \f{e^{\zeta}U(3/2-d/2,1-d/2,-\zeta)}{\G((d-1)/2)}\quad &\zeta<0\\
    \f{U(-1/2,1-d/2,\zeta)}{\G(3/2)}\quad &\zeta>0
\end{cases} \label{201-rho}
\end{align}
where we changed variable $\b=(1+t)/T_c$ in the first line and $U(a,b,z)$ is the Tricomi's confluent hypergeometric function. This expression is a smooth interpolation between two asymptotic regions and particularly improves the estimation around $\zeta=0$, at which the two asymptotic expressions $\sim \sqrt{\zeta}$ and $\sim(-\zeta)^{(d-3)/2}e^{\zeta}$ are both singular.

\subsubsection{1-loop correction for $d> 9$} \label{1-loop_d9}

For $d>9$ we have real $h<1/2$. In this case, the exact eigenvalues $\mM_{+}$ are always real. From
Fig. \ref{fig:1b} we see that in low temperature $T_{r}<T_{0}$
(but not too low), the approximate eigenvalues $\tilde{\mM}_{+}$
matches with $\mM_{+}$ quite well except the pure imaginary one in
$\tilde K_{+}$. On the other hand, there is indeed no imaginary eigenvalues
in the exact set $\mM_{+}$. Therefore, to compensate $\tilde K_{+}$, we
need to include at least the pair of exact eigenvalues with smallest magnitude,
which has no approximate proxy in $\tilde{\mM}_{+}$ in Fig. \ref{fig:1b}. 

For $h<1/2$ and small $y$, the first term of \eqref{eq:115} is singular. To find the
first solution of $f_{+}(\t_{r})=0$ we need to take $n\pi/(2x)$
close to $\pi(1-h)$ because $1/\G(x)$ is a entire function. It follows
that the first positive exact solution is
\begin{equation}
p/x \app 2(1-h)-\f{2\G(h+1/2)(y/4)^{1-2h}}{\G(2h-1)\G(3/2-h)\sin h\pi}\label{eq:117}
\end{equation}
Taking this back to (\ref{eq:110}) gives a smooth and finite contribution
to $I_{+}\sim\log\cos^{2}(\pi h/2)+O(y^{1-2h})$.

For $\mM_-$, we again find all eigenvalues being real for $h<1/2$. In this case, as long as the temperature is not too low, the eigenvalues in $\mM_{-}$ has one-to-one match with the real eigenvalues in $\tilde{\mM}_{-}$
with small errors (Fig. \ref{fig:2a}). These errors asymptotically
approach to zero as $n$ goes large. The first positive exact eigenvalue is around $n/x\sim 2h$, which gives a finite value for $I_-\sim\log\cos^{2}(\pi h/2)$.

By above analysis, we find no singular term appearing at the 1-loop level for $h<1/2$. It follows that the mechanism for the tail of the spectrum
in the IR for $h=1/2+i\nu$ breaks down. This is consistent with the observation in  Fig. \ref{fig:Three-critical-temperatures} that the critical temperature
$T_{+}$ goes to zero when $d\ra9$. One may wonder if there is a different mechanism to produce a tail at 1-loop order if we consider a larger set of $K_\pm$ of exact eigenvalues. This sounds plausible since the approximation of $\tilde K_\pm$ becomes worse if we go to lower temperatures and perhaps the accumulation of the errors may cause something dramatic. We do this analysis in Appendix \ref{app:2}, and numerically find a correction to the partition function from a large size ($\sim 100$) of $K_\pm$ that
\begin{equation}
Z(\b)\sim\f 1{\b^{3/2+a_{2}/2}}e^{-\b(\tilde{E}_{0}+\f 12a_{1}\sqrt{\lambda})+c/\b} \label{201-eq}
\end{equation}
where $\tilde E_0$ is the ground energy of the leading large $d$ approximation $q\ra \bar q$ by \eqref{39-eq}
\begin{align}
\tilde{E}_{0}=&-\f 2{\sqrt{(1-q)(1-1/d)}} \nn\\
=&-\f 2{\sqrt{\lambda(1-1/d)}}\left(1+\f{\lambda}4+\cdots\right)\label{eq:106-1}
\end{align}
and $a_{1,2}$ are two parameters to fit with the numeric data. From Fig. \ref{fig:6b}, we find both $a_{1,2}$ go to zero as $d$ increases as expected \footnote{Rigorously speaking, for $h=1/2+i\nu$ in Sec. \ref{sec:iv-3} we should also consider the accumulation effect of large number of eigenvalues in $K_\pm$ because the ground energy \eqref{eq:104} has an order $O(\sqrt{\lam})$ gap from $\tilde E_0$ at large $d$ limit. However, this does not change the singularity of the partition function from the first null eigenvalue.}. 

Using the $\xi$ from \eqref{197-xi} with $E_0=\tilde E_0+a_1\sqrt{\lam}/2$, the inverse Laplace transformation of \eqref{201-eq} leads to the spectrum
\begin{equation}
\r(\xi)\sim \xi^{(a_{2}+1)/4}I_{(a_{2}+1)/2}(2\sqrt{\xi})
\end{equation}
where $I_{k}$ is the modified Bessel function of first kind. The
low and high energy behavior of this spectrum is
\begin{equation}
\r(\xi)\sim\begin{cases}
\xi^{(a_{2}+1)/2} & \xi\ll 1\\
\xi^{a_{2}/4}e^{2\sqrt{\xi}} & \xi\gg 1
\end{cases}\label{eq:119}
\end{equation}
In this expression, $a_{1}$ shifts the ground energy to the left and $a_{2}$ modifies
the shape of the spectrum.

This numeric computation shows the low temperature partition function
and the spectrum for large but finite $d>9$ beyond the Schwarzian
limit of the double-scaled SYK at 1-loop order. However, the 1-loop analysis does not give the tail of the non-compact
spectrum. We expect this highly suppressed tail follows from higher
loop corrections, which we leave for future investigation.

\section{Thermodynmical meaning of the critical temperature} \label{sec:5}

In previous sections, we have computed a critical temperature $T_c$ for the dcSYK model in the $q=0$ case and the $q\ra 1$ case with $d<9$. Since this critical temperature is read from the exponent of the infinite exponential tail of the spectrum, we should naturally expect that there does not exist quantum chaos for $T\lessapprox T_c$ because the system simply does not have enough eigenvalues to scramble in the tail regime with a sparse spectrum. On the other hand, the existence of the $h=2$ mode for $q\rightarrow 1$, which is identical to the Schwarzian mode, suggests that quantum chaos should exist for $T\gtrapprox T_c$. This transition of quantum chaos will be further supported by our numerical studies of spectral form factor in Sec. \ref{sec:SFF}. However, this physical interpretation has nothing to do with the critical temperature in the thermodynamical sense, defined as a non-smoothness change of free energy. 

To compute the thermodynamical critical temperature $T_{t.c.}$, we need to compute the ensemble average of $\log Z(\b)$. However, the chord diagram formalism on a disk in the double scale limit is equivalent to compute the ensemble average of $Z(\b)$. As a canonical ensemble, these two are equivalent in most temperatures because eigenvalues thermalize almost in the same way across different random samples. As we go to the low temperatures when only finite numbers of eigenvalues contribute to the free energy, the fluctuation among random samples becomes huge, and $\E [\log Z(\b)]$ and $\log \E [Z(\b)]$ can have a significant difference. Physically, we often call the former as ``quenched" and the latter as ``annealed". 

It is usually quite hard to compute the quenched free energy 
\be 
F_q(\b)=-\f 1 {N\b}\E [\log Z(\b)]
\ee
for a chaotic system because of the ensemble average of a log function. A practical way is to use a replica trick to compute $F_n(\b)=-\f 1 {N\b} \E[ Z(\b)^n-1]/n$ for all positive integer $n$ and analytically continue $n$ to 0
\be 
F_q(\b)=\lim_{n\ra 1} F_n(\b)
\ee
Computing the $n$ replica ensemble average of the partition function in the double scale limit requires computing the chord diagrams connecting $n$ circles following the method in \cite{Berkooz:2020fvm}, though we do not have a closed form for general $n$ even for the double-scaled SYK model. 

Nevertheless, the annealed free energy 
\be 
F_a(\b)=-\f 1 {N\b} \log  \E [Z(\b)]
\ee 
is still informative because it gives a lower bound to the truly thermodynamical critical temperature $T_{t.c.}$. Indeed, it follows from the following general inequality for any finite $N$ system (see e.g. \cite{mezard2009information})
\be 
F_q(\b)\geq F_a(\b ), \quad -\b^2 \del_\b F_q(\b)=\del_T F_q(1/T)\leq 0
\ee
The first is due to $-\log x$ being a convex function, and the second is the positivity of entropy. If the ensemble average gives a continuous spectrum in IR, the annealed free energy $F_a(1/T)$ will not be a monotonically decreasing function for small enough $T$. Therefore, the thermodynamically critical temperature is bounded from below
\be 
T_{t.c.}\geq T_{l.b.},\quad \del_T F_a(1/T)|_{T=T_{l.b.}}=0
\ee

Let us consider two relevant simple examples. The first example is for unbounded Gaussian spectral density \eqref{55-rho} of commuting SYK model with Hamiltonian $H_a/\sqrt{d}$. Note the full spectrum density includes the $2^N$ factor to \eqref{55-rho}. The annealed free energy is given by
\be 
F_a(\b)=-T\log 2-\f 1 {2dN T}
\ee
which leads to
\be 
T_{l.b.}=\f 1 {\sqrt{2d N\log 2 }} \label{210-eq}
\ee
The second example is for the square-root edge spectrum $\rho(E)\app \f {2}{\sqrt{\pi}}e^{N s_0}\sqrt E$ where $s_0$ is of $O(1)$. This leads to partition function $Z(\b)=e^{N s_0}/\b^{3/2}$. For this partition function, it is easy to find
\begin{align} 
F_a(\b)=-T s_0-\f {3T}{2N}  \log T \label{211-eq}\\
\implies T_{l.b.}=e^{-2N s_0/3-1} \label{212-eq}
\end{align}
From these two examples, we see clearly that the compact spectrum potentially has a much lower critical temperature than the Gaussian non-compact spectrum. 

For the $q=0$ dcSYK model, the tail shape \eqref{87-eq} is essentially the same as the Gaussian case \eqref{55-rho} times a factor of $d e^{d+1}$. Regarding the annealed free energy, if $T>T_c$ we use \eqref{211-eq} as an approximation, and if $T<T_c$ we can just consider the Laplace transformation of \eqref{87-eq} for $y<0$ as an estimation for $Z(\b)$. This is reasonable because the lower bound critical temperature \eqref{212-eq} from the $T>T_c$ part is much lower than $T_c$. It turns out that the lower bound critical temperature from \eqref{87-eq} is similar to \eqref{210-eq}
\be 
T_{l.b.}\app \f 1 {\sqrt{2d N (\log 2+\f {\log d +d+1}{N}) }}
\ee
Note that this lower bound of critical temperature is much lower than $T_c=1/(2d)$ since $d\ll N$. If we assume the true thermodynamical critical temperature $T_{t.c.}$ is of the same order of $T_{l.b.}$ \footnote{The exact gap between $T_{l.b.}$  and $T_{t.c.}$ is model dependent. For example, the finite $p$-spin SK model \cite{gardner1985spin} has the Gaussian spectra and its critical temperature depends on $p$, which in our notation are in the same order as $T_{l.b.}$ \eqref{210-eq}. In the $p\ra\infty$ limit (noting it is within the regime of our double-scaled limit), the $p$-spin SK model becomes the random energy model \cite{derrida1980random,derrida1981random, mezard2009information} and the critical temperature becomes exactly \eqref{210-eq}. Since the tail \eqref{87-eq} is essentially the same as the Gaussian tail, this assumption seems reasonable.}, there is a wide temperature window between $T_{t.c.}$ and $T_c$, in which the self-averaging still holds though the dynamics is non-chaotic and perhaps similar to the commuting SYK model like \cite{Gao:2023gta}.

For the $q\ra 1$ and $d<9$ case, the scenario is different. Starting with the 1-loop partition function \eqref{eq:103-3} with additional normalization $e^{N s_0}$ with $s_0\sim O(1)$, the annealed free energy can be computed directly
\be 
F_a(\b)=-Ts_0+\f{E_0}{N}-\f{c T^2}{N}+\f{T\log\f{(T/T_c-1)^{d-1}}{T^{d-4}}}{2N}
\ee
Noting that $c\sim \lam^{3/2}$ and $T_c\sim \lam^{1/2}$, we find
\be 
T_{l.b.}\app T_c \left(1+\f{d-1}{2N s_0}\right)
\ee
which is indeed very close to $T_c$. If we again assume the thermodynamical critical temperature $T_{t.c.}$ is of the same order as $T_{l.b.}$, it is also of the same order as $T_c$. Therefore,  $T_c$ has a thermodynamical meaning that the self-averaging holds only for $T\gtrsim T_c$, and there is a new phase in $T<T_c$. Note that in the strict $\lam\ra 0$ limit, we will have Gaussian spectral density, whose $T_{t.c.}$ is roughly \eqref{210-eq} and scales as $1/\sqrt{N}$. This is compatible with $T_c\sim\lam^{1/2}\sim O(1/\sqrt{N})$.

\section{Numeric result} \label{sec:num}

\subsection{The spectrum}

To check our analytic results, we compare them with results from exact diagonalization. In Fig. ~\ref{allfit} we show the comparison between the numerical result of spectrum of the $N=15,~p=2$, $d=2,5$ model as well as the theoretical prediction of the double-scaled SYK model with parameter $\bar{q}$. We find good agreement between the bulk spectrum of analytic and numerical results. However, if we zoom in to the edge of the spectrum, we find a tail exceeding the theoretical prediction. 

\begin{figure}[ht]
    \centering
    \begin{subfigure}{0.45\linewidth}
        \centering
        \includegraphics[height=2.6cm]{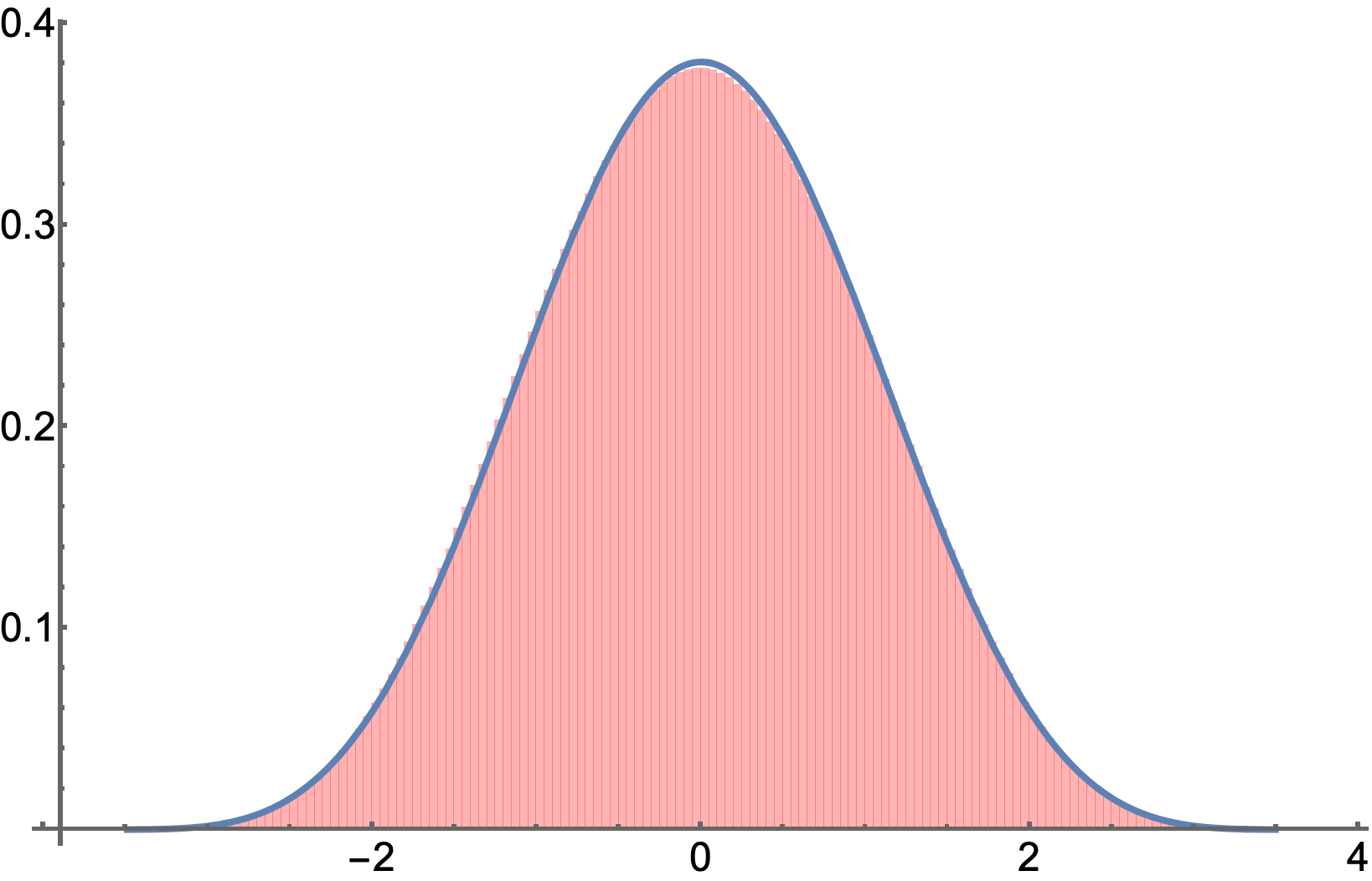}
        \caption{$N=15$ $d=2$ }
        \label{fig:fullN15d5}
    \end{subfigure}
~
    \begin{subfigure}{0.45\linewidth}
        \centering
        \includegraphics[height=2.6cm]{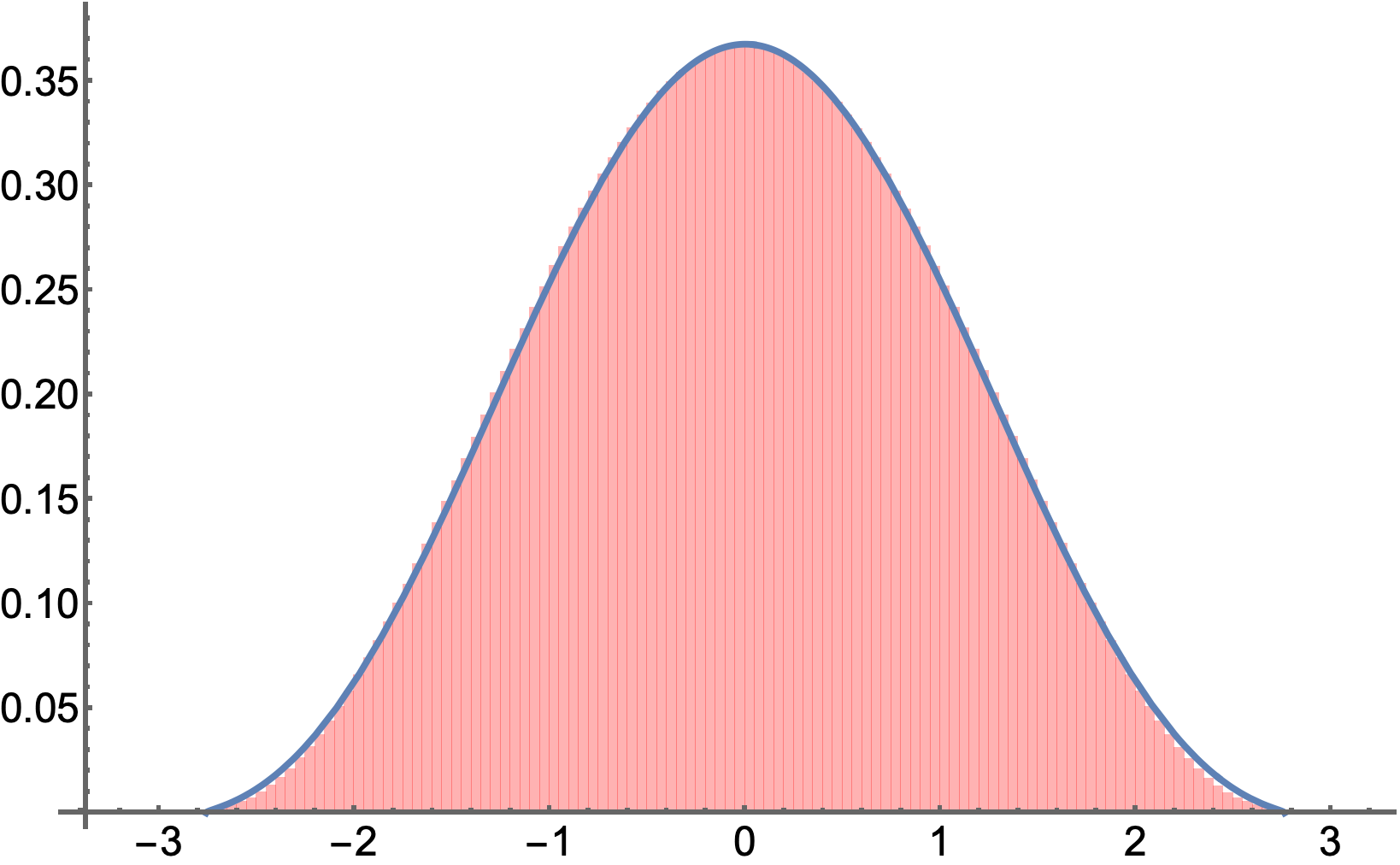}
        \caption{$N=15$ $d=5$ }
        \label{fig:fullN13d9}
    \end{subfigure}

    \caption{\justifying A fit of the whole spectra of models with $N=15$, $p=2$ and various $d$.}
    \label{allfit}
\end{figure}

In the previous section,  we provide an analytical expression describing the behavior of the tail. Here we can also compare the tail part of the spectrum from the numerical computation with that predicted by the analytic results. We list in Table \ref{tab:1} the theoretic and numeric values of the critical temperature $T_c$, which can be read from the reciprocal of the exponent of the exponential tail in~\eqref{eq:122}. 

Let us emphasize that our numerics has $\lam=16/15 \sim O(1)$, it is inappropriate to use the analytic expression of the ``ground" state energy $E_0$~\eqref{eq:104} because the loop expansion in $\lambda$ is not appropriate.  Instead, we will simply use the theoretic value of the double-scaled SYK model with $q \ra \bar q$ replacement, which leads to 
\be 
\tilde{E_0}=-2/\sqrt{(1-1/d)(1-e^{-\lam})}
\ee
We can compare this with the numerical results. Examples of the fitting to the tail part is shown in Fig.~\ref{fullspec15}. We also list the comparison between the theoretical value and the numerics in Table \ref{tab:2}.  We observe that for small $d$ the agreement is quite well, as $d$ increases the tail region becomes shorter so the fitting is less reliable.

In this computation, we sampled over 5000 samples for the $N=15$ model. In this case, $\lam>T_+$ for $d\geq2$ and the theoretical prediction of the spectrum \eqref{201-rho} applies. To extract the parameter $T_c$ and $E_0$, we fit the tail part of the numerical result by \eqref{201-rho}. A technical detail is that we should select the appropriate range of data to fit; we should not include the data points on the very end of the numerical spectrum since they are subject to large fluctuations and we also should not include the data that are deep into the bulk of the spectrum and do not belong to the tail.  In practice, since there is a prior no way to determine the best range that minimizes the error due to the above factors, we simply scan over and fit different ranges of data with the ansatz and choose the range that gives the result closest to the theoretical value. This gives the best-fitting result and should be considered as the most optimal bound of the fitting result. This is the most information about the tail we can extract from the numerical results. Further notice that in doing this, we carried out a few consistency checks, such as the fitting range should be below the theoretical predicted $\tilde{E_0}$. Another aspect to note is that we do not fit $\tilde E_0$ for $d=3$ in Table \ref{tab:2} because the theoretical prediction of the tail is a simple exponential decay and $\tilde E_0$ dependence is only through an overall coefficient.

Notice that strictly speaking, this analytic prediction works exactly in the $N\to \infty$ and $\lam\ra 0$ limit. However, as we can see explicitly that at $d=2,3$ the agreement between this prediction and the numerical results is already quite well even though $N=15$ and $\lam=16/15$ are not close to an ideal case. On the other hand, this matching becomes worse for larger $d$. This is precisely as expected since the current analysis of the tail is only at the 1-loop order and assuming $d\ll N$. There are two possible corrections that come to play a role as $d$ increases. The first is from higher loop orders, which modifies the pole of the partition function. As we argued in Section \ref{1-loop_d9} that higher loops are crucial to explain the tail of the spectrum for $d>9$, since our $\lam$ is not small enough, we should expect higher order corrections start to play a role at even smaller $d$'s. The second source of corrections come from  finite $N$ effects, possibly in orders of $d/N$, as our $d$ is not relatively much smaller than $N$. This finite $N$ correction is beyond the double scaling limit and requires new computational techniques to analyze it. Nevertheless, the numerical values for both $T_c$ and $|\tilde E_0|$ monotonically decrease as $d$ increases, which is qualitatively compatible with our expectation.

\begin{table}
    \centering
    \begin{tabular}{|c|c|c|}
    \hline
        $T_c/\lam^{1/2}$ & $d=2$ & $d=3$ 
        \\
    \hline
        Theory  & 0.1809 & 0.07717 
        \\
    \hline
        Numerics &  0.1836 &  0.07709 
        \\
    \hline
       Error &  1.5\% &  0.1\% 
       \\
    \hline   
    \end{tabular}
    \caption{The critical temperature $T_c$ for different $d$ and comparison of theoretical and numerical values.}
    \label{tab:1}
\end{table}

\begin{table}
    \centering
    \begin{tabular}{|c|c|c|}
    \hline
        $\tilde E_0$ & $d=2$ & $d=3$ 
        \\
    \hline
        Theory  & -3.4926 &  -3.0247 
        \\
    \hline
        Numerics & -3.5606 & NA 
        \\
    \hline
        Error & 1.9\% & NA 
        \\
    \hline
    \end{tabular}
    \caption{The ``ground" energy $\tilde E_0$ for different $d$. For $d=3$ the parameter $E_0$ does not appear in the  fitting ansatz~\eqref{eq:122}, so the above table does not have the data for $d=3$.}
    \label{tab:2}
\end{table}

\begin{figure}[ht]
    \centering
    \begin{subfigure}{0.45\linewidth}
        \centering
        \includegraphics[width=\linewidth]{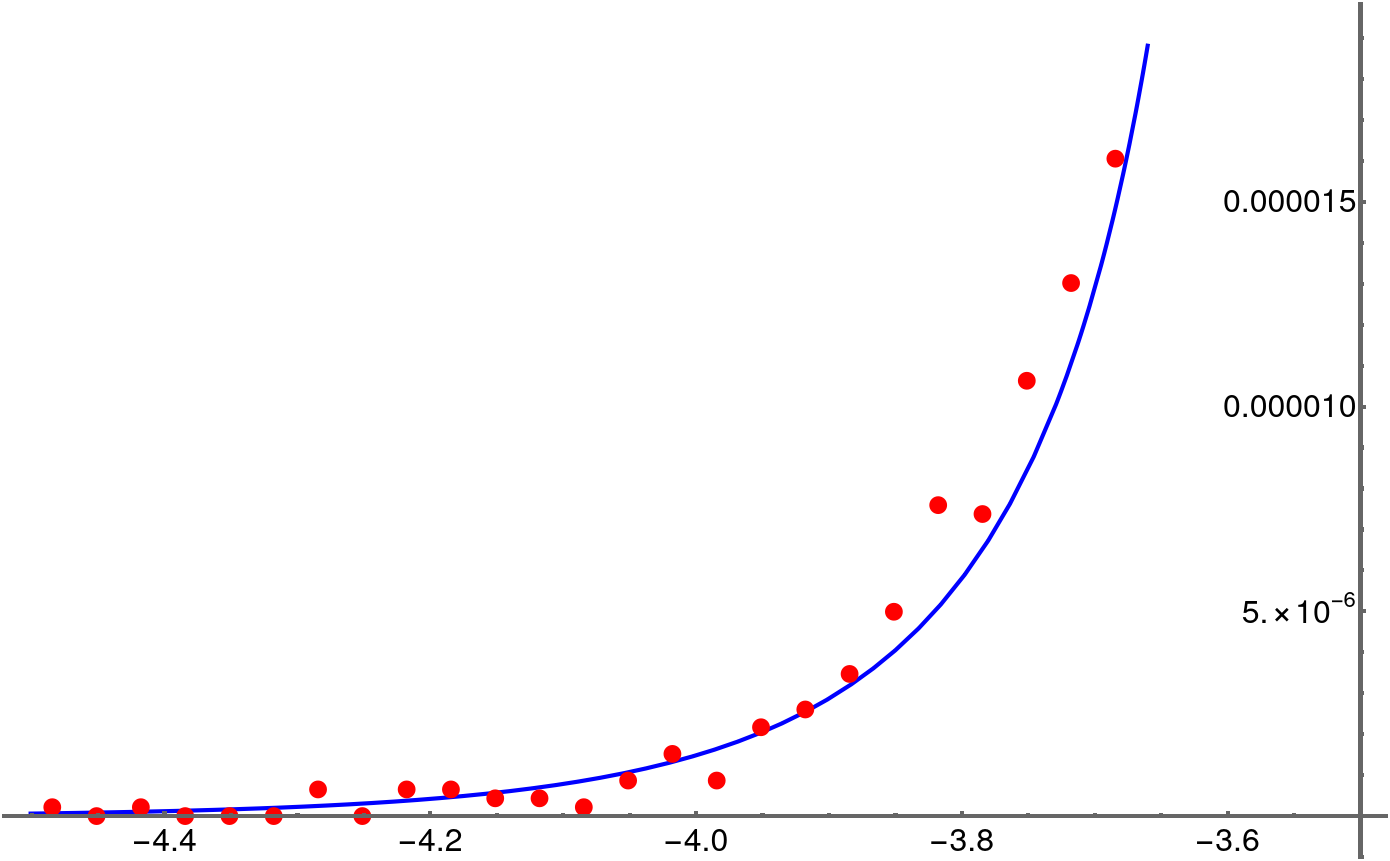}
        \caption{$d=2$}
        \label{fig:d2fit}
    \end{subfigure}
    \hfill
    \begin{subfigure}{0.45\linewidth}
        \centering
        \includegraphics[width=\linewidth]{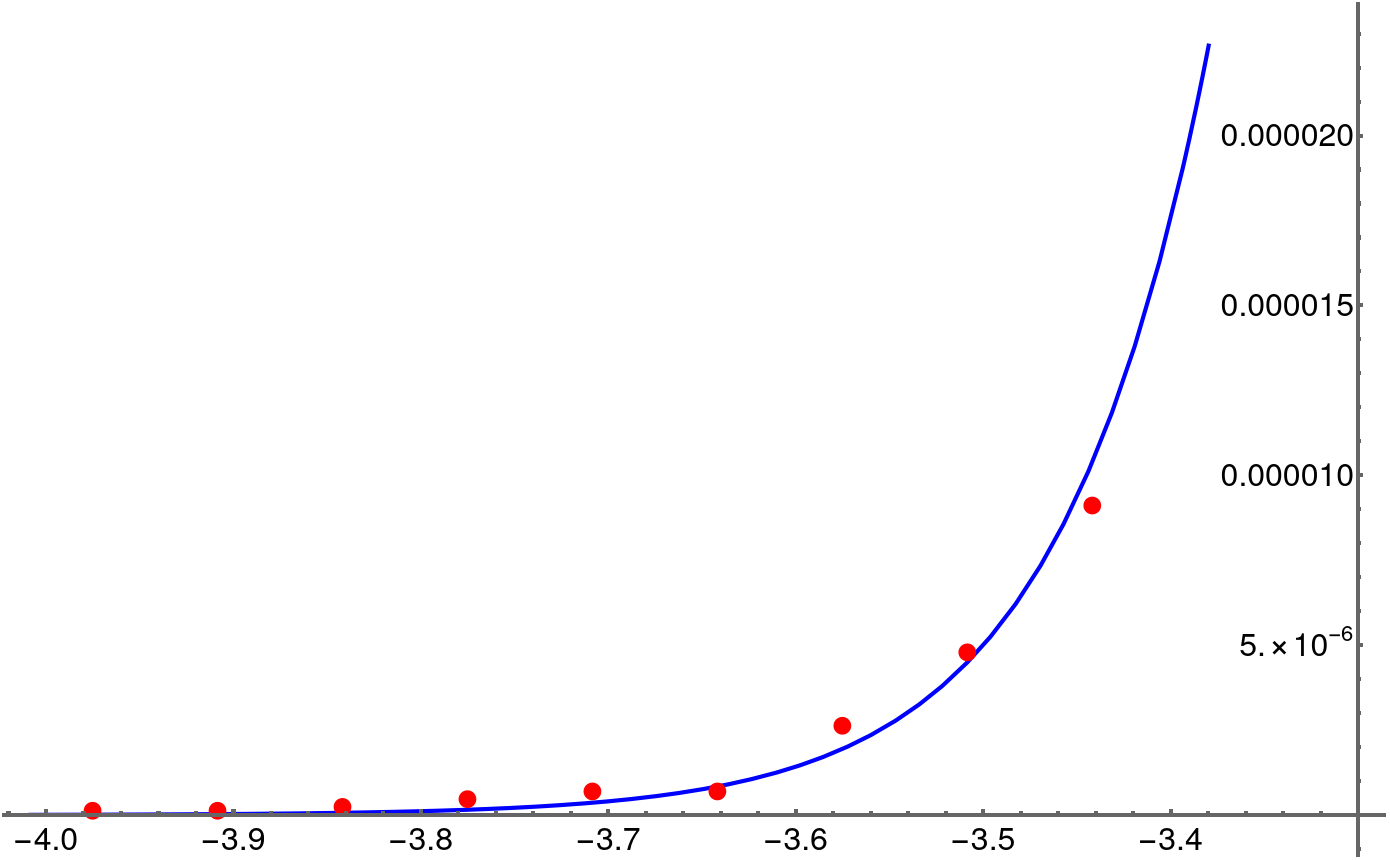}
        \caption{$d=3$}
        \label{fig:d3fit}
    \end{subfigure}
    \caption{\justifying A fit of the tails of the spectra of models with $d=2,3$.}
    \label{fullspec15}
\end{figure}

\subsection{The spectral form factor} \label{sec:SFF}

An independent diagnosis of quantum chaos additional to the shape of the spectrum is the spectral form factor (SFF) \cite{Cotler:2016fpe}
\be 
\text{SFF}(\b,t)=\E[Z(\b+it)Z(\b-it)]
\ee
In our dcSYK model, we can also compute the spectral form factors numerically for theories with various $d$. 

In Fig.~\ref{SFFs}, we demonstrate the SFF for the dcSYK models with $N=12$ for $d=1$ and $d=2$, as well as the regular SYK model with random couplings normalized according to~\eqref{5-syk}, which can be considered as the $d\to \infty$ limit of the dcSYK model. In each of the plots, the dashed lines are SFF for individual samples in the computation, and the solid curve is the average over all the samples. From the results, we observe that for $d=1$ the SFF does not contain a ramp and plateau even at infinite temperature $\b=0$. On the contrary, in Fig. \ref{SFF24d2b0} for $d=2$ we observe the clear appearance of ramps and plateaus in the averaged SFF at $\b=0$. As we decrease the temperature but still above the critical temperature $T_c$ (in the setup of Fig.~\ref{SFFs}, $T_c=0.187$, $\beta_c=5.35$ for $d=2$), the ramp becomes less steep comparing with the high temperature result. For example, as shown in Fig.~\ref{SFF24d2b2} the slope of the ramp at $\beta=2$ is $k=0.45$ which is smaller than the slope of the ramp at $\beta = 0$, which is $k=0.77$. As we decrease the temperature further below $T_c$ in Fig.~\ref{SFF24d2b10}, we observe big fluctuations in samples even at earlier time, this is qualitatively different from the regular SYK model as well as the dcSYK model at higher temperatures. Additionally, we do not observe increasing ramp in the SFF, in fact, if we fit the region from $10^2$ to $10^3$ in SFF, we observe a ``down" ramp with a negative slope $k=-0.10$. In comparison, the SFF of the regular SYK model~\eqref{4-syk} at the same low temperature $\b=10$ still exhibits clear ramp and plateau structure. 

Since SFF reflects the feature of spectral correlation or level spacing statistics in the ensemble of energy eigenstates around the scale of the temperature, the comparison between Fig.~\ref{SFF24d2b2} and \ref{SFF24d2b10}~indicates two distinct sectors of eigenstates with different statistics effectively engaging the SFF at different temperatures. The energy scale (above the ground state) that separates these two sectors should be the same order as $T_c$, which is compatible with our spectrum analysis. 

As we discussed in Section \ref{sec:5}, the critical temperature $T_c$ suggests the existence of a thermodynamical critical temperature $T_{t.c.}$ at the same order. Though our analysis in the main part of this paper is about the averaged spectrum, here the numerics of SFF give a further support for a phase transition below a temperature of the same order as $T_c$, and there should exists a significant difference of spectral correlations in the two phases. This is another strong evidence that the dcSYK model is quantum chaotic in the long time scale above this critical temperature. We will report further studies to better understand the low temperature phase in future.

\begin{figure*}[t]

    \begin{subfigure}[t]{0.195\linewidth}
        \centering

\includegraphics[width=\linewidth]{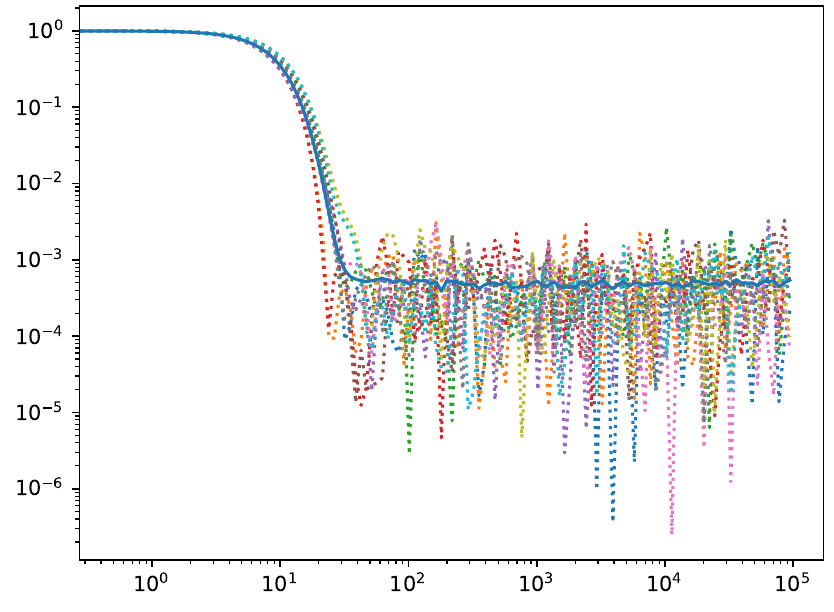}
        \caption{$d=1, \beta=0$, $h=4.9\times 10^{-4}$ }
        \label{fig:SFFd1N24}
    \label{SFFd1}
    \end{subfigure}
    \hfill
    \begin{subfigure}[t]{0.195\linewidth}
        \centering
     \includegraphics[width=\linewidth]{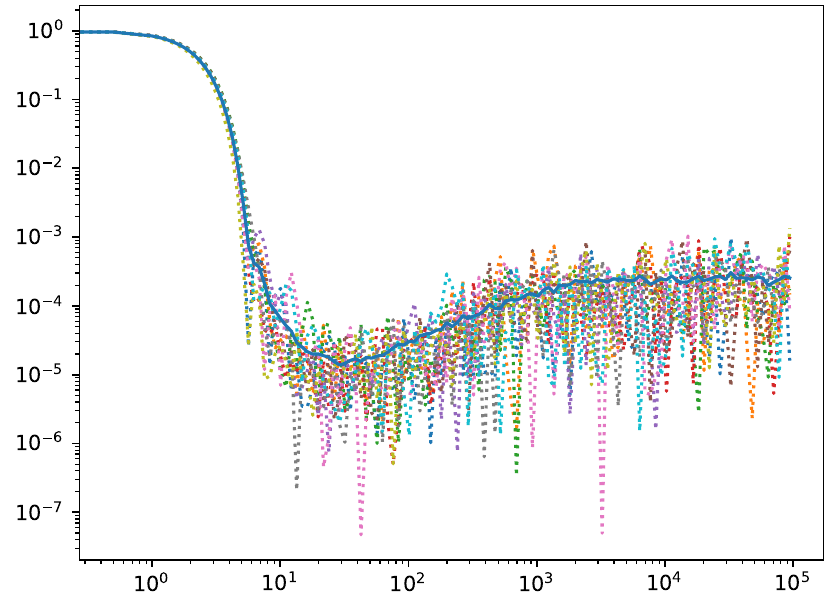}
        \caption{$d=2, \beta=0$, $k=0.74, h=2.4\times 10^{-4}$ }
        \label{SFF24d2b0}
    \end{subfigure}
    \hfill
    \begin{subfigure}[t]{0.195\linewidth}
        \centering
        \includegraphics[width=\linewidth]{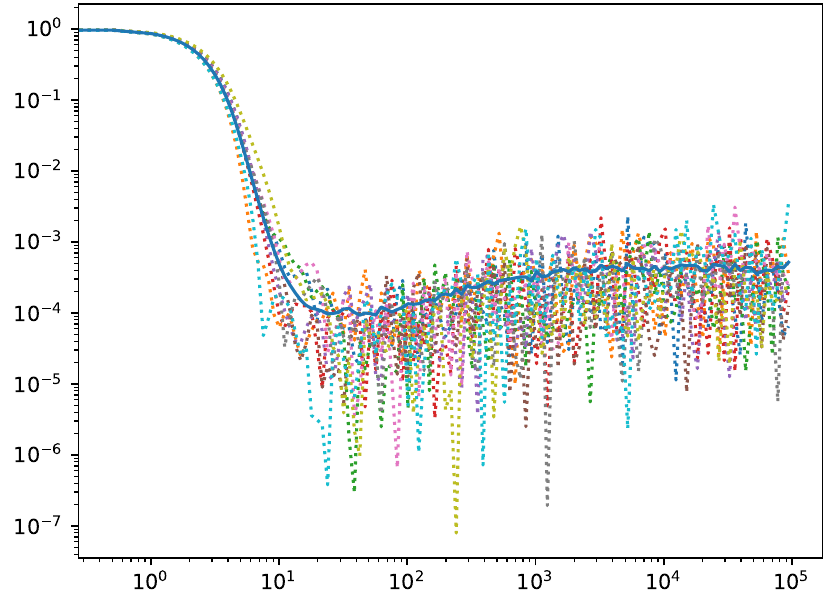}
        \caption{$d=2, \beta=2$, $k=0.45, h=4.3\times 10^{-4}$ }
        \label{SFF24d2b2}
    \end{subfigure}
    \hfill
    \begin{subfigure}[t]{0.195\linewidth}
        \centering
        \includegraphics[width=\linewidth]{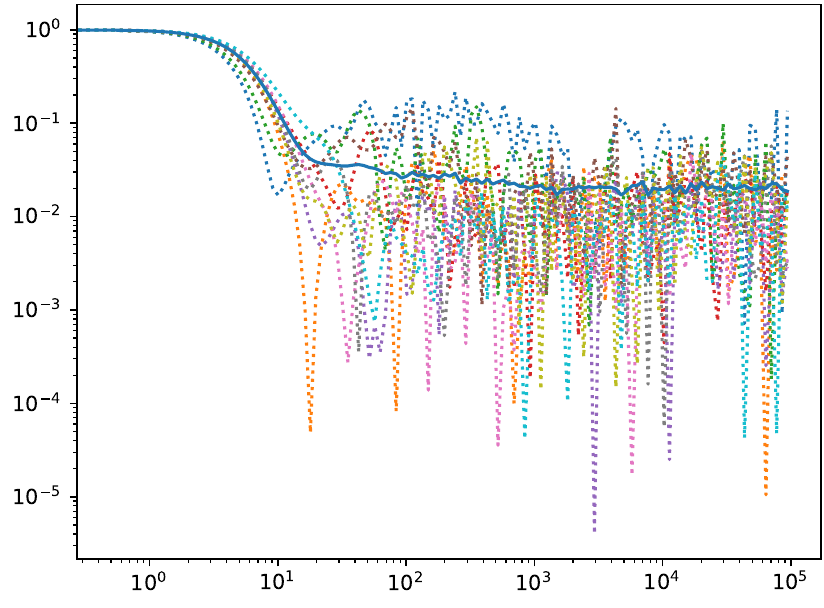}
        \caption{$d=2, \beta=10$, $h=2.0\times 10^{-2}$ }
        \label{SFF24d2b10}
    \end{subfigure}
    \hfill
    \begin{subfigure}[t]{0.195\linewidth}
        \centering
        \includegraphics[width=\linewidth]{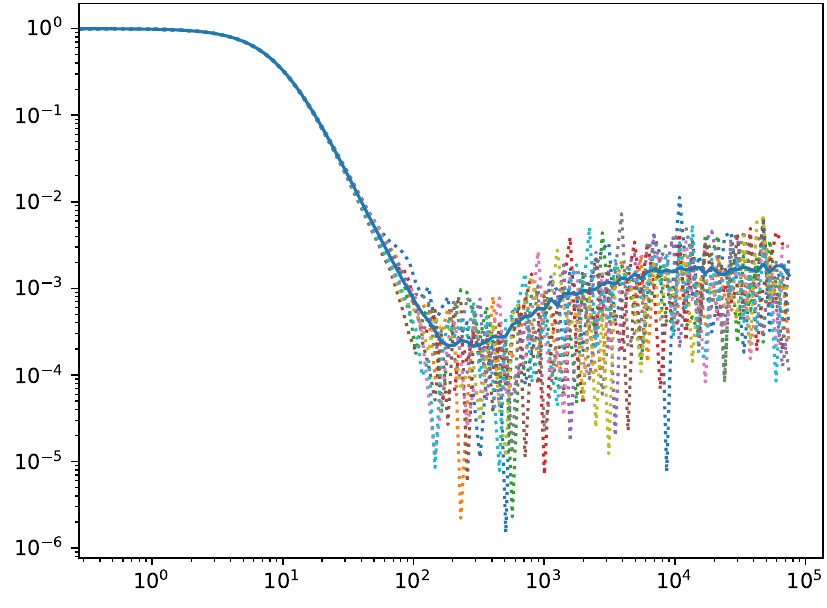}
        \caption{$d\to \infty$, $\beta=10$, $k=0.84, h=1.6\times 10^{-3}$ }
        \label{SFF24regb10}
    \end{subfigure}
    \caption{\justifying (a)-(d) The SFF of models $d=1$ and $d=2$ for different temperatures. (e) The SFF of regular SYK model \eqref{4-syk} at low temperature $\b=10$. In all cases we have $N=12$, and the solid blue curve is the averaged SFF. In the captions, $k$ is the slope of the ramp if it exists, and $h$ is the height of the plateau.}
    \label{SFFs}
\end{figure*}

\section{Conclusion and discussions} \label{sec:concl}

In this work, we have studied a variant of the $p$-local SYK model by combining $d$ copies of the integrable commuting SYK models, which we call the d-commuting SYK model. Though all terms in each commuting SYK Hamiltonian are commutative, each commuting SYK Hamiltonian does not commute with each other because they are constructed by different basis. Since the total Hamiltonian is the sum of these commuting SYK components, the integrability breaks down for $d>1$. 

We study the spectrum of this model in the double-scale limit with $\lam=4p^2/N$ fixed and $N\ra\infty$. The spectrum is non-compact for generic $d$ and becomes compact only when $d\ra\infty$. Regarding its shape, it is close to the compact form of $d\ra \infty$ above some energy scale but has an exponential suppressed tail below it. Therefore, this model can be regarded as an interpolation between integrability and chaos. For the canonical ensemble with temperature $T$, we can regard the system as chaotic for $T>T_c$ and non-chaotic for $T<T_c$.

We compute $T_c$ in two cases. The first case is $\lam\ra\infty$ where we can use free probability to find the exact form of the spectrum. The critical temperature $T_c=1/(2d)$ is read from the exponent of the exponential tail $e^{(E-E_0)/T_c}$ for $E<E_0$ for some $E_0$. The exponential tail is from a non-perturbative effect in $d$. The second case is $\lam\ra 0$ in the Schwarzian limit where we use the coarse-grained path integral to find the approximate spectrum in the IR. At the 1-loop level, we find the spectrum has an exponential tail for $d<9$. For $d>9$ we speculate the tail is due to higher loop effects. In both cases, we observe $T_c$ decreases as $d$ increases, which indicates a bigger chaotic regime. The thermodynamical meaning of the critical temperature $T_c$ is different for the two cases. For $\lam\ra\infty$, the thermodynamical critical temperature for a phase transition is much lower than $T_c$; for $\lam \ra 0$, it is of the same order as $T_c$. 

There have been other discussions about integrability to chaos transition in the context of SYK-like models. They generically fall into two classes. One class involves a direct coupling between an integrable model and a chaotic model, and the transition is a direct consequence of the competition between the two constituent models~\cite{Peng:2017kro}. The second class does not have such explicit coupling and the transition is purely due to the strong interaction within the single Hamiltonian~\cite{Peng:2018zap,Chang:2021wbx,Chang:2021fmd} and the transition appears at a special value of a parameter of the model. The theory we consider in this paper is another example of the second class, where the transition appears at $d=1$.

To justify the 1-loop computation of the coarse-grained path integral for $d<9$, we simulate the dcSYK model for $N=15$ and $p=2$, and draw the spectrum by direct diagonalization. Fitting the data, we find a nice match of both critical temperature $T_c$ and the ``ground" energy $\tilde E_0$ with the 1-loop computation for $d=2,3$. The fit becomes worse as $d$ gets larger due to both higher loop and the finite $d/N$ corrections. Numeric spectral form factor at different temperatures also demonstrate a strong evidence that the dcSYK model has a phase transition at some temperature $\sim O(T_c)$. The model is quantum chaotic only above this temperature in long time scales.

\paragraph*{Applications to simulation of holography} 

The existence of the critical temperature has a useful implication for future quantum simulation experiments for holographic systems (e.g. SYK model) under the slogan of ``quantum gravity in a lab" initiated by Susskind in \cite{Susskind:2017ney}. However, the regular SYK model has $O(N^p)$ numbers of terms and is extremely hard to implement on a quantum simulator. On the other hand, the sparse SYK model \cite{Xu:2020shn} containing randomly picked $O(N)$ terms is much simpler and argued to have a holographic description as well. Along this line, one of the most interesting simulation projects is the teleportation protocol of traversable wormholes \cite{Gao:2016bin,Gao:2019nyj,Schuster:2021uvg}. In a recent simulation \cite{Jafferis:2022crx}, an
attempt was made with a small $N=7$ SYK model with only five terms in the Hamiltonian that is learned to produce an arguable signal of traversable wormholes \cite{Kobrin:2023rzr,Jafferis:2023moh} on Google's Sycamore quantum chip. In this simulation, all terms in the learned Hamiltonian are commutative to each other and the model is integrable though it is slightly different from our construction \eqref{1-dcsyk} for $d=1$. The analysis in \cite{Gao:2023gta} shows that for small $N$ one cannot clearly distinguish the simulation signals from the traversable wormhole mechanism in the chaotic regular SYK model or the thermalization process in the integrable commuting SYK model (see also \cite{Lykken:2024ypy}). 

More generally, to simulate an authentic signal of traversable wormholes in a holographic system, besides increasing $N$, we need to break integrability by including adequate numbers of non-commuting terms and make sure the system is in a state exhibiting chaotic behavior. Note that any Hamiltonian in the simulation can be decomposed in the form of \eqref{1-dcsyk} where each $\tilde H_a$ only contains commuting terms, which may be different from our concrete construction in \eqref{1-dcsyk} though. For a decomposition with minimal number $d$ of groups, one can compute the averaged non-commutativity parameter $q=e^{-\lam}$ among groups, and our double-scaled analysis gives a typical estimate of the critical temperature $T_c$ in terms of $d$ and $\lam$. For temperatures above $T_c$ we can regard the system as having a convincing signal for traversable wormholes or other broader phenomena of holography. Roughly speaking, $T_c$ should play a role as the benchmark for future quantum simulations of quantum gravity. In this sense, we expect the d-commmuting SYK model to have a kind of universality of random models, which are parameterized by $\lam$ and $d$. It would be interesting to have a stringent test of this universality proposal in future works.

There are a few future directions. 
\begin{enumerate}
    \item Though we argue in Sec. \ref{sec:5} that 
there is a thermodynamical critical temperature for a phase transition for which self-averaging breaks down, it is interesting to study the property of the low-temperature phase. Note that the commuting SYK model is essentially $p$-spin SK model, which has 1-step replica symmetry breaking (RSB) at the thermodynamical critical temperature \cite{gross1984simplest}, and the low-temperature phase is a spin glass. It is perhaps natural to expect an RSB for the dcSYK model (probably unless $d\ra\infty$) if the temperature is low enough \footnote{See also a recent interesting variant of the SYK model without RSB \cite{Biggs:2023mfn}.}. This requires a careful study and generalization of the $n$-replica chord diagrams developed in \cite{Berkooz:2020fvm}.
   \item From the shape of the spectrum we speculate the chaotic regime above the critical temperature $T_c$. However, to justify this statement, we need to study the dynamics of the dcSYK model at different temperatures. Correlation functions can be represented in the language of chord diagrams and the coarse-grained path integral should also work when we include matters. It is interesting to see how the $h\neq 2$ modes contribute to the correlation functions at different temperatures in the $q\ra 1$ regime. In particular, if there is an RSB phase transition, the chord diagram techniques need to be updated to account for replicas, and it is worth exploring if quantum chaos exists in the RSB phase like \cite{Anous:2021eqj}.
   
   \item Recently, the sparse SYK model draws a lot of interest (e.g. \cite{Xu:2020shn, Garcia-Garcia:2020cdo,Orman:2024mpw}) because it reduces the difficulty of simulation by randomly choosing only $O(N)$ terms from the $O(N^p)$ terms in the regular SYK model. As explained in Sec. \ref{sec:2C}, the dcSYK model can be regarded as a sparse version of the regular SYK model. We can easily increase the sparseness of this model by sparsifying each commuting component $\tilde H_a$ to $O(N)$ terms such that the total number of terms in the Hamiltonian $\tilde H$ is of order $O(N)$. Assuming each term in $\tilde H_a$ appears with probability $p_0$, we need to rescale the variance of the coupling constant $J^a_{i_1\cdots i_p}$ by $1/p_0$. Following \cite{Garcia-Garcia:2020cdo}, it is easy to see that the spectrum is unchanged in the double-scale limit. However, as pointed in \cite{Orman:2024mpw}, when the sparseness is too high, the model has accidental degeneracy and quickly becomes non-chaotic. Since each commuting component is non-chaotic and increasing $d$ tends to quantum chaos, it is interesting to study in more detail how the critical sparseness depends on the interplay between $p_0$ and $d$.
   \item Another diagonose of quantum chaos is the out-of-time-order correlators (OTOC). Since our model demonstrates a chaotic-integrable transition, it is interesting to check if we can observe similar transition in OTOC. On the other hand, since there are examples of OTOC overestimating the chaotic behavior of special models with unstable fixed points~\cite{Xu:2019lhc,Rozenbaum:2019nwn,Hashimoto:2020xfr,Trunin:2023xmw,Trunin:2023rwm}, it is also interesting to test this in our model. 
\end{enumerate}

\begin{acknowledgments}
We thank Michael Winer, Anna Biggs, Roland Speicher, Robert de Mello Koch, Gregory Korchemsky for stimulating discussions. PG is supported by the US Department of Energy under grant DE-SC0010008. CP is supported by NSFC NO.~12175237, and NO.~12247103 the Fundamental Research Funds for the Central Universities, and funds from the Chinese Academy of Sciences. 
\end{acknowledgments}

\appendix

\section{Overlap of a generic pair of $X^a_I$ and $X^b_{I'}$} \label{app:1}

For each pair $a\neq b$, we can first redefine
$\psi_{i}$ appropriately such that the relative overlapping is equivalent
to $X_{I}^{1}$ and $X_{I'}^{2}$. To see this, consider $X_{I}^{a}$ and $X_{I'}^{a+b}$ with $b>0$.
We can first redefine $\psi_{2i}\ra\psi_{2i-2(a-1)}$ to shift $X_{I}^{a}\ra X_{I}^{1}$
and $X_{I}^{a+b}\ra X_{I}^{1+b}$. We first assume $N$ is coprime
with any number between $2$ and $d-1$. Then, we redefine $\psi_{2i-1}\ra\psi_{2bi-(2b-1)}$
and $\psi_{2i}\ra\psi_{2bi-(2b-2)}$. Since $b$ takes values between
$1$ and $d-1$, by our coprime assumption this redefinition is a
permutation with one cycle. Since random overlapping does not care
about the explicit numbering of Majorana fermions, this problem is
equivalent to the overlapping between $X_{I}^{1}$ and $X_{I'}^{2}$.

For a generic case where $N$ is not coprime with a number between
$2$ and $d-1$. The second redefinition does not lead to the overlapping
problem between $X_{I}^{1}$ and $X_{I'}^{2}$ but $b$ copies of
overlapping problem with $N\ra N/b$. To be precise, assume $N/b\in\Z$
and there exists $i,i'=1,\cdots,2N$ such that
\begin{align}
2bi-(2b-1)=2bi'-(2b-1)+2kN\nn\\
\implies i=i'+kN/b,\quad k\in\{0,\cdots,b-1\}
\end{align}
Therefore, after the second redefinition, the Majorana fermions in
each $\mX_{i+kN/b}^{1+b}$ for $i=1,\cdots,N/b$ form a cycle, and
there are $b$ cycles in total. Correspondingly $\mX_{i}^{1}$ can
be grouped into $b$ cycles and different cycles have no chance to
overlap with each other. 

In this case, we need to split $p$ into
$p_{1},\cdots,p_{b}$ with probability $P(\{p_{i}\})$, and for each
group of $2N/b$ Majorana fermions, we compute the overlap probability
for $s_{i}$ such that $\sum_{i}p_{i}=p$ and $\sum_{i}s_{i}=s$.
As we assume $d$ is much smaller than $N$, in the double scale limit
the distribution of $p_{i}$ is highly peaked at $p_{i}=p/b$. Define
$\lambda_{i}=p_{i}^{2}/(N/b)$, and we have
\begin{align}
P(s)=&\sum_{p_{i}s_{i}}P(\{p_{i}\})\prod_{i}\f{(2\lambda_{i})^{s_{i}}}{s_{i}!}e^{-2\lambda_{i}}\nn\\
=&\sum_{p_{i}}P(\{p_{i}\})\f{(2\sum_{i}\lambda_{i})^{s}}{s!}e^{-2\sum_{i}\lambda_{i}}\nn\\
\app&\f{(2\lambda_{0})^{s}}{s!}e^{-2\lambda_{0}}
\end{align}
where in the last step we used large $N$ approximation $\sum_{p_{i}}P(\{p_{i}\})\ra\int dp_{i}\d(p_{i}-p/b)$.
This shows that the overlapping probability is universal in large
$N$ limit. Nevertheless, if $d$ is very large and carefully chosen such that $b$ scales with $N$ or we define $\mX_{i}^{a}$ alternatively with order $N$ numbers of relative cycles between two different $a$'s,
the delta function approximation for $P(\{p_{i}\})$ does not hold
and we have to consider the correction to this probability distribution.

\section{Numeric study of the 1-loop effect for $d>9$}
\label{app:2}

Let us first observe that there is a nontrivial subleading shift of
ground energy if we just include the first eigenvalue in $K_\pm$. By the naive large $d$ approximation of
the chord diagram, the change is a simple correction to $q\ra \bar q$, which leads to the ground energy \eqref{eq:106-1}. Comparing with (\ref{eq:104}), we see that the subleading corrections
are different by 
\begin{equation}
E_{0}-\tilde{E}_{0}=-\f 2{\sqrt{\lambda(1-1/d)}}\left(\f{\lambda}4+O(\lam/d)\right)\label{eq:107}
\end{equation}
The first term does not vanish in large $d$ limit, and it is from
the 1-loop correction of $h\neq2$ piece. Indeed, this discrepancy comes from the 1-loop correction of $h\neq 2$ modes. If we only had the $h=2$ piece, the ground energy was given by \eqref{156-eq}, which matches with (\ref{eq:106-1}) in large $d$ limit. This is
reasonable because in the double-scaled SYK we only have a $h=2$ piece. 

\begin{figure}
\begin{centering}
\includegraphics[height=4cm]{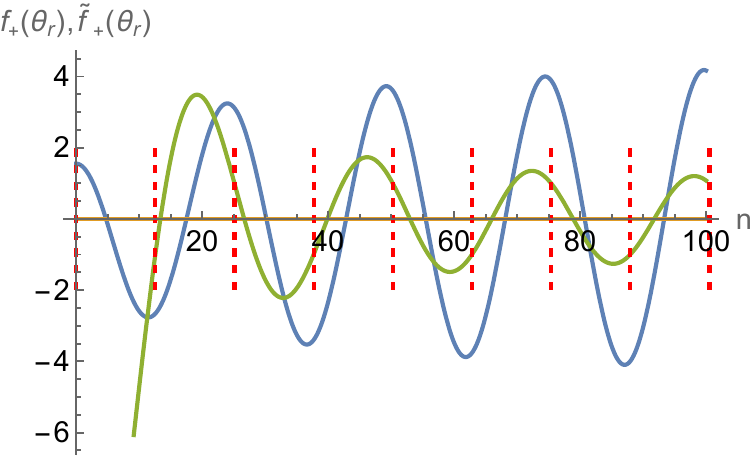}
\includegraphics[height=4cm]{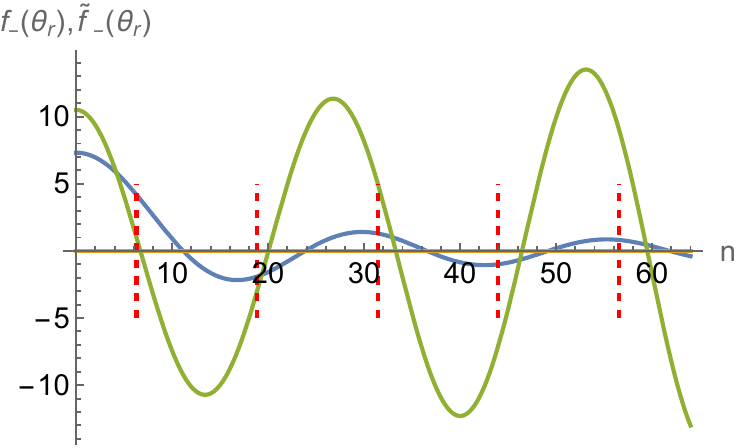}
\par\end{centering}
\caption{\justifying For low temperature and $h<1/2$, more and more eigenvalues in $\tilde{\protect\mM}_{\pm}$
deviates from $\protect\mM_{\pm}$. In this plot, we choose $h=0.3$
and $\protect\t_{r}=0.01$; and the blue/yellow curve is the real/imaginary
part of the exact eigenfunction $f_{\pm}(\protect\t_{r})$ and the
green curve is the approximate eigenfunction $\tilde{f}_{\pm}(\protect\t_{r})$.
\protect\label{fig:For-low-temperature}}
\end{figure}

Since this discrepancy does not vanish in the large $d$ limit, we must include more exact eigenvalues in $K_\pm$ to fix it. To proceed, let us first analyze the low temperature behavior of the eigenvalues for $h<1/2$. As we drop the temperature, more and more approximate real
eigenvalues in $\tilde{\mM}_{\pm}$ with small magnitude deviates
from the exact eigenvalues in $\mM_{\pm}$ as we can see from Fig.
\ref{fig:For-low-temperature} and compare it with Fig. \ref{fig:1b}
and \ref{fig:2a}. Following our strategy of Sommerfeld-Watson resummation,
we need to choose larger sets $K_{\pm}$ and $\tilde K_{\pm}$ as we decrease
the temperature. 

For low temperatures, the first few exact eigenvalues can be solved
by (\ref{eq:115}). For $\mM_{+}$, the first positive one is solved
by (\ref{eq:117}) but others in $\mM_{\pm}$ are similarly given
by \footnote{For $d=9$ and $h=1/2$, the approximate solution is $p/x=2n-1-2/(c_{n}+\log y)$
with $c_{n}$ a number depending on $n$ for $n\in\Z_{\mp}^{>0}$.
The estimate (\ref{eq:111}) and the numeric computation below still
hold.}
\begin{equation}
p/x=2(n-h)-\f{2\G(h+1/2)(y/4)^{1-2h}}{\G(3/2-h)\G(2h-n)\G(n)\sin h\pi}\label{eq:110-1}
\end{equation}
where $n\in\Z_{\mp}^{>0}$ for $\mM_\pm$ respectively. These exact eigenvalues, if included in $K_{\pm}$, each gives a finite
contribution to the determinant 
\begin{equation}
I_{\pm}\sim\log\cos^{2}(\pi h/2)+O(y^{1-2h})\label{eq:111}
\end{equation}
On the other hand, the approximate eigenvalues are more singular as
we can see from Fig. \ref{fig:For-low-temperature} that they are
quite close to the red dashed lines. These approximate eigenvalues
can also be solved perturbatively in low temperature as
\begin{equation}
p=\begin{cases}
4\pi n-8ny/\g & \tilde{\mM}_{+}\\
4\pi n-2\pi-4(2n-1)y/\g & \tilde{\mM}_{-}
\end{cases}\label{eq:112}
\end{equation}
where $n\in\Z^{>0}$ in both cases. By \eqref{eq:109} their contributions to determinant are
\begin{align}
I_{+}\sim&-\log\sin^{2}(2ny/\g)\\
I_{-}\sim&-\log\sin^{2}((2n-1)y/\g)\label{eq:113}
\end{align}
At leading order of $y$, they both scale as $-2\log y$. However,
this approximation breaks down quickly because as $n$ increases $ny/\g$
is no longer a small number. 

Therefore, we need to consider the accumulated effect of these eigenvalues,
which may give correction to the $\log y$ term. From Fig. \ref{fig:For-low-temperature}
we see that in each interval separated by $4\pi\Z$ or $4\pi\Z+2\pi$
for $f_{+}$ or $f_{-}$, there is only one exact eigenvalue and one
approximate eigenvalue. For these pair of eigenvalues, the approximate
one is always smaller than the exact one and tends to the latter as
$n$ goes large. Therefore, we can consider the approximation (\ref{eq:112})
works until $p$ hits the exact value (\ref{eq:110-1}), which gives
a natural estimate of the size of $\tilde K_{\pm}$ and $K_{\pm}$. Within
this regime, the sum of (\ref{eq:113}) and (\ref{eq:111}) is 
\begin{equation}
I_{+}+I_{-}\sim\sum_{n}\log\f{\sin^{2}(ny/|\g|)}{\cos^{2}(\pi h/2)}
\end{equation}
in which we should include $(\pi(1-h)/2)/(y/|\g|)\sim O(\b_{r})$
terms. This indicates a linear in $\b_{r}$ term from the accumulation
of eigenvalues, which gives a correction to the ground energy. 

\begin{figure*}
\begin{centering}
\subfloat[\label{fig:6a}]{\begin{centering}
\includegraphics[height=3.4cm]{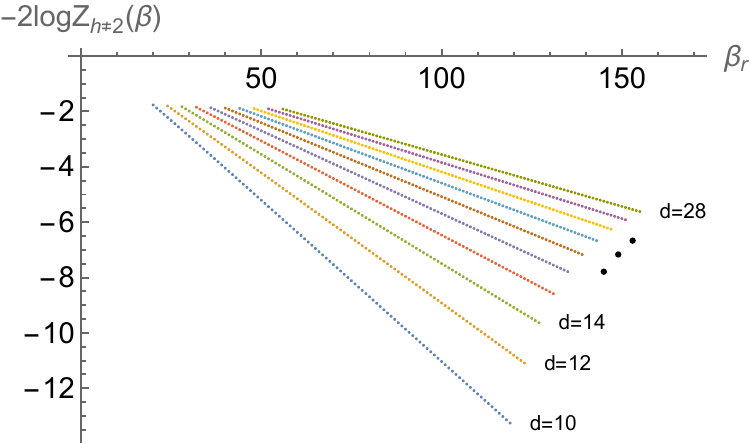}
\par\end{centering}
}\subfloat[\label{fig:6b}]{\begin{centering}
\includegraphics[height=3.4cm]{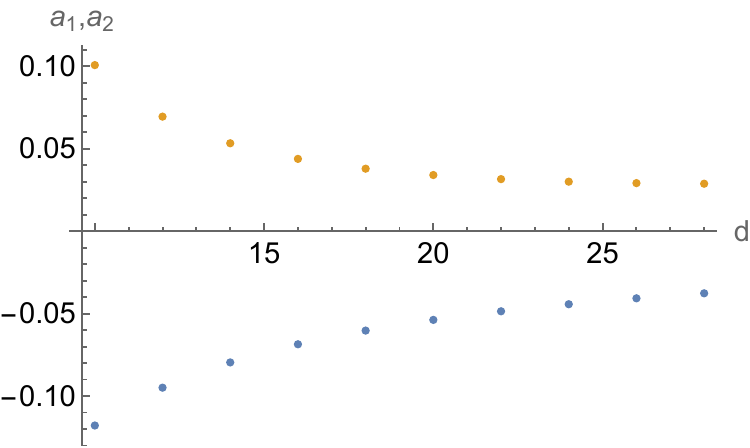}
\par\end{centering}
}\subfloat[\label{fig:6c}]{\begin{centering}
\includegraphics[height=3.4cm]{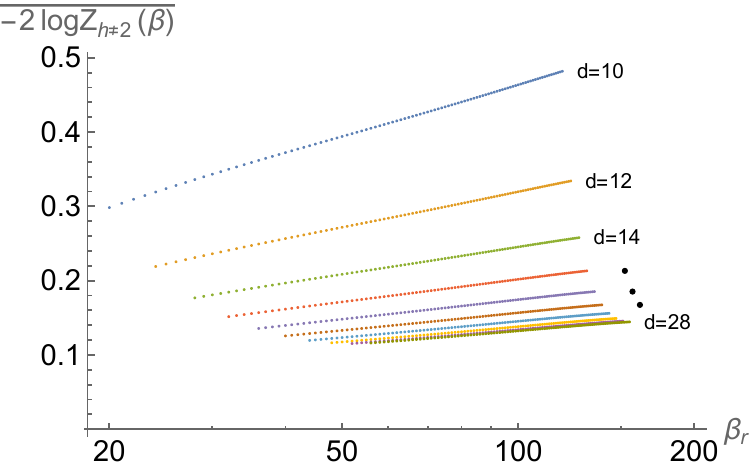}
\par\end{centering}
}
\par\end{centering}
\caption{\justifying (a) Numeric plot of $-2\log Z_{h\protect\neq2}(\protect\b)$ for $d=10$
to $d=28$. (b) The fit parameters $a_{1}$ (blue) and $a_{2}$ (yellow)
for $d=10$ to $d=28$. (c) Numeric plot of $\overline{-2\log Z_{h\protect\neq2}(\protect\b)}$,
which has constant and linear fit parts subtracted, for $d=10$ to
$d=28$. Note that the horizontal axis is in log scale.}
\end{figure*}

Numeric computation justifies our estimation. We can also consider
a large size of $K_{\pm}$ and $\tilde K_{\pm}$ and plot the 1-loop partition
function for the $h\neq2$ modes as a function of temperature. For
all terms discussed before, we have
\begin{align}
&-2\log Z_{h\neq2}(\b)=  \log\left|1-\f{x\tan(x/2)}{2(d-1)}\right|^{d-1}+x\tan\f x2 \nn\\
&+\f{d-1}2\left(\sum_{p\in K_{+}}\log|\sin^{2}\f p4|+\sum_{p\in K_{-}}\log|\cos^{2}\f p4|\right)\nonumber \\
 & -\f{d-1}2\left(\sum_{p\in K_{+}'}\log|\sin^{2}\f p4|+\sum_{p\in K_{-}'}\log|\cos^{2}\f p4|\right)\label{eq:115-2}
\end{align}
The result is shown in Fig. \ref{fig:6a}, in which we choose the
size of the $K_{\pm}$ to include the first 100 pairs of points. The
highest temperature (the smallest $\b_{r}$) is chosen to right above
$T_{0}$, which is a monotonically decreasing function of $d$, and
the lowest temperature (the largest $\b_{r}$) is chosen such that
the difference between the largest eigenvalues of $K_{\pm}$ and $\tilde K_{\pm}$
is about 0.001 \footnote{In principle, choosing larger $K_{\pm}$ and $K_{\pm}'$ allows us
to probe lower temperatures. However, Mathematica will loose accuracy
for the associated Legendre functions for too large $n$ and makes
the location of eigenvalues less reliable. }.

For different $d$, we see that their contributions (\ref{eq:115-2})
has a clear linear behavior with negative slope, whose absolute value is
indeed much smaller than $1/\sqrt{1-1/d}$, the naive expectation from \eqref{eq:107}.
The negative slope means that the ground energy is shifted to the
left. To find the subleading behavior, we fit the data using $a_{0}+a_{1}\b_{r}+a_{2}\log\b_{r}$
and find that the subleading $\log\b_{r}$ works quite well in low
temperature limit (see Fig. \ref{fig:6c}). Given this ansatz, the
low temperature partition function including all tree and 1-loop corrections
is
\begin{equation}
Z(\b)\sim\f 1{\b^{3/2+a_{2}/2}}e^{-\b(\tilde{E}_{0}+\f 12a_{1}\sqrt{\lambda})+c/\b}
\end{equation}


%

\end{document}